\documentclass[12pt,letterpaper]{article}
\pdfoutput=1

\usepackage[dvipsnames]{xcolor}

\usepackage{amsmath,amssymb,calc, amsthm,bbm, epsfig,psfrag, mathtools,comment}

\newtheorem{lemma}{Lemma}
\newtheorem{corollary}{Corollary}

\usepackage{tocloft}
\usepackage{quiver}
\usepackage{mathrsfs}
\usepackage{physics}
\usepackage{graphicx, enumerate}
\usepackage{float}

\usepackage[numbers,sort&compress]{natbib}

\usepackage{tikz}
\usepackage{tikz-cd}
\usepackage{enumitem}
\usepackage{tensor}

\makeatletter
\newcommand*\bigcdot{\mathpalette\bigcdot@{.5}}
\newcommand*\bigcdot@[2]{\mathbin{\vcenter{\hbox{\scalebox{#2}{$\m@th#1\bullet$}}}}}
\makeatother

\newcommand{\deq}{\stackrel{\bigcdot}{=}}

\usepackage[utf8]{inputenc} 
\definecolor{azure}{rgb}{0.0, 0.5, 1.0}
\definecolor{darkblue}{rgb}{0.15,0.35,0.7}
\definecolor{reddish}{rgb}{0.65, 0.2, 0.2}
\definecolor{brandeisblue}{rgb}{0.0, 0.44, 1.0}
\definecolor{ceruleanblue}{rgb}{0.16, 0.32, 0.75}
\definecolor{indigo(dye)}{rgb}{0.0, 0.25, 0.42}
\usepackage[linktocpage=true]{hyperref}
\hypersetup{
colorlinks=true,
citecolor=ceruleanblue,
linkcolor=ceruleanblue,
urlcolor=ceruleanblue,
pdfauthor={},
pdftitle={},
pdfsubject={}
}

\newcommand{\overbar}[1]{\mkern 1.5mu\overline{\mkern-1.5mu#1\mkern-1.5mu}\mkern 1.5mu}

\newcommand{\TT}{T\overbar{T}}

\newcommand{\ul}{\underline}

\newcommand{\fL}{\mathfrak{L}}
\newcommand{\fJ}{\mathfrak{J}}

\DeclareFontEncoding{LS1}{}{}
\DeclareFontSubstitution{LS1}{stix}{m}{n}
\DeclareSymbolFont{stixsymbols}{LS1}{stixscr}{m}{n}
\SetSymbolFont{stixsymbols}{bold}{LS1}{stixscr}{b}{n}
\DeclareMathSymbol{\kay}{\mathalpha}{stixsymbols}{"6B}
\DeclareMathSymbol{\hay}{\mathalpha}{stixsymbols}{"68}

\newcommand\CC[1]{\mathfrak{#1}}

\makeatletter
\renewcommand\section{\@startsection {section}{1}{\z@}%
                               {-3.5ex \@plus -1ex \@minus -.2ex}
                               {2.3ex \@plus.2ex}%
                               {\normalfont\large\bfseries}}
\renewcommand\subsection{\@startsection{subsection}{2}{\z@}%
                                 {-3.25ex\@plus -1ex \@minus -.2ex}%
                                 {1.5ex \@plus .2ex}%
                                 {\normalfont\bfseries}}
\makeatother

\let\non\nonumber

\newfont{\goth}{ygoth.tfm scaled 1200}                   

\numberwithin{equation}{section}

\begin{document}
\begin{titlepage}
\begin{flushright}
\today
\end{flushright}
\vspace{5mm}

\begin{center}
{\Large \bf 
Integrable Higher-Spin Deformations of Sigma Models from Auxiliary Fields}
\end{center}

\begin{center}

{\bf
Daniele Bielli${}^{a}$,
Christian Ferko${}^{b}$,
Liam Smith${}^{c}$,\\
Gabriele Tartaglino-Mazzucchelli${}^{c}$
} \\
\vspace{5mm}

\footnotesize{
${}^{a}$
{\it 
High Energy Physics Research Unit, Faculty of Science \\ 
Chulalongkorn University, Bangkok 10330, Thailand
}
 \\~\\
${}^{b}$
{\it 
Center for Quantum Mathematics and Physics (QMAP), 
\\ Department of Physics \& Astronomy,  University of California, Davis, CA 95616, USA
}
 \\~\\
${}^{c}$
{\it 
School of Mathematics and Physics, University of Queensland,
\\
 St Lucia, Brisbane, Queensland 4072, Australia}
}
\vspace{2mm}
~\\
\texttt{d.bielli4@gmail.com,
caferko@ucdavis.edu,
liam.smith1@uq.net.au,
g.tartaglino-mazzucchelli@uq.edu.au
}\\
\vspace{2mm}

\end{center}

\begin{abstract}
\baselineskip=14pt

We construct a new infinite family of integrable deformations of the principal chiral model (PCM) parameterized by an interaction function of several variables, which extends the formalism of \cite{Ferko:2024ali}, and includes deformations of the PCM by functions of both the stress tensor and higher-spin conserved currents. We show in detail that every model in this class admits a Lax representation for its equations of motion, and that the Poisson bracket of the Lax connection takes the Maillet form, establishing the existence of an infinite set of Poisson-commuting conserved charges. We argue that the non-Abelian T-dual of any model in this family is classically integrable, and that T-duality ``commutes'' with a general deformation in this class, in a sense which we make precise. Finally, we demonstrate that these higher-spin auxiliary field deformations can be extended to accommodate the addition of a Wess-Zumino term, and we exhibit the Lax connection in this case.

\end{abstract}
\vspace{5mm}

\vfill
\end{titlepage}

\tableofcontents

\section{Introduction}

Conserved quantities have long played an important role in the study of quantum field theory, both because these quantities have special properties and because their existence constrains the dynamics of a theory in useful ways. The most familiar examples are vector conserved currents $j_\mu$, which arise from conventional global symmetries via the Noether procedure and give rise to local operators which satisfy constraints such as Ward identities. However, there are also many well-known examples of currents $j_{\mu_1 \ldots \mu_s}$ which transform in higher-spin representations of the Lorentz group; perhaps the most ubiquitous is the energy-momentum tensor $T_{\mu \nu}$, which exists in any translationally-invariant field theory.\footnote{Another class of examples includes totally antisymmetric currents $j_{\mu_1 \ldots \mu_s}$ in $d$-dimensional QFTs, where $1 < s < d$, which are associated with higher-form global symmetries \cite{Gaiotto:2014kfa}; see \cite{Sharpe:2015mja,Gomes:2023ahz,Schafer-Nameki:2023jdn,Brennan:2023mmt,Bhardwaj:2023kri,Shao:2023gho} for reviews.}

The existence of any conserved current of spin greater than $1$, besides the stress tensor (and, in theories with supersymmetry, the spin-$\frac{3}{2}$ supercurrent) imposes much stronger constraints on the dynamics of a QFT than the conditions associated with ordinary spin-$1$ currents. For instance, under mild assumptions, any three dimensional conformal field theory with a conserved current of spin greater than $2$ is essentially free \cite{Maldacena:2011jn}.\footnote{Analogous statements in higher dimensions have been discussed in \cite{Stanev:2013qra,Alba:2013yda,ABDALLA1982181}.} However, the situation in two spacetime dimensions is especially rich: here the existence of higher-spin conserved currents is a signature of integrability \cite{PhysRevD.17.2134,PARKE1980166}. Although they are not forced to be free, such $2d$ integrable quantum field theories (IQFTs) are sufficiently constrained that one can often solve for their dynamics exactly, which places them in a desirable ``Goldilocks zone'' of complexity and thus makes them an attractive theoretical laboratory.

Given the amount of research interest in $2d$ IQFTs, one would like to generate many examples of integrable models, and -- if possible -- to chart out the space of all such integrable theories. One way to make progress towards this goal is to find ways of \emph{deforming} a given integrable field theory in a way which preserves integrability, therefore producing a parameterized family of IQFTs. In the construction of such deformations, conserved quantities again play a starring role. It was shown in \cite{Smirnov:2016lqw} (building on earlier work \cite{Zamolodchikov:2004ce}) that one can construct an infinite family of integrability-preserving deformations of any $2d$ IQFT which are defined using bilinears in conserved quantities. For instance, consider a totally symmetric spin-$s$ current $j_{\mu_1 \ldots \mu_s}$ with non-vanishing ``positive'' components
\begin{align}
    j_{s +} = j_{\underbrace{+ + \ldots +}_{s \text{ times } } } \, , \qquad j_{(s-2)+} = j_{- \underbrace{+ \ldots +}_{s - 1  \text{ times } } } \, ,
\end{align}
along with ``negative'' components
\begin{align}
    j_{s -} = j_{\underbrace{- - \ldots -}_{s \text{ times } } } \, , \qquad j_{(s-2) -} = j_{+ \underbrace{- \ldots -}_{s - 1  \text{ times } } } \, ,
\end{align}
which satisfies the conservation equation
\begin{align}
    \partial^{\mu_1} j_{\mu_1 \ldots \mu_s} = 0 \, .
\end{align}
Then the analysis of \cite{Smirnov:2016lqw} shows that the coincident point limit
\begin{align}\label{OK_defn_intro}
    \mathcal{O}_s ( x ) = \lim_{y \to x} \left( j_{s +} ( x ) j_{s -} ( y ) - j_{(s - 2) +} ( x ) j_{(s - 2) -} ( y ) \right) \, ,
\end{align}
gives rise to a well-defined local operator in the spectrum of the theory. One can then deform the theory by adding this integrated local operator to its action, which generates a family of new theories that nonetheless are still integrable. When $s = 2$ and the conserved current used in this construction is the energy-momentum tensor, the object produced in this way is known as the $\TT$ operator \cite{Zamolodchikov:2004ce,Cavaglia:2016oda}. 
However, for the purposes of this discussion, nothing is special about $s=2$; any higher-spin current, of which there are typically infinitely many in integrable field theories, serve equally well for generating integrable deformations.

Let us emphasize that the quantum-mechanical definition of these integrable deformations requires that the deforming operator take precisely the form (\ref{OK_defn_intro}), which is bilinear in components of the currents with a particular relative coefficient. Therefore, for studies of \emph{quantum integrability}, the functional dependence of the operators $\mathcal{O}_k$ is quite constrained. Quantum integrability is a deep and interesting topic, with a long history that notably includes the original observations by the Zamolodchikov brothers that there exist $2d$ IQFTs with factorized S-matrices \cite{ZAMOLODCHIKOV1978525,ZAMOLODCHIKOV1979253}; the quantum integrable deformations (\ref{OK_defn_intro}) dress such factorized S-matrices with momentum-dependent phases known as Castillejo-Dalitz-Dyson (CDD) factors \cite{PhysRev.101.453}, which preserve factorization. Despite the many fascinating aspects of quantum integrability, for many purposes, even \emph{classically} integrable field theories are of considerable interest. One might ask whether more general deformations by higher-spin currents, beyond the specific bilinears (\ref{OK_defn_intro}), preserve integrability at the classical level.

In this work, we will address this question in the context of a specific integrable $2d$ field theory, the principal chiral model (PCM).\footnote{We also consider the non-Abelian T-dual of the PCM, and the PCM with Wess-Zumino (WZ) term.} The PCM belongs to a sub-class of $2d$ IQFTs known as integrable sigma models, which describe the dynamics of a mapping from a two-dimensional spacetime manifold into a particular target space. Sigma models find applications in many areas of theoretical physics, especially in string theory, where they describe the embedding of the string worldsheet into a target spacetime.\footnote{In particular, integrable sigma models which appear in descriptions of strings on $\mathrm{AdS}$ spacetimes have played an important role in holography; see \cite{Beisert:2010jr,Demulder:2023bux} for reviews.} Another famous example is the $O(3)$ sigma model, which has been used in condensed matter physics to describe the continuum field theory arising from a limit of spin chains \cite{HALDANE1983464}. The PCM, in particular, can also be viewed as a toy model for four-dimensional Yang-Mills theory, since the quantum theory is asymptotically free but becomes strongly coupled in the IR.

In addition to its many physical applications, the principal chiral model has the advantage that it admits many known integrable deformations (besides the ones associated with deformations by current bilinears, which we have introduced above). Examples of these include the addition of a Wess-Zumino term \cite{ABDALLA1982181}, Yang-Baxter deformations \cite{Klimcik:2002zj,Klimcik:2008eq}, $\lambda$ deformations \cite{Sfetsos:2013wia}, and many others; see \cite{Zarembo:2017muf,Orlando:2019his,Seibold:2020ouf,Klimcik:2021bjy,Hoare:2021dix} for reviews.

The main question of interest in this article -- whether deformations of the PCM by more general functions of higher-spin currents (\ref{OK_defn_intro}) are classically integrable -- has already been answered in the special case $s = 2$, in which case it is known that deformations by \emph{arbitrary} functions of the spin-$2$ conserved current (namely, the energy-momentum tensor) preserve classical integrability \cite{Ferko:2024ali}. This result was established by using an auxiliary field technique, inspired by the four-dimensional Ivanov-Zupnik formalism \cite{Ivanov:2002ab,Ivanov:2003uj}, which involves an interaction function of a single real variable. This family of auxiliary field deformations of the principal chiral model has been dubbed the ``auxiliary field sigma model'' or AFSM, and has recently been understood in terms of $4d$ Chern-Simons theory \cite{Fukushima:2024nxm}.

Our main novel result is to generalize the auxiliary field sigma model so that the interaction function depends on several real variables, each of which is constructed from a higher-spin combination of auxiliary fields. At least to leading order around the principal chiral model, we will see that this generalized family of auxiliary field models includes deformations of the PCM by \emph{arbitrary} functions of Lorentz scalars constructed from higher-spin conserved currents, much like the original AFSM of \cite{Ferko:2024ali} incorporates deformations by arbitrary functions of the energy-momentum tensor. We will conjecture, with some evidence, that this generalized family includes deformations by arbitrary functions of higher-spin currents at all orders. To the best of our knowledge, these observations represent the first results at the level of deformed Lagrangians which obey flow equations driven by combinations of higher-spin currents. We will also prove that every model in this family is classically integrable, which means that -- assuming our conjecture is correct -- this would answer our original question in the affirmative, and one could indeed conclude that general higher-spin current deformations of the PCM preserve classical integrability.

 Another motivation for the analysis in this work, besides understanding integrable deformations, concerns \emph{duality}. It is sometimes the case that two, apparently different, physical theories actually describe the same underlying system. In this case, we say that the two theories have ``dual descriptions'' of the system, or that they are related by a duality. Dualities are ubiquitous in physics; the main example of interest here, which we review in section \ref{sec:t_duality}, is T-duality, which was first noticed in the setting of string compactifications and which relates the principal chiral model to an equivalent (but seemingly different) dual model.

One can often learn more about the dual descriptions of a system by studying deformations on both sides of the duality. A famous example in $2d$ field theory is that of the duality between the massive Thirring theory and the sine-Gordon model \cite{PhysRevD.11.2088} (see also the recent reviews \cite{Torrielli:2022byn,Torrielli:2024bpa}). These two theories are described, in notation similar to that of \cite{tong2018gauge}, by the Lagrangians
\begin{align}
    \mathcal{L}_{g, m}^\psi = i \overbar{\psi} \gamma^\mu \partial_\mu \psi - m \overbar{\psi} \psi - g \overbar{\psi} \gamma^\mu \psi \overbar{\psi} \gamma_\mu \psi \, ,
\end{align}
and
\begin{align}
    \mathcal{L}_{g, m}^\phi = \frac{1}{2} \left( \frac{1}{4 \pi} + \frac{g}{2 \pi^2} \right) \left( \partial_\mu \phi \right)^2 + \frac{m}{\pi \epsilon} \cos ( \phi ) \, ,
\end{align}
respectively, where $\epsilon$ is a regularization parameter. These dual descriptions for a pair of two-parameter families of deformed theories can be visualized in the following diagram, where marginal deformations move horizontally and relevant deformations flow downward:
\begin{align}
    \begin{tikzcd}[row sep=large, column sep=large,ampersand replacement=\&]
    \& \mathcal{L}_{0, 0}^{\phi} \arrow[<->,color=cyan,"\text{dual}", swap,dl] \arrow[color=ForestGreen, "{\color{ForestGreen} g}", rr] \arrow[color=red, "{\color{red} m}", dd, pos=0.7] \& \& \mathcal{L}_{g, 0}^\phi \arrow[color=cyan,<->,"\text{dual}",dl] \arrow[color=red, "{\color{red} m}", dd] \\
    \mathcal{L}_{0,0}^\psi \arrow[color=ForestGreen, "{\color{ForestGreen} g}", rr, crossing over, pos = 0.6] \arrow[color=red, "{\color{red} m}", swap, dd] \& \& \mathcal{L}_{g, 0}^\psi \\
    \& \mathcal{L}_{0, m}^{\phi} \arrow[color=cyan,<->,"\text{dual}",dl] \arrow[color=ForestGreen, "{\color{ForestGreen} g}", pos=0.35, rr] \& \& \mathcal{L}_{g, m}^\phi \arrow[color=cyan, <->,"\text{dual}",dl] \\
    \mathcal{L}_{0, m}^\psi \arrow[color=ForestGreen, "{\color{ForestGreen} g}", swap, rr] \& \& \mathcal{L}_{g, m}^\psi \arrow[color=red, "{\color{red} m}", from=uu, crossing over, pos=0.6]\\
    \end{tikzcd} \, .
\end{align}
Although the degrees of freedom in these models are very different -- the Thirring model describes the dynamics of a massive fermionic field $\psi$ with a four-fermion interaction, while the sine-Gordon model describes a compact scalar field $\phi$ subject to a cosine potential -- the two models are in fact exactly equivalent. Furthermore, the deformations of the two models map into one another. Changing the radius of the compact scalar $\phi$ from its self-dual value $R_0^2 = \frac{1}{4 \pi}$ to another value $R^2 = \frac{1}{4 \pi} + \frac{g}{2 \pi^2}$ maps onto the activation of a four-fermion interaction with coupling constant $g$ in the dual model. Likewise, deforming the Thirring model by tuning the mass parameter $m$, in the fermionic duality frame, corresponds to adjusting the coefficient $m$ of the cosine potential in the bosonic duality frame. Thus, in the duality between the massive Thirring and sine-Gordon theories, these deformations look very different in the two duality frames: a mass term is quite dissimilar to a cosine potential, and a rescaling of the boson kinetic term looks rather unlike a four-fermion interaction.

Unlike the Thirring/sine-Gordon example, our auxiliary field deformations of the principal chiral model will have almost identical behavior on both sides of the duality. This is a strength of the family of deformations we consider here, and is a consequence of the fact that our deformations are built from conserved currents. In a sense, this means that the deformations are ``intrinsically defined'' since one can construct such deformations for any theory with higher-spin currents without needing to know the specific details about the fundamental fields in the model from which these currents are comprised. We will illustrate this general lesson about deformations and dual descriptions by studying auxiliary field deformations of the T-dual PCM explicitly, and we will find that they interact in a particularly simple way with the duality transformation. In particular, these deformations ``commute with T-duality'' in a sense which we will explain, and we check that they preserve classical integrability in both duality frames (as they should).

The structure of this paper is as follows. In section \ref{sec:hs_afsm}, we introduce the models of interest in this work, present their equations of motion, and explain their connection to higher-spin deformations. Section \ref{sec:integrability} shows that these models are classically integrable, in the sense that they possess an infinite set of Poisson-commuting conserved charges. In section \ref{sec:t_duality}, we discuss the T-duals of our deformed models, which are related to the original models by a canonical transformation and are therefore also integrable. Section \ref{sec:extensions} shows how to combine our auxiliary field deformations with the activation of a Wess-Zumino term. Finally, section \ref{sec:conclusion} summarizes our results and presents directions for future research. In an attempt to be thorough and pedagogical, we have devoted appendices \ref{app:eom} through \ref{app:tdual_details} to presenting the details of many elementary manipulations used in the body of this manuscript, which we hope will make this work accessible to students and researchers in adjacent areas. 

\textbf{Note added}. While this work was in preparation, the preprint \cite{Fukushima:2024nxm} appeared on the arXiv, which has some overlap with section \ref{sec:extensions} of this article. In particular, our section \ref{sec:wess_zumino} constructs a Lax representation for the equations of motion of an auxiliary field sigma model with Wess-Zumino term and an interaction function $E ( \nu_2, \ldots , \nu_N )$ of several variables. In \cite{Fukushima:2024nxm}, this was done in the special case where $E ( \nu_2 )$ depends on only one of these variables, and their construction agrees with ours in this case (up to differences of conventions).

\section{Auxiliary Field Sigma Model with Higher-Spin Couplings}\label{sec:hs_afsm}

In this section, we will introduce the family of integrable deformations of the principal chiral model which will be the focus of this work. We will also explain the basic features of these models, such as their equations of motion. Many of the technical steps in this discussion will proceed similarly to those in \cite{Ferko:2024ali}, which corresponds to a special sub-class of our models in which the interaction function $E$ is a function of one variable rather than of several. However, we find it instructive to present these arguments in somewhat more detail than those which appeared in the original treatment of \cite{Ferko:2024ali}.

\subsection{Definition of Deformed Models}\label{sec:models_defn}

We focus on classical field theories which are defined on a flat, two-dimensional spacetime manifold $\Sigma$. We will sometimes refer to $\Sigma$ as the worldsheet, and we choose coordinates $\sigma^\alpha = ( \sigma, \tau )$ on $\Sigma$. For concreteness, one can imagine that $\Sigma$ is the plane $\mathbb{R}^{1,1}$, in which case both the spatial coordinate $\sigma$ and the time coordinate $\tau$ are non-compact, or one can take $\Sigma = S^1 \times \mathbb{R}$ to be a cylinder in which the spatial coordinate is compactified with identification $\sigma \sim \sigma + 2 \pi$.

In particular, our interest lies in two-dimensional sigma models where the target space $G$ is a Lie group, whose Lie algebra will be written as $\mathfrak{g}$. The fundamental field of these models is denoted as $g ( \sigma, \tau )$, which is a map from the worldsheet $\Sigma$ into $G$. From this field $g$, one can construct two important quantities, the left- and right-invariant Maurer-Cartan forms, which are defined by
\begin{align}\label{left_right_maurer_cartan}
    j = g^{-1} d g \, , \text{ and } \qquad \tilde{j} = - ( dg ) g^{-1} \, ,
\end{align}
respectively. All of our subsequent discussion can be phrased either in terms of the left-invariant Maurer-Cartan form $j$ or, equivalently, in terms of its right-invariant analogue $\tilde{j}$. To avoid repeating ourselves, we will phrase the entire construction using the left-invariant form $j$, although the reader should keep in mind that one could straightforwardly reformulate our presentation in terms of $\tilde{j}$.

It is convenient to introduce local coordinates $\phi^\mu$ on the Lie group $G$, where $\mu = 1 , \ldots , \dim ( G )$ is a target-space index. One can then express the group-valued field $g$ as $g ( \phi^\mu ( \sigma^\alpha ) )$. Note that we will always use early Greek letters like $\alpha$, $\beta$ for indices on the worldsheet $\Sigma$, and middle Greek letters like $\mu$, $\nu$ for target-space indices on the Lie group. From $\phi^\mu$ one can define the pull-back map
\begin{align}\label{pullback}
    \frac{\partial \phi^\mu}{\partial \sigma^\alpha} \, ,
\end{align}
which allows us to convert between indices on $G$ and those on $\Sigma$. For instance, one can define the pull-back of the left-invariant Maurer-Cartan form by
\begin{align}
    j_\alpha = \frac{\partial \phi^\mu}{\partial \sigma^\alpha} j_\mu \, .
\end{align}
We find it useful to introduce light-cone coordinates $\sigma^{\pm} = \frac{1}{2} ( \tau \pm \sigma )$ on $\Sigma$. These indices are raised or lowered with the flat worldsheet metric $\eta_{\alpha \beta}$ or its inverse $\eta^{\alpha \beta}$, whose non-zero components in light-cone coordinates are
\begin{align}
    \eta_{+ - } = \eta_{- + } = - 2 \, , \qquad \eta^{+ -} = \eta^{- + } = - \frac{1}{2} \, .
\end{align}
By virtue of its definition, the form $j_\alpha$ satisfies the Maurer-Cartan identity, which can be written using our light-cone notation as
\begin{align}\label{mc_identity}
    \partial_+ j_- - \partial_- j_+ + [ j_+ , j_- ] = 0 \, .
\end{align}
We emphasize that (\ref{mc_identity}) is an off-shell relation, which holds identically and is not a consequence of the equations of motion for any particular model. When the Euler-Lagrange equations are satisfied, the form $j_\alpha$ will satisfy additional on-shell constraints. One familiar example is the principal chiral model, whose Lagrangian is
\begin{align}\label{pcm_lagrangian}
    \mathcal{L}_{\text{PCM}} &= \frac{1}{2} \eta^{\alpha \beta} \tr ( j_\alpha j_\beta ) \nonumber \\
    &= - \frac{1}{2} \tr ( j_+ j_- ) \, .
\end{align}
Here ``$\tr$'' is a matrix trace over a representation of the Lie algebra to which $j$ belongs.
The equation of motion which arises from variation of (\ref{pcm_lagrangian}) is
\begin{align}
    \partial_\alpha j^\alpha = 0 \, ,
\end{align}
so when the field $j_\alpha$ of the PCM is on-shell, it satisfies both the Maurer-Cartan identity (\ref{mc_identity}) and the conservation equation $\partial_+ j_- + \partial_- j_+ = 0$.

The objective of this article is to study a class of models which deforms the PCM (\ref{pcm_lagrangian}) through the activation of specific interactions. The most economical way to present this family of deformed models is to introduce, in addition to the physical group-valued field $g$, an additional Lie algebra valued vector field $v_\pm$ on the worldsheet, which is purely auxiliary in the sense that its equations of motion are algebraic.

Much like the combination $\tr ( j_+ j_- )$ which appears in the PCM Lagrangian (\ref{pcm_lagrangian}), any trace of a product involving equal numbers of $v_+$ and $v_-$ factors is both a c-number (as opposed to an element of $\mathfrak{g}$) and a Lorentz scalar. Useful combinations of this type include
\begin{align}
    \nu_2 = \tr ( v_+ v_+ ) \tr ( v_- v_- ) \, , \quad \ldots \, , \quad \nu_k = \tr ( v_+^k ) \tr ( v_-^k ) \, ,
\end{align}
for any integer $k \geq 2$. For each choice of finite-dimensional Lie algebra $\mathfrak{g}$, not all of the combinations $\nu_k$ will be independent. One way to see this is by representing the elements of $\mathfrak{g}$ as $N \times N$ matrices for some fixed $N$. In this case, by the Cayley-Hamilton theorem, we know that for any $N \times N$ matrix $M$, the first $N$ traces
\begin{align}
    \tr ( M^j ) \, , \qquad j = 1 , \ldots , N \, , 
\end{align}
are functionally independent, but all higher traces $\tr ( M^j )$ for $j > N$ can be expressed in terms of lower traces. Similarly, only finitely many traces $\nu_1 , \ldots , \nu_N$ are independent variables, after which the $\nu_k$ with larger indices are functions of those with smaller indices. Despite this, we will sometimes find it convenient to speak of the complete set of $\nu_k$ for all $k \in \mathbb{N}$, even though only a finite subset is functionally independent. 

We are now in a position to introduce our class of deformed PCM-like models. Letting $N$ be the number of functionally independent traces $\nu_j$ as above, consider the Lagrangian
\begin{align}\label{deformed_family}
    \mathcal{L}_{\text{AFSM}} = \frac{1}{2} \mathrm{tr} ( j_+ j_- ) + \mathrm{tr} ( v_+ v_- ) + \mathrm{tr} ( j_+ v_- + j_- v_+ ) + E ( \nu_2 , \nu_3 , \ldots , \nu_N ) \, ,
\end{align}
where $E$ is an arbitrary differentiable function of $N - 1$ variables. We refer to (\ref{deformed_family}) as a higher-spin auxiliary field sigma model, or simply an auxiliary field sigma model (AFSM).

Let us next give the Euler-Lagrange equations for a general model of the form (\ref{deformed_family}). The detailed derivation of these equations of motion has been relegated to appendix \ref{app:eom}. The equation of motion for the physical group-valued field $g$ can be written as
\begin{align}\label{j_eom}
    \partial_+ ( j_- + 2 v_- ) + \partial_- ( j_+ + 2 v_+ ) = 2 \left( [ v_- , j_+ ] + [ v_+, j_- ] \right) \, .
\end{align}
To give the equation of motion for the auxiliary field $v_{\pm}$, we first introduce some notation. We expand any Lie algebra valued quantity $X$ as
\begin{align}
    X = X^A T_A \, ,
\end{align}
where $T_A$ is a basis of generators for the representation which we are considering for $\mathfrak{g}$. We use capital early Latin letters like $A$, $B$, $C=1,\ldots,\dim{G}$ for indices in the adjoint representation of the Lie algebra. All such indices are lowered with the Killing-Cartan form,
\begin{align}\label{killing_cartan}
    \gamma_{AB} = \tr ( T_A T_B ) \, .
\end{align}
We assume that $\gamma_{AB}$ is non-degenerate, or equivalently that the Lie algebra $\mathfrak{g}$ is semi-simple, and we raise indices with the inverse $\gamma^{AB}$ of the Killing-Cartan form (\ref{killing_cartan}).

In terms of these quantities, the equation of motion for $v_{\pm}$ can be written as
\begin{align}\label{v_eom}
    0 = j_{\pm} + v_{\pm} + \sum_{n = 2}^{N} n E_n \tr ( v_{\pm}^n ) v_{\mp}^{A_1} v_{\mp}^{A_2} \ldots v_{\mp}^{A_{n-1}} T^{A_n} \tr ( T_{(A_1} T_{A_2} \ldots T_{A_{n} )} ) \, ,
\end{align}
where we have defined the shorthand
\begin{align}
    E_n = \frac{\partial E}{\partial \nu_n} \, ,
\end{align}
and where our convention for symmetrization of indices is
\begin{align}
    X_{(A_1 \ldots A_n)} = \frac{1}{n!} \sum_{\sigma \in S_n} X_{A_{\sigma ( 1 ) } A_{\sigma ( 2 ) } \ldots A_{\sigma ( n ) }} \, .
\end{align}
Just as the status of the physical field $g$ differs from that of the auxiliary field $v_{\pm}$, as the latter does not give rise to a true propagating degree of freedom, the two equations of motion (\ref{j_eom}) and (\ref{v_eom}) also play different roles in our analysis. While (\ref{j_eom}) is interpreted as a dynamical condition, we view (\ref{v_eom}) as simply an algebraic constraint. We will highlight this distinction by introducing another piece of notation: we use the symbol $\deq$ to indicate that two quantities are equal when the auxiliary field equation of motion is satisfied.

\subsection{Implications of Equations of Motion}\label{sec:implications}

We now turn to studying the implications of the Euler-Lagrange equations associated with the theories (\ref{deformed_family}), and in particular, the consequences of the auxiliary field equation of motion. We will find that the structure of this equation implies several simple relations which will be useful in establishing other properties of this family of models, such as integrability.

We are especially interested in expressions involving commutators of fields, since such commutators appear in the flatness conditions which a would-be Lax connection should satisfy in order for the theory to admit a zero-curvature representation. First let us consider the commutator of $v_{\mp}$ with $j_{\pm}$. Using the auxiliary field equation of motion (\ref{v_eom}), this is
\begin{align}\label{last_term_vanish}
    [ v_{\mp} , j_{\pm} ] &\deq \left[ v_{\mp} , - v_{\pm} - \sum_{n = 2}^{N} n E_n \tr ( v_{\pm}^n ) v_{\mp}^{A_1} \ldots v_{\mp}^{A_{n-1}} T^B \tr ( T_{(B} T_{A_1} \ldots T_{A_{n-1} )} ) \right] \nonumber \\
    &= [ v_\pm , v_\mp ] + \sum_{n=2}^{N} n E_n \tr ( v_{\pm}^n ) v_{\mp}^{A_1} \ldots v_{\mp}^{A_{n-1}} \tr ( T_{(B} T_{A_1} \ldots T_{A_{n-1} )} ) [ T^B , v_{\mp}^C T_C ] \nonumber \\
    &= [ v_\pm , v_\mp ] + \sum_{n=2}^{N} n E_n \tr ( v_{\pm}^n ) v_{\mp}^{A_1} \ldots v_{\mp}^{A_{n-1}} v_{\mp}^C \tr ( T_{(B} T_{A_1} \ldots T_{A_{n-1} )} ) \tensor{f}{^B_C^D} T_D \, .
\end{align}
To analyze the terms in the sum of the final line of (\ref{last_term_vanish}), it is helpful to define
\begin{align}
    M_{B A_1 \ldots A_{n-1}} = \tr ( T_{(B} T_{A_1} \ldots T_{A_{n-1} )} ) \, .
\end{align}
Note that each term in this sum involves a contraction against $v_{\mp}^{A_1} \ldots v_{\mp}^{A_{n-1}} v_{\mp}^C$, which is totally symmetric in its indices. Thus we are free to symmetrize over $A_1$, $\ldots$, $A_{n-1}$, $C$ without penalty. However, upon performing this symmetrization, one finds that
\begin{align}\label{gen_jac_body}
    M_{B ( A_1 \ldots A_{n-1}}  \tensor{f}{^B_{C)}^D} = 0 \, .
\end{align}
The relation (\ref{gen_jac_body}) is a simple corollary of a result which is sometimes known as the ``generalized Jacobi identity'' for Lie algebras, which we discuss and prove in appendix \ref{app:jac}, and which itself is a consequence of the ad-invariance of the trace. It follows that each term in the sum of the final line of (\ref{last_term_vanish}) vanishes identically, and we conclude that
\begin{align}\label{fundamental_commutator_identity}
    [ v_{\mp} , j_{\pm} ] \deq [ v_{\pm} , v_{\mp} ] \, .
\end{align}
The relation (\ref{fundamental_commutator_identity}) allows us to rewrite the physical equations of motion for auxiliary field sigma models in an especially simple way, assuming that the auxiliary field equation of motion is satisfied. For instance, we note that (\ref{fundamental_commutator_identity}) implies
\begin{align}
    [ v_- , j_+ ] + [ v_+ , j_- ] \deq [ v_+ , v_- ] + [ v_- , v_+ ] = 0 \, ,
\end{align}
and thus the terms on the right side of the $g$-field equation of motion (\ref{j_eom}) vanish when the $v_{\pm}$ field equation of motion is satisfied. This means that the Euler-Lagrange equation for $g$ can be written as
\begin{align}\label{j_eom_deq}
    0 \deq \partial_+ ( j_- + 2 v_- ) + \partial_- ( j_+ + 2 v_+ ) \, .
\end{align}
The form of (\ref{j_eom_deq}) suggests that one should define a modified current,
\begin{align}\label{frakJ_defn}
    \CC{J}_{\pm} = - \left( j_{\pm} + 2 v_{\pm} \right) \, ,
\end{align}
so that the $g$-field equation of motion can be represented as a conservation condition,
\begin{align}\label{afsm_eom_dot}
    \partial_\alpha \CC{J}^\alpha \deq 0 \, .
\end{align}
The reason for the choice of an overall minus sign in the definition (\ref{frakJ_defn}) is the following. If we choose the trivial interaction function $E = 0$, the auxiliary field equation of motion (\ref{v_eom}) becomes $v_{\pm} \deq - j_{\pm}$. In this case,  (\ref{frakJ_defn}) becomes
\begin{align}
    \CC{J}_{\pm} = - \left( j_{\pm} - v_{\pm} \right) = j_{\pm} \qquad \text{(when } E = 0 \text{) } \, ,
\end{align}
so this choice of sign ensures that $\CC{J}_{\pm}$ and $j_{\pm}$ coincide in the trivial case. We also note that, when $E = 0$, eliminating the auxiliary field using its equation of motion gives
\begin{align}
    \mathcal{L} \deq - \frac{1}{2} \tr ( j_+ j_- ) \, ,
\end{align}
which agrees with the Lagrangian (\ref{pcm_lagrangian}) for the undeformed PCM.

Let us also remark that, in terms of the modified currents (\ref{frakJ_defn}), the Lagrangian (\ref{deformed_family}) can be expressed as
\begin{align}
    \mathcal{L}_{\text{AFSM}} = \frac{1}{2} \tr \left( \CC{J}_+ \CC{J}_- \right) - \tr ( v_+ v_- ) + E ( \nu_2 , \ldots , \nu_N ) \, .
\end{align}
To conclude this subsection, we record a few more implications of the equations of motion which will be useful in section \ref{sec:integrability}. By virtue of the definition (\ref{frakJ_defn}) and the auxiliary field equation of motion (\ref{v_eom}), one has
\begin{align}
    [ \CC{J}_+ , j_- ] &= - [ j_+ , j_- ] - 2 [ v_+ , j_- ] \nonumber \\
    &\deq - [ j_+ , j_- ] - 2 [ v_- , v_+ ] \, ,
\end{align}
and likewise
\begin{align}
    [ j_+ , \CC{J}_- ] &= - [ j_+ , j_- ] - 2 [ j_+ , v_- ] \nonumber \\
    &\deq - [ j_+ , j_- ] - 2 [ v_- , v_+ ] \, ,
\end{align}
and comparing these two expressions gives
\begin{align}\label{imp_one}
    [ \CC{J}_+ , j_- ] \deq [ j_+ , \CC{J}_- ] \, .
\end{align}
Using similar manipulations, we also find
\begin{align}\label{imp_two}
    [ \CC{J}_+ , \CC{J}_- ] &= [ j_+ , j_- ] + 2 [ v_+, j_- ] + 2 [ j_+ , v_- ] + 4 [ v_+ , v_- ] \nonumber \\
    &\deq [j_+ , j_- ] + 2 [ v_- , v_+ ] + 2 [ v_- , v_+ ] + 4 [ v_+ , v_- ] \nonumber \\
    &= [ j_+ , j_- ] \, ,
\end{align}
and then one gets the useful result $[ \CC{J}_+ , \CC{J}_- ] \deq[ j_+ , j_- ]$. The utility of these formulas is due to the fact that, when we are interested in ``$\deq$ equations'' that hold when the auxiliary field equations of motion are satisfied, we may work directly with the fields $j_{\pm}$ and $\CC{J}_{\pm}$ rather than referring to the auxiliary fields $v_{\pm}$ directly. Furthermore, these objects $j_{\pm}$ and $\CC{J}_{\pm}$ interact in an especially simple way.

\subsection{Connection to Higher-Spin Deformations}\label{sec:higher_spin}

It is worth pausing to comment on the interpretation of the deformation function $E$, and in particular the role of the different variables $\nu_k$. Let us first illustrate by considering an infinitesimal deformation around the undeformed principal chiral model which depends on a single variable $\nu_k$ for some fixed $k \geq 2$. That is, we take the interaction function to be
\begin{align}
    E ( \nu_2 , , \ldots , \nu_N ) = \lambda \nu_k \, ,
\end{align}
and work to leading order in $\lambda$. In this case, by the equation of motion for $v_{\pm}$, one has
\begin{align}
    v_{\pm} &\deq - j_{\pm} - k \lambda E_k \tr ( v_{\pm}^k ) v_{\mp}^{A_1} v_{\mp}^{A_2} \ldots v_{\mp}^{A_{n-1}} T^{A_k} \tr ( T_{(A_1} T_{A_2} \ldots T_{A_{k} )} ) \nonumber \\
    &= - j_{\pm} - k \lambda E_k \mathcal{V}_{\pm} \, ,
\end{align}
where we have defined
\begin{align}
    \mathcal{V}_{\pm} = \tr ( v_{\pm}^k ) v_{\mp}^{A_1} v_{\mp}^{A_2} \ldots v_{\mp}^{A_{n-1}} T^{A_k} \tr ( T_{(A_1} T_{A_2} \ldots T_{A_{k} )} ) 
\end{align}
to ease notation. Eliminating $v_{\pm}$ using its equation of motion in (\ref{deformed_family}) gives
\begin{align}
    \mathcal{L} &\deq \frac{1}{2} \tr ( j_+ j_- ) + \tr \left( \left( j_+ + \lambda \mathcal{V}_+ \right) \left( j_- + \lambda \mathcal{V}_- \right) \right) - \tr \left( j_+ \left( j_- + \lambda \mathcal{V}_- \right) + j_- \left( j_+ + \lambda \mathcal{V}_+ \right) \right) \nonumber \\
    &\qquad + \lambda \tr ( j_+^k ) \tr ( j_-^k ) + \mathcal{O} ( \lambda^2 ) \nonumber \\
    &= - \frac{1}{2} \tr ( j_+ j_- ) + \lambda \tr ( j_+^k ) \tr ( j_-^k ) + \mathcal{O} ( \lambda^2 ) \, .
\end{align}
Therefore, at leading order in the parameter $\lambda$, this interaction function has implemented a deformation of the principal chiral model by the operator
\begin{align}\label{Ok_defn}
    \mathcal{O}_k = \mathcal{J}_{k+} \mathcal{J}_{k-} \, , \qquad \mathcal{J}_{k \pm} = \tr ( j_{\pm} ^k ) \, .
\end{align}
This operator is a scalar bilinear constructed from spin-$k$ conserved currents in the principal chiral model.\footnote{Although our discussion is classical, we note that the combinations (\ref{Ok_defn}) are of the type (\ref{OK_defn_intro}) which give rise to well-defined local operators by the point-splitting argument of \cite{Smirnov:2016lqw}, as mentioned in the introduction. Thus these higher-spin combinations also give integrable deformations of the PCM at the quantum level.} Indeed, to see that the objects $\mathcal{J}_{k \pm}$ are conserved for the PCM, one can combine the equations of motion and Maurer-Cartan identity,
\begin{align}
    \partial_+ j_- + \partial_- j_+ = 0 \, , \qquad \partial_+ j_- - \partial_- j_+ + [ j_+ , j_- ] = 0 \, ,
\end{align}
to find
\begin{align}
    \partial_{\pm} j_{\mp} + \frac{1}{2} [ j_{\pm} , j_{\mp} ] = 0 \, .
\end{align}
Then, for any $k$, one finds that
\begin{align}\label{conservation_higher_spin}
    \partial_{\pm} \mathcal{J}_{k \mp} &= k \tr \left( j_{\mp}^{k-1} \partial_{\pm} j_{\mp} \right) \nonumber \\
    &= - \frac{k}{2} \tr \left( j_{\mp}^{k-1} [ j_{\pm} , j_{\mp} ] \right)] \nonumber \\
    &= - \frac{k}{2} \tr \left( j_{\mp}^{k-1} j_{\pm} j_{\mp} - j_{\mp}^{k-1} j_{\mp} j_{\mp} \right) \nonumber \\
    &= 0 \, ,
\end{align}
where in the last step we have used cyclicity of the trace.

For the case $k = 2$, these conserved currents are nothing but the non-zero components of the energy-momentum tensor for the PCM,
\begin{align}
    \mathcal{J}_{\pm \pm} = T_{\pm \pm} \, ,
\end{align}
and the operator $\mathcal{O}_k$ in equation (\ref{Ok_defn}) is simply $T_{++} T_{--}$, which yields the leading-order $\TT$ deformation. More generally, it was shown in \cite{Ferko:2024ali} that a general interaction function $E ( \nu_2 )$ corresponds to a model obtained by solving a flow equation of the form
\begin{align}\label{stress_tensor_flow}
    \frac{\partial \mathcal{L}}{\partial \lambda} = f \left( T_{\alpha \beta}^{(\lambda)} , \lambda \right) \, , \qquad \mathcal{L} ( \lambda = 0 ) = \mathcal{L}_{\text{PCM}} \, , 
\end{align}
which is driven by some function of the stress tensor $T_{\alpha \beta}^{(\lambda)}$ associated with $\mathcal{L} ( \lambda )$. It is possible to argue for this conclusion, in the case $k = 2$, because the energy-momentum tensor can be computed for a general theory in the class (\ref{deformed_family}) using the Hilbert definition,
\begin{align}
    T_{\alpha \beta} = - \frac{2}{\sqrt{| g | } } \frac{\delta S}{\delta g^{\alpha \beta}} \, ,
\end{align}
after coupling the theory to a background metric $g_{\alpha \beta}$. For specific stress tensor deformations of the form (\ref{stress_tensor_flow}), the solutions for the interaction function $E ( \nu_2 )$ are known in closed form. For instance, the solution to the root-$\TT$ \cite{Ferko:2022cix,Conti:2022egv,Babaei-Aghbolagh:2022uij,Babaei-Aghbolagh:2022leo} flow equation
\begin{align}
    \frac{\partial \mathcal{L}}{\partial \gamma} = \frac{1}{\sqrt{2}} \sqrt{ T^{\alpha \beta} T_{\alpha \beta} - \frac{1}{2} \left( \tensor{T}{^\alpha_\alpha} \right)^2 } \, ,
\end{align}
with initial condition $E ( \gamma = 0 ) = 0$, is given by
\begin{align}
    E_{\text{Root-} \TT} ( \nu_2 ) = \tanh \left( \frac{\gamma}{2} \right) \sqrt{ \nu_2 } \, ,
\end{align}
while the solution to the $\TT$ flow
\begin{align}
    \frac{\partial \mathcal{L}}{\partial \lambda} = T^{\alpha \beta} T_{\alpha \beta} - \left( \tensor{T}{^\alpha_\alpha} \right)^2 \, ,
\end{align}
again with initial condition $E ( \lambda = 0 ) = 0$, is
\begin{align}
    E_{\TT} ( \nu_2 ) = - \frac{3}{8 \lambda} \left( {}_{3} F_2 \left( - \frac{1}{2} , - \frac{1}{4} , \frac{1}{4} ; \frac{1}{3} , \frac{2}{3} ; - \frac{256}{27} \lambda^2 \nu_2 \right) - 1 \right) \, .
\end{align}
For $k > 2$, to the best of our knowledge, it is not known how to systematically compute the higher-spin conserved currents $\mathcal{J}_{k \pm}$ for a general theory of the form (\ref{deformed_family}). Despite this, we find it natural to conjecture -- by analogy with the case $k = 2$, and due to the perturbative analysis given above to leading order in $\lambda$ for general $k$ -- that the collection of deformed models (\ref{deformed_family}) contains all solutions to flow equations of the form 
\begin{align}\label{stress_tensor_flow_higher_spin}
    \frac{\partial \mathcal{L}}{\partial \lambda} = f \left( \mathcal{J}_{\pm 2} , \ldots , \mathcal{J}_{\pm N} , \lambda \right) \, , \qquad \mathcal{L} ( \lambda = 0 ) = \mathcal{L}_{\text{PCM}} \, , 
\end{align}
which are driven by an arbitrary function of the Lorentz scalars that can be constructed from each of the individual currents $\mathcal{J}_{\pm k}$. However, strictly speaking, we have only proven that this class of models includes \emph{some} higher-spin deformations, since it includes deformations by the operators $\mathcal{O}_k$ of (\ref{Ok_defn}) for all $k$, which coincide with these higher-spin deformations at leading order in $\lambda$. Nonetheless, in the remainder of this work we will use terms like ``higher-spin couplings'' or ``higher-spin auxiliary field sigma models'' to refer to interaction functions which may depend on all of the independent traces $\nu_k$. One may also justify this language by noting that the $\nu_k$ are themselves higher-spin combinations of auxiliary fields.

\section{Classical Integrability}\label{sec:integrability}

The conventional definition of a classically integrable field theory is one which exhibits an infinite collection of independent conserved charges, all of which are mutually Poisson-commuting, or in involution. The existence of such a set of charges severely constrains the dynamics of the theory, which often allows one to answer questions about certain physical observables exactly. Because the existence of the conserved charges in an integrable field theory may not be obvious \emph{a priori}, it is also said that such a theory exhibits a large algebra of ``hidden symmetries.'' The first step, therefore, in applying the toolkit of integrability is to establish the existence of this infinite set of charges in a theory, and to construct them.

One might proceed in identifying an infinite set of charges in several ways. One is to simply write down expressions for these charges directly. We have already seen, around equation (\ref{conservation_higher_spin}), that the PCM possesses an infinite set of local higher-spin conserved currents $\mathcal{J}_{k \pm}$, from which one can construct corresponding charges $Q_k$. However, for more general models, it is typically not obvious how to identify such an infinite set of local conserved quantities, so one must use a different strategy.

One particularly powerful approach, which has been used to establish integrability of several deformations of the PCM, involves carrying out the following two steps:

\begin{enumerate}[label = (\Roman*)]

    \item\label{lax_step} First, one recasts the classical equations of motion for the model under consideration as a ``zero-curvature condition'' for a one-form $\CC{L} ( z )$ called the Lax connection. That is, one proves that $\CC{L}$ is flat (for any value of an auxiliary variable $z \in \mathbb{C}$ called the spectral parameter) if and only if the equations of motion are satisfied.\footnote{It is important that this Lax connection depends meromorphically on the spectral parameter $z$, and not in a trivial way. For instance, if $\CC{L}$ is independent of $z$, then one cannot perform appropriate Taylor series expansions and generate conserved charges in the way which we will review in section \ref{sec:charge_review}.}

    \item\label{maillet_step} Second, one extracts an infinite set of conserved charges using a particular path-ordered exponential of the Lax connection (which is called the monodromy matrix), and demonstrates that these charges are mutually Poisson-commuting. The proof of Poisson-commutativity usually requires demonstrating that the Poisson bracket of the spatial component of the Lax connection takes a special form, such as the one considered by Skylanin \cite{Sklyanin:1980ij} or that of Maillet \cite{MAILLET198654,MAILLET1986401}.

\end{enumerate}

Because the second step \ref{maillet_step} is a refinement of the first step \ref{lax_step}, one sometimes says that a model whose equations of motion admit a zero-curvature representation is ``weakly integrable,'' and that one whose Lax connection furthermore has the special properties described in step \ref{maillet_step} is ``strongly integrable.''

The goal of this section is to follow the two-step procedure outlined above for the case of the AFSM with higher-spin couplings. We complete step \ref{lax_step}, the construction of a Lax representation which establishes weak integrability, in section \ref{sec:weak_integrability}. We then show that the Poisson bracket of this Lax connection takes the form due to Maillet, proving strong integrability and the existence of infinitely many involutive charges, in section \ref{sec:maillet}.

\subsection{Zero-Curvature Representation}\label{sec:weak_integrability}

In this subsection, we demonstrate that the classical equations of motion for any model in the family (\ref{deformed_family}) can be presented as a flatness condition for a Lax connection $\CC{L} ( z )$ that depends on a spectral parameter $z \in \mathbb{C}$. The steps of this argument depend only on the implications of the equations of motion that were developed in section \ref{sec:implications}, which are identical to those presented in \cite{Ferko:2024ali} (and to those which hold in special cases such as $\TT$ and root-$\TT$ deformations of the PCM \cite{Borsato:2022tmu}). However, we reiterate that the deformations considered in this article are more general than the ones of these previous works, as our interaction function $E$ depends on many variables $\nu_k$ rather than on a single variable $\nu_2$.

\subsubsection*{\ul{\it Definition of Lax representation for auxiliary field models}}

First let us establish some terminology. Quite generally, suppose that we are interested in a Lagrangian $\mathcal{L} ( X )$ for a set of physical fields $X$, which may carry an arbitrary collection of Lorentz and internal indices that we suppress for ease of notation. By the term ``physical fields'' we mean that the fields $X$ are genuine propagating degrees of freedom within the model, whose equations of motion are differential rather than algebraic. Let
\begin{align}
    \mathcal{E}_X = \partial_\alpha \left( \frac{\partial \mathcal{L}}{\partial ( \partial_\alpha X ) } \right) - \frac{\partial \mathcal{L}}{\partial X}
\end{align}
be the quantities which arise from variation of $\mathcal{L}$ with respect to the fields $X$. For simplicity, here we have assumed that the Lagrangian depends on first derivatives of the fields but not higher derivatives. Then the equations of motion are satisfied when
\begin{align}\label{general_eom}
    \mathcal{E}_X = 0 \, ,
\end{align}
where (\ref{general_eom}) is understood to hold for each independent field $X$, assuming that there are several such fields labeled by indices that have been suppressed.

We say that the theory described by the Lagrangian $\mathcal{L}$ admits a \emph{zero-curvature representation} -- or equivalently, that it possesses a \emph{Lax representation}, or that the model is \emph{weakly integrable} -- if there exists a field-valued one-form $\CC{L} ( z )$, which is a meromorphic function of a spectral parameter $z \in \mathbb{C}$, such that
\begin{align}\label{general_weak_integrability}
    \left( \, \mathcal{E}_X = 0 \, \right) \; \iff \; \left( \, d_\CC{L} \CC{L} = 0 \, , \, \forall z \in \mathbb{C} \, \right) \, .
\end{align}
Here $d_{\CC{L}} = d + \CC{L} \, \wedge \,$ is the covariant exterior derivative associated with the connection $\CC{L}$, and hence $d_{\CC{L}} \CC{L}$ is the curvature of $\CC{L}$, whose expression in light-cone coordinates is
\begin{align}
    d_{\CC{L}} \CC{L} = \partial_+ \CC{L}_- - \partial_- \CC{L}_+ + [ \CC{L}_+ , \CC{L}_- ] \, .
\end{align}
This definition is unambiguous when the Lagrangian involves only physical fields $X$. However, now suppose that we are interested in a Lagrangian $\mathcal{L} ( X, Y )$ which, in addition to the propagating degrees of freedom $X$, also involves a collection of auxiliary fields $Y$ with algebraic equations of motion. In this case, we should clarify \emph{which} equations of motion should be equivalent to the flatness of $\CC{L}$ in an appropriate definition of weak integrability.

The definition we will use is the one proposed in \cite{Ferko:2024ali}. Following the notation introduced above, we use the symbol $\deq$ to indicate two quantities that are equal when all of the equations of motion for the auxiliary fields $Y$ are satisfied. Then we say that the Lagrangian $\mathcal{L} ( X, Y )$ admits a zero-curvature (Lax) representation, or is weakly integrable, if
\begin{align}\label{aux_weak_integrability}
    \left( \, \mathcal{E}_X \deq 0 \, \right) \; \iff \; \left( \, d_{\CC{L}} \CC{L} \deq 0 \, , \, \forall z \in \mathbb{C} \, \right) \, ,
\end{align}
where $\CC{L} ( X, Y )$ is the Lax connection.

Our reason for using this definition of a Lax representation is that, if one integrates out the auxiliary fields in all quantities using their equations of motion, the definition (\ref{aux_weak_integrability}) agrees with (\ref{general_weak_integrability}). More precisely, suppose that an auxiliary field Lagrangian $\mathcal{L} ( X , Y )$ satisfies the condition (\ref{aux_weak_integrability}) for some Lax connection $\CC{L} ( X, Y )$. Formally, one can define the two corresponding quantities $\mathcal{L} ( X )$ and $\CC{L} ( X )$ which depend only on physical fields as
\begin{align}
    \mathcal{L} ( X ) \deq \mathcal{L} ( X , Y ) \, , \qquad \CC{L} ( X ) \deq \CC{L} ( X, Y ) \, .
\end{align}
After all auxiliary fields have been eliminated, any two quantities which were equal with respect to the ``dotted'' equality symbol $\deq$ now become equal with respect to the standard equality. Therefore, since $\mathcal{L} ( X, Y)$ and $\CC{L} ( X, Y )$ obey (\ref{aux_weak_integrability}), their auxiliary-free cousins $\mathcal{L} ( X ) $ and $\CC{L} ( X )$ obey (\ref{general_weak_integrability}). For this reason, the definition (\ref{aux_weak_integrability}) of a zero-curvature representation for an auxiliary field model is quite natural, since it is simply equivalent to the standard definition when one ``forgets the auxiliaries'' and instead views an auxiliary field model as describing only the dynamics of the physical degrees of freedom.

\subsubsection*{\ul{\it Proof of weak integrability for auxiliary field sigma models}}

Having established our definition (\ref{aux_weak_integrability}) of a Lax representation for an auxiliary field model, we now show that every member of the family of Lagrangains (\ref{deformed_family}) satisfies this condition, regardless of the interaction function $E$.

We begin by postulating the form of the Lax connection,
\begin{align}\label{lax_connection}
    \CC{L}_{\pm} = \frac{j_{\pm} \pm z \CC{J}_{\pm}}{1 - z^2} \, ,
\end{align}
where $z \in \mathbb{C}$ is the spectral parameter and $\CC{J}_{\pm}$ is defined in equation (\ref{frakJ_defn}).

We are interested in the curvature of this Lax connection (\ref{lax_connection}) when the auxiliary field equations of motion are satisfied, since this is the quantity that enters into the condition (\ref{aux_weak_integrability}). It will be helpful to first compute the commutator of the Lax connection with itself, again assuming that the $v_{\pm}$ equation of motion is obeyed. One finds
\begin{align}\label{lax_commutator}
    [ \CC{L}_+ , \CC{L}_- ] &= \frac{1}{\left( 1 - z^2 \right)^2 } [ j_+ + z \CC{J}_+ , j_- - z \CC{J}_- ] \nonumber \\
    &= \frac{1}{\left( 1 - z^2 \right)^2 } \left( [ j_+, j_- ] + z \left( [ \CC{J}_+ , j_- ] + [ j_+ , \CC{J}_- ] \right) - z^2 [ \CC{J}_+ , \CC{J}_- ] \right) \nonumber \\
    &\deq \frac{1}{\left( 1 - z^2 \right)^2 } \left( [ j_+, j_- ] - z^2 [ j_+ ,  j_- ] \right) \nonumber \\
    &= \frac{[j_+, j_-]}{1 - z^2} \, .
\end{align}
In the third line, we have used the two implications (\ref{imp_one}) and (\ref{imp_two}) of the auxiliary field equations of motion, which give $[ \CC{J}_+ , j_- ] + [ j_+ , \CC{J}_- ] \deq 0$ and $[ \CC{J}_+ , \CC{J}_- ]  \deq [ j_+ , j_- ]$, respectively.

With the commutator (\ref{lax_commutator}) in hand, it is now straightforward to compute the curvature two-form associated with the would-be Lax connection (\ref{lax_connection}):
\begin{align}
    d_{\CC{L}} \CC{L} &= \partial_+ \CC{L}_- - \partial_- \CC{L}_+ + [ \CC{L}_+ , \CC{L}_- ] \nonumber \\
    &\deq \frac{1}{1 - z^2} \left( \partial_+ \left( j_- - z \CC{J}_- \right) - \partial_- \left( j_+ + z \CC{J}_+ \right) + [ j_+ , j_- ] \right) \nonumber \\
    &= \frac{1}{1 - z^2} \left( \partial_+ j_- - \partial_- j_+ + [ j_+ , j_- ] - z \left( \partial_+ \CC{J}_- + \partial_- \CC{J}_+ \right) \right) \, .
\end{align}
Therefore, the condition that $d_{\CC{L}} \CC{L} \deq 0 $ for any $z \in \mathbb{C}$ is equivalent to the two conditions
\begin{align}\label{MC_and_eom}
    \partial_+ j_- - \partial_- j_+ + [ j_+ , j_- ] \deq 0 \, , \qquad \partial_+ \CC{J}_- + \partial_- \CC{J}_+ \deq 0 \, .
\end{align}
The first of the equations in (\ref{MC_and_eom}) is the Maurer-Cartan identity, which holds regardless of whether the auxiliary field equation of motion is satisfied. The second of these equations is the physical field equation of motion (\ref{afsm_eom_dot}) for the model.

We conclude that every model in the family (\ref{deformed_family}) obeys the definition (\ref{aux_weak_integrability}). This means that every higher-spin auxiliary field sigma model admits a Lax representation for its equation of motion, or equivalently is weakly integrable.

Again, we stress that this technical details of this argument are almost identical to those of \cite{Ferko:2024ali}, and rely only on the two implications (\ref{imp_one}) and (\ref{imp_two}) of the auxiliary field equations of motion. The novelty in the present analysis is that the same set of implications hold for arbitrary interaction functions $E ( \nu_2 , \ldots , \nu_N )$ of several variables, and thus that the larger family of auxiliary field sigma models which include these higher-spin couplings also admit zero-curvature representations for their equations of motion.

\subsection{Construction of Conserved Charges}\label{sec:charge_review}

We have now seen that the class of higher-spin auxiliary field sigma models is weakly integrable, in the sense that its equations of motion can be represented as the flatness of a Lax connection for every value of a spectral parameter. Although it is useful, the existence of such a zero-curvature representation is not, by itself, sufficient to establish classical integrability of a model. To show that a field theory is integrable, one must demonstrate the existence of an infinite set of independent and Poisson-commuting conserved charges in the model. The Lax representation developed above allows us to construct such an infinite set of conserved charges, but we are not guaranteed that an infinite subset of these charges will be in involution without making further assumptions about the Lax connection.

It is instructive to briefly review the construction of conserved charges from the Lax connection for a $2d$ field theory; we follow the pedagogical lecture notes \cite{Driezen:2021cpd} (see also \cite{Torrielli:2016ufi} or the book \cite{Babelon:2003qtg} for further details), to which we refer the reader for more details. Given a Lax connection $\CC{L}$, one first defines the transport matrix as the path-ordered exponential
\begin{align}\label{transport_matrix}
    T ( y, x ; z ) = \overleftarrow{\mathrm{Pexp}} \left( - \int_x^y d \sigma \, \mathfrak{L}_\sigma ( z ) \right) \, .
\end{align}
One can then take suitable limits of the endpoints of integration $x$ and $y$ to define the monodromy matrix. The appropriate limits depend on the topology of the spatial manifold on which the field theory is defined. As we mentioned above, we are interested in the cases where $\Sigma = \mathbb{R}^{1,1}$ for theories on the plane, or $\Sigma = \mathbb{R} \times S^1$ for models defined on a cylinder. In the case of the plane, the monodromy matrix is given by
\begin{align}
    T ( \infty, - \infty ; z ) = \overleftarrow{\mathrm{Pexp}} \left( - \int_{\infty}^{\infty} d \sigma \, \mathfrak{L}_\sigma ( z ) \right) \, ,
\end{align}
and given suitable fall-off conditions on the fields, any power of this monodromy matrix is independent of the time coordinate $\tau$:
\begin{align}
    \partial_\tau \left( T ( \infty, - \infty ; z )^n \right) = 0 \, .
\end{align}
For theories on the cylinder with identification $\sigma \sim \sigma + 2 \pi$, we instead take $x = 0$ and $y = 2 \pi$ to define the monodromy matrix as
\begin{align}
    T ( 2 \pi , 0 ; z ) = \overleftarrow{\mathrm{Pexp}} \left( - \int_{0}^{2 \pi} d \sigma \, \mathfrak{L}_\sigma ( z ) \right) \, ,
\end{align}
and in this case the trace of any power of the monodromy matrix is conserved,
\begin{align}
    \partial_\tau \left( \Tr \left(T ( 2 \pi , 0 ; z )^n \right) \right) = 0 \, .
\end{align}
These conservation conditions follow from the flatness of the Lax connection, which is equivalent to the equations of motion, and therefore they hold only on-shell.

To handle both cases of the plane and cylinder simultaneously, it is convenient to define
\begin{align}
    T ( z ) = \begin{cases} T ( \infty, - \infty ; z ) \, , &\quad \Sigma = \mathbb{R}^{1,1} \\ \Tr \left( T ( 2 \pi , 0 ; z ) \right) \, ,  &\quad \Sigma = S^1 \times \mathbb{R}  \end{cases} \, .
\end{align}
Given any point $z_0$ about which $T ( z )$ is analytic, one can then perform a Taylor expansion
\begin{align}\label{Q_from_T}
    T ( z ) = \sum_{n = 0}^{\infty} Q_n \left( z - z_0 \right)^n \, .
\end{align}
It then follows that each of the expansion coefficients $Q_n$ is independent of time, which furnishes us with an infinite set of conserved charges. However, as we mentioned above, it is not automatic that any infinite subset of these charges are mutually Poisson-commuting.

However, there are several additional assumptions which are known to be sufficient to guarantee the existence of such an infinite set of charges in involution. To discuss these, we must first introduce some notation. Given an element $X \in \mathfrak{g}$, define the combinations
\begin{align}\label{doubled_X}
    X_1 = X \otimes 1 \, , \qquad X_2 = 1 \otimes X \, ,
\end{align}
where $1$ is the identity matrix. Because our charges $Q_n$ are constructed from the monodromy matrix, the Poisson bracket of two charges is related to the bracket of two copies of the monodromy matrix, which is itself determined by the Poisson bracket of the spatial component of the Lax connection (since it is this component which appears in the path-ordered exponential (\ref{transport_matrix})). Therefore, additional assumptions about the Poisson bracket of the spatial component of the Lax connection will imply corresponding statements about Poisson brackets of conserved charges.

One sufficient condition for involution of the conserved charges, due to Sklyanin \cite{Sklyanin:1980ij}, is the assumption that the bracket of the Lax connection takes the form
\begin{align}\label{ultralocal}
    \left\{ \CC{L}_{\sigma, 1} ( \sigma, z ) \, , \, \CC{L}_{\sigma, 2} ( \sigma', z' ) \right\} = [ \CC{L}_{\sigma, 1} ( \sigma, z ) + \CC{L}_{\sigma, 2} ( \sigma, z' ) , r_{12} (z , z' ) ] \delta ( \sigma - \sigma' ) \, ,
\end{align}
where $r_{12}$ is called the $r$-matrix. Suppose we assume that the Poisson bracket of the Lax connection takes the form (\ref{ultralocal}), where the $r$-matrix is antisymmetric, $r_{12} ( z, z' ) = - r_{21} ( z', z )$, independent of the fields, and satisfies the classical Yang-Baxter equation,
\begin{align}\label{cybe}
    0 = [ r_{12} ( z_1, z_2 ) \, , \, r_{13} ( z_1, z_3 ) ] + [ r_{12} ( z_1, z_2 ) \, , \, r_{23} ( z_2, z_3 ) ] + [ r_{32} ( z_3, z_2 ) \, , \, r_{13} ( z_1, z_3 ) ]  \, .
\end{align}
(This collection of assumptions is merely sufficient, rather than necessary, and can be relaxed as we shall see in a moment.) Then it follows from these conditions that the charges $Q_n$ obtained from the monodromy matrix as in equation (\ref{Q_from_T}) are all in involution,
\begin{align}
    \left\{ Q_n , Q_m \right\} = 0 \, .
\end{align}
Therefore, assuming the existence of a Lax connection whose bracket takes the so-called ``ultralocal'' form (\ref{ultralocal}), one can conclude that the associated model possesses an infinite set of conserved charges in involution. This property is sometimes called \emph{strong integrability}.

However, even for the undeformed principal chiral model, the above assumptions are too strong because the Poisson bracket of the Lax connection does not take the required ultralocal form (\ref{ultralocal}). Rather, one finds a structure which includes an additional ``non-ultralocal'' term involving a derivative of a delta function, and where the $r$-matrix associated to the ultralocal term is not antisymmetric. 

In situations where the Poisson bracket of the Lax exhibits such a modified structure, a generalization of the above argument due to Maillet \cite{MAILLET198654,MAILLET1986401} can be applied to conclude that there still exists an infinite set of Poisson-commuting charges. More precisely, suppose that the Poisson bracket of the spatial component of the Lax connection can be written as
\begin{align}\label{maillet_bracket}
    \left\{ \CC{L}_{\sigma, 1} ( \sigma, z ) \, , \, \CC{L}_{\sigma, 2} ( \sigma', z' ) \right\} &= [ r_{12} ( z, z' ) \, , \, \CC{L}_{\sigma, 1} ( \sigma, z ) ] \delta ( \sigma - \sigma' ) - [ r_{21} ( z', z ) \, , \, \CC{L}_{\sigma, 2} ( \sigma, z' ) ] \delta ( \sigma - \sigma' ) \nonumber \\
    &\qquad - s_{12}( z, z' ) \partial_\sigma \delta ( \sigma - \sigma' ) \, .
\end{align}
Now the $r$-matrix $r_{12} ( z, z' )$ is still taken to satisfy (\ref{cybe}) but is no longer assumed to be antisymmetric, and indeed the symmetric part of this $r$-matrix is called
\begin{align}
    s_{12} ( z , z' ) = r_{12} ( z, z' ) + r_{21} ( z' , z ) \, .
\end{align}
We refer to the expression (\ref{maillet_bracket}) as the Maillet form of the Poisson brackets, or as the $r/s$ form due to the presence of the two matrices $r$ and $s$. When the $r$-matrix is purely antisymmetric, we see that $s_{12} = 0$, so the non-ultralocal term vanishes and the expression (\ref{maillet_bracket}) reduces to the previous ultralocal form of the bracket (\ref{ultralocal}) considered by Sklyanin.

A new subtlety emerges when one considers theories for which the Poisson bracket of the Lax takes the non-ultralocal form (\ref{maillet_bracket}). In this case, there is a discontinuity in the bracket of the monodromy matrices $\{ T_1 ( x', y' ; z ) , T_2 ( x, y ; z' ) \}$ when $x' \to x$, $y' \to y$. As a result, if one takes a naive coincident-point limit, it appears that the Jacobi identity is violated (even when the $r$-matrix satisfies the classical Yang-Baxter equation).

The resolution proposed by Maillet is to perform a regularization procedure which involves symmetrizing nested Poisson brackets and coincident-point limits, as described in \cite{MAILLET1986401} and nicely reviewed in \cite{Dorey:2006mx,Vicedo:2008ryn}. Using this prescription, the regularized Poisson brackets in fact \emph{do} satisfy the Jacobi identity. Using these regularized brackets to define the coincident-point limits of the Poisson brackets of the monodromy matrices, one then finds
\begin{align}
    \Tr_{12} \left\{ \Tr_1 \left( T_1 ( y, x ; z )^n \right) \, , \, \Tr_2 \left( T_2 ( y, x ; z' )^m \right) \right\} = 0 \, ,
\end{align}
which likewise allows us to conclude that
\begin{align}
    \left\{ Q_n, Q_m \right\} = 0 \, ,
\end{align}
for the charges $Q_n$ constructed from the monodromy matrix in (\ref{Q_from_T}).

For this reason, constructing $r$ and $s$ matrices such that the Poisson bracket of the Lax connection for a given theory takes the form (\ref{maillet_bracket}) also constitutes a proof of ``strong integrability'' of the model, in the sense that it implies the existence of an infinite set of conserved charges in involution. We also note that the argument of Maillet applies equally well to theories with constraints, for which the Poisson brackets should be replaced with the appropriate Dirac brackets, which will be relevant for auxiliary field sigma models.

\subsection{Maillet $r/s$ Matrices}\label{sec:maillet}

Having reviewed the procedure for constructing an infinite set of Poisson-commuting conserved charges, assuming that the Lax connection for the theory obeys (\ref{maillet_bracket}), let us now carry out the construction of the Maillet $r/s$ matrices for the family of auxiliary field sigma models (\ref{deformed_family}). We first develop some aspects of the canonical structure of these theories.

\subsubsection*{\ul{\it Canonical structure of AFSM}}

Recall that our conventions for light-cone coordinates are such that, for any vector $V_\alpha$,
\begin{align}
    V_{\pm} = V_\tau \pm V_\sigma \, , \qquad V^{\pm} = \frac{1}{2} \left( V^\tau \pm V^\sigma \right) \, ,
\end{align}
Therefore in coordinates $(\tau, \sigma)$, the Lagrangian (\ref{deformed_family}) is
\begin{align}\label{afsm_lag_for_canonical}
    \mathcal{L}_{\text{AFSM}} = \frac{1}{2} \tr \left( j_\tau^2 - j_\sigma^2 \right) + \tr \left( v_\tau^2 - v_\sigma^2 \right) + 2 \tr \left( j_\tau v_\tau - j_\sigma v_\sigma \right)+ E \left( \nu_2 , \ldots , \nu_N \right) \, .
\end{align}
As in section \ref{sec:models_defn}, we choose local coordinates $\phi^\mu$ on the Lie group $G$, and we convert between middle Greek (e.g. $\mu$, $\nu$) target-space indices and early Greek (e.g. $\alpha$, $\beta$) worldsheet indices using the pull-back map (\ref{pullback}). The Maurer-Cartan form $j_\alpha$ can be written in terms of target-space coordinates as
\begin{align}\label{j_to_target}
    j_\alpha = \partial_\alpha \phi^\mu j^A_\mu T_A \, ,
\end{align}
whereas we expand the auxiliary fields as
\begin{align}\label{v_expansion}
    v_\alpha = v_\alpha^A T_A \, .
\end{align}
Note that the auxiliary fields are naturally defined as Lie algebra valued forms on the worldsheet $\Sigma$, while the Maurer-Cartan form is intrinsically defined in target space and thus must be pulled back to the worldsheet.

Using the expansions (\ref{j_to_target}) and (\ref{v_expansion}), the Lagrangian (\ref{afsm_lag_for_canonical}) becomes
\begin{align}
    \mathcal{L}_{\text{AFSM}} &= \frac{1}{2} \tr \left( \partial_\tau \phi^\mu j_\mu^A T_A \partial_\tau \phi^\nu j_\nu^B T_B - \partial_\sigma \phi^\mu j_\mu^A T_A \partial_\sigma \phi^\nu j_\nu^B T_B \right) + \tr \left( v_\tau^A v_\tau^B T_A T_B - v_\sigma^A v_\sigma^B T_A T_B \right) \nonumber \\
    &\qquad + 2 \tr \left( \partial_\tau \phi^\mu j_\mu^A T_A v_\tau^B T_B - \partial_\sigma \phi^\mu j_\mu^A T_A v_\sigma^B T_B \right) + E ( \nu_2 , \ldots, \nu_N ) \, ,
\end{align}
or in terms of the Killing-Cartan form $\gamma_{AB} = \tr ( T_A T_B )$,
\begin{align}
    \mathcal{L}_{\text{AFSM}} &= \gamma_{AB} \left( \frac{1}{2} \left( \partial_\tau \phi^\mu \partial_\tau \phi^\nu - \partial_\sigma \phi^\mu \partial_\sigma \phi^\nu \right) j_\mu^A  j_\nu^B + v_\tau^A v_\tau^B - v_\sigma^A v_\sigma^B + 2 \partial_\tau \phi^\mu j_\mu^A v_\tau^B - \partial_\sigma \phi^\mu j_\mu^A v_\sigma^B \right) \nonumber \\
    &\qquad + E \left( \nu_2 , \ldots , \nu_N \right) \, .
\end{align}
The momentum conjugate to the target-space coordinate $\phi^\mu$ is
\begin{align}
    \pi_\mu = \frac{\partial \mathcal{L}}{\partial ( \partial_\tau \phi^\mu ) } = \partial_\tau \phi^\nu j_\mu^A j_\nu^B \gamma_{AB} + 2 j_\mu^A v_\tau^B \gamma_{AB} \, ,
\end{align}
or in terms of the pull-back $j_\tau^B = ( \partial_\tau \phi^\nu ) j_\nu^B$,
\begin{align}\label{pi_from_J}
    \pi_\mu = j_\mu^A \left( j_\tau^B + 2 v_\tau^B \right) \gamma_{AB} \, .
\end{align}
We now recognize the same combination of fields $\fJ_\alpha = - ( j_\alpha + 2 v_\alpha )$ which appeared in equation (\ref{frakJ_defn}). In terms of $\fJ_\tau$, the conjugate momentum can be written
\begin{align}
    \pi_\mu = - j_\mu^A \fJ_\tau^B  \gamma_{AB} \, .
\end{align}
It is also convenient to define the field $j^\mu_A$, which satisfies
\begin{align}
    j^\mu_A j^B_\mu = \tensor{\delta}{_A^B} \, .
\end{align}
Since the Maurer-Cartan form plays the role of the vielbein $e_\mu^A$ on a Lie group, where $\mu$ is a ``curved'' index and $A$ is a ``flat'' (tangent space) index, this new field $j^\mu_A$ can be interpreted as the inverse vielbein. In terms of this inverse field, we can solve for $\fJ_\tau^B$ as
\begin{align}\label{J_to_pi}
    \fJ_\tau^B = - j^\mu_A \pi_\mu \gamma^{AB} \, .
\end{align}
Since $\pi_\mu$ are the momenta which are canonically conjugate to the fundamental fields $\phi^\mu$ of the model, the fundamental Poisson brackets for this theory are simply
\begin{align}\label{fundamental_poisson}
    \left\{ \pi_\mu ( \sigma ) , \phi^\nu ( \sigma' ) \right\} = \tensor{\delta}{^\nu_\mu} \delta ( \sigma - \sigma' ) \, ,
\end{align}
which holds at equal times $\tau$. The Poisson bracket of any other quantities $A$, $B$ is
\begin{align}\label{bracket_defn}
    \left\{ A ( \sigma ) , B ( \sigma' ) \right\} &= \int d \sigma'' \left( \frac{\delta A ( \sigma )}{\delta \phi^\rho ( \sigma'' )} \frac{\delta B ( \sigma' )}{\delta \pi_\rho ( \sigma'' ) } - \frac{ \delta B ( \sigma' )}{\delta \phi^\rho ( \sigma'' )} \frac{\delta A ( \sigma )}{\delta \pi_\rho ( \sigma'' ) } \right) \, .
\end{align}
The Hamiltonian of the AFSM is given by the Legendre transform
\begin{align}\label{H_AFSM_one}
    \mathcal{H}_{\text{AFSM}} &= \pi_\mu \dot{\phi}^\mu - \mathcal{L}_{\text{AFSM}} \nonumber \\
    &= j_\tau^A \left( j_{\tau, A} + 2 v_{\tau, A} \right) - \Bigg( \frac{1}{2} j_\tau^A j_{\tau, A} - \frac{1}{2} j_{\sigma}^A j_{\sigma, A} + v_\tau^A v_{\tau, A} - v_{\sigma}^A v_{\sigma, A} + 2 j_\tau^A v_{\tau, A} \nonumber \\
    &\qquad \qquad \qquad \qquad \qquad \qquad \qquad - 2 v_\sigma^A v_{\sigma, A} + E ( \nu_2 , \ldots , \nu_N ) \Bigg) \nonumber \\
    &= \frac{1}{2} j_\tau^A j_{\tau, A} + \frac{1}{2} j_\sigma^A j_{\sigma, A} - v_\tau^A v_{\tau, A} + v_\sigma^A v_{\sigma, A} + 2 j_\sigma^A v_{\sigma, A} - E ( \nu_2 , \ldots , \nu_N ) \, .
\end{align}
We must now express everything in terms of the canonical momentum. Note that, by virtue of the definition (\ref{frakJ_defn}) of $\CC{J}_\alpha$, one has
\begin{align}\label{jj_to_JJ}
    \frac{1}{2} j_\tau^A j_{\tau, A} = \frac{1}{2} \fJ_\tau^A \fJ_{\tau,A} - 2 j_\tau^A v_{\tau, A} - 2 v_\tau^A v_{\tau, A} \, ,
\end{align}
and substituting (\ref{jj_to_JJ}) into (\ref{H_AFSM_one}) gives
\begin{align}
    \mathcal{H}_{\text{AFSM}} = \frac{1}{2} \fJ_\tau^A \fJ_{\tau, A} + \frac{1}{2} j_\sigma^A j_{\sigma, A}  + v_\tau^A v_{\tau, A} + v_\sigma^A v_{\sigma, A} + 2 \left(  j_\sigma^A v_{\sigma, A} + \fJ_\tau^A v_{\tau, A} \right) - E ( \nu_2 , \ldots , \nu_N )  \, .
\end{align}
We can then express this in terms of momenta using the inverse relation (\ref{J_to_pi}) to find
\begin{align}\label{final_afsm_ham}
    \mathcal{H}_{\text{AFSM}} &= \frac{1}{2} j^\mu_A \pi_\mu j^{\nu, A} \pi_\nu + \frac{1}{2} j_\sigma^A j_{\sigma, A}  + v_\tau^A v_{\tau, A} + v_\sigma^A v_{\sigma, A} + 2 \left( j_\sigma^A v_{\sigma, A} - j_\mu^A \pi^\mu v_{\tau, A} \right) - E ( \nu_2 , \ldots , \nu_N ) \, .
\end{align}
This completes the determination of the AFSM Hamiltonian in terms of the fields $\phi^\mu$ and conjugate momenta $\pi_\mu$. Note that $j^\mu_A = j^\mu_A ( \phi )$ depend on the coordinates $\phi^\mu$.

\subsubsection*{\ul{\it Constraints and Dirac brackets}}

Let us make a brief technical aside about the Hamiltonian structure of our auxiliary field sigma models. In addition to the canonical momentum $\pi_\mu$ associated with the target-space coordinates $\phi^\mu$, there is also a conjugate momentum
\begin{align}\label{auxiliary_momentum}
    \mathfrak{p}_\alpha = \frac{\partial \mathcal{L}}{\partial ( \partial_\tau v_\alpha ) } = 0 \, ,
\end{align}
which vanishes because the Lagrangian does not depend on derivatives of auxiliary fields.

Ordinarily, one would pass from the configuration space of a field theory to the phase space by inverting to express the time derivatives of fields in terms of canonical momenta. Here this process fails, since the naive set of phase space coordinates $( \phi^\mu , \pi_\mu , v^\alpha , \mathfrak{p}_\alpha )$ cannot be expressed in terms of the fields and their time derivatives; instead the phase space coordinates are linearly dependent, since $\mathfrak{p}_\alpha = 0$.

This is a ubiquitous feature of classical mechanical systems subject to constraints, whose solution is well-known: such a model should be treated using the formalism for constrained quantization due to Dirac \cite{Dirac:1964:LQM}. See also \cite{hanson1976constrained,henneaux1992quantization,gitman2012quantization,prokhorov2011hamiltonian,wipf_constraints} for pedagogical discussions of the Dirac method in various settings, including applications to gauge theories.

Let us briefly review the steps of the Dirac procedure. Consider a Hamiltonian system which is subject to a collection of constraint functions $F^I$, $I = 1 , \ldots, M$, which depend on the phase space coordinates, and which must vanish for points on the phase space manifold which satisfy the constraints. It is convenient to introduce the symbol $\approx$ to indicate that two quantities are equal when the constraints are satisfied. For instance, in our case there are two constraints $F_\tau = \mathfrak{p}_\tau$ and $F_\sigma = \mathfrak{p}_\sigma$ which are constrained to vanish due to (\ref{auxiliary_momentum}), so
\begin{align}\label{constraints_vanishing}
    F_\alpha \approx 0 \, .
\end{align}
The subspace of the phase space manifold on which all of the constraint functions $F^I$ vanish is called the constraint surface. The functions $F^I$ are referred to as primary constraints, and a quantity which equals zero on the constraint surface is said to ``vanish weakly.''

It is worth pointing out that, in order to obtain consistent results within the Dirac formalism, no ``weak equations'' can be used in intermediate steps of calculations, such as inside of Poisson brackets. That is, one may use only ``strong equations'' (those that hold identically, without using the constraints) before evaluating derivatives, computing Poisson brackets, etc.. Only at the end of a calculation can one use the constraint equations $F_\alpha \approx 0$, which projects the resulting final expressions onto the constraint surface.

Note that the relation $\approx$ is related to, but distinct from, the dotted equality symbol $\deq$. The latter is used in the Lagrangian formulation, and indicates that two quantities agree when the $v_\alpha$ equation of motion is satisfied, whereas $\approx$ is used in the Hamiltonian formalism and indicates that the conjugate momentum $\mathfrak{p}_\alpha$ is vanishing. The relationship between the two is that the Hamiltonian equations of motion defined using the so-called primary Hamiltonian, which is defined by
\begin{align}\label{primary_hamiltonian}
    H_p = H + \int d^d x \, \lambda_I ( x ) F^I ( x ) \, ,
\end{align}
are the same as the Lagrangian equations of motion for the original system. Here the Hamiltonian, $H = \int d^d x \, \mathcal{H}$, is the integral of the ordinary Hamiltonian density, and $d$ is the number of spatial dimensions (so that $d+1$ is the total spacetime dimension). The objects $\lambda_I$ play the role of additional Lagrange multiplier fields, which do not have any associated canonical momenta.

The $F^I$ are referred to as ``primary constraints'' since they are the original constraint conditions that are direct consequences of the initial Hamiltonian $H$. In addition, there can be additional ``secondary constraints'' which are generated by time evolution of the primary constraints. More precisely, the time evolution of any constraint function $F^I$ is given by
\begin{align}\label{FI_dot}
    \dot{F}^I \approx \left\{ F^I , H_p \right\} \, .
\end{align}
If the quantity (\ref{FI_dot}) does not vanish weakly, then time evolution maps points on the constraint surface onto points outside the constraint surface, which is inconsistent. For each function $F^I$ such that (\ref{FI_dot}) is non-zero, we obtain an additional constraint function
\begin{align}
    \widehat{F}^I = \left\{ F^I , H_p \right\} \, .
\end{align}
One must then iterate this procedure, appending the secondary constraint to the original list of primary constraints and checking whether the new list of constraints is closed under time evolution; if not, one incorporates further tertiary constraints, and so on. Eventually this process terminates and one arrives at a complete list of constraints.

For the AFSM Hamiltonian, there will be secondary constraints arising from time evolution of the constraint (\ref{auxiliary_momentum}). As an example,
\begin{align}
    \left\{ H_{\text{AFSM}} \, , \, \mathfrak{p}_{\tau, A} \right\} = v_{\tau, A} - 2 j^\mu_A \pi_\mu - \frac{\partial E}{\partial v_\tau^A} \, ,
\end{align}
with a similar expression for $\mathfrak{p}_\sigma^A$.

Now, suppose that all of the constraints have been enumerated and collected into a list $F^I$, where we no longer distinguish between primary constraints, secondary constraints, and so on. For each $I$, we say that $F^I$ is a first-class constraint if
\begin{align}
    \left\{ F^I , F^J \right\} = 0 \text{ for all } J \, .
\end{align}
Otherwise, if $F^I$ has a non-vanishing Poisson bracket with at least one other constraint function $F^J$, we will decorate this index $I$ with a tilde, writing $F^{\widetilde{I}}$, and refer to it as a second-class constraint. For systems with second-class constraints, such as the AFSM, one requires a generalization of the Poisson bracket with the property that the bracket of a second-class constraint with any other quantity must be zero. This generalization is provided by the Dirac bracket, which is defined by
\begin{align}\label{dirac_brackets}
    \left\{ f , g \right\}_{\text{D}} = \left\{ f, g \right\} - \sum_{\widetilde{I}, \widetilde{J}} \big\{ f, F^{\widetilde{I}} \big\} \; \left( D^{-1} \right)_{\widetilde{I} \widetilde{J}} \; \big\{ F^{\widetilde{J}} , g \big\} \, ,
\end{align}
where sum runs only over those values of the indices $\widetilde{I}$, $\widetilde{J}$ for which the corresponding constraints $F^{\widetilde{I}}$, $F^{\widetilde{J}}$ are second-class, and where the matrix $D^{\widetilde{I} \widetilde{J}}$ is defined by
\begin{align}\label{M_matrix_dirac}
    D^{\widetilde{I} \widetilde{J}} = \big\{ F^{\widetilde{I}} , \widetilde{F}^J \big\} \, .
\end{align}
Any brackets that do not carry a subscript $D$, which indicates the Dirac bracket, are assumed to be ordinary Poisson brackets. This construction is possible because the matrix $D^{\widetilde{I} \widetilde{J}}$ defined in (\ref{M_matrix_dirac}) is always invertible, and the Dirac bracket $\left\{ \, \cdot \, , \, \cdot \, \right\}_D$ constructed in this way enjoys all of the same algebraic properties as the ordinary Poisson bracket, including antisymmetry, bilinearity, the Jacobi identity, and the Leibniz rule. 

In summary, because the auxiliary field sigma model is a constrained Hamiltonian system with second-class constraints, the Dirac bracket (\ref{dirac_brackets}) is the appropriate machinery to use in this setting. However, by construction, the Dirac brackets coincide with the ordinary Poisson brackets on the constraint surface, when all $F^I$ are equal to zero. The distinction between the two structures is that, when one uses Dirac brackets, the constraint equations can be used as ``strong equations'' in intermediate steps of calculations, whereas for the Poisson brackets, these constraints must be viewed as ``weak equations'' that \emph{cannot} be used in intermediate steps, as they would lead to inconsistent results.

The upshot of this discussion is that, if one is computing brackets of quantities for which the constraint equations play no role, then one is free to use either the Dirac or Poisson brackets, and the results will coincide. As it turns out, the brackets which appear in the construction of the Maillet $r/s$ matrices do not depend on the constraints at all: they only involve the canonical fields $\phi^\mu$ and their canonical momenta $\pi_\mu$, and the auxiliary fields (or their conjugate momenta) enter only through the implicit dependence of $\pi_\mu$ on $v_\tau$ via equation (\ref{pi_from_J}), but not explicitly. Said differently, the auxiliary field equations of motion are not used anywhere in the computation of the brackets of the Lax connection.

For this reason, in the following argument, we will not distinguish between the Dirac and Poisson brackets, since we are guaranteed that they will produce identical results for the quantities of interest. We will sometimes use the generic term ``canonical bracket'' to refer to the Poisson or Dirac bracket, in settings where the two coincide.

\subsubsection*{\ul{\it Maillet form of canonical brackets}}

We now turn to the computation of the (Poisson or Dirac) bracket of the spatial component of the Lax connection with itself. This spatial component is given by
\begin{align}\label{spatial_lax}
    \fL_\sigma &= \frac{1}{2} \left( \fL_+ - \fL_- \right) \nonumber \\
    &= \frac{j_\sigma + z \fJ_\tau}{ 1 - z^2 } \, ,
\end{align}
so we are interested in computing
\begin{align}\label{lax_bracket_intermediate}
    \left\{ \fL_{\sigma, 1} ( \sigma, z ) , \fL_{\sigma, 2} ( \sigma', z' ) \right\} &= \frac{1}{(1 - z^2) ( 1 - z^{\prime 2} ) } \left( \left\{ j_{\sigma, 1} ( \sigma ) + z \fJ_{\tau, 1} ( \sigma ) , j_{\sigma, 2} ( \sigma' ) + z' \fJ_{\tau, 2} ( \sigma' ) \right\} \right) \nonumber \\
    &= \frac{1}{(1 - z^2) ( 1 - z^{\prime 2} ) } \Bigg( \left\{ j_{\sigma, 1} ( \sigma ) , j_{\sigma, 2} ( \sigma' ) \right\} + z \left\{ \fJ_{\tau, 1} ( \sigma ) , j_{\sigma, 2} ( \sigma' )  \right\}  \nonumber \\
    &\qquad + z' \left\{ j_{\sigma, 1} ( \sigma ) , \fJ_{\tau, 2} ( \sigma' ) \right\}  + z z' \left\{ \fJ_{\tau, 1} ( \sigma )  , \fJ_{\tau, 2} ( \sigma' ) \right\} \Bigg) \, .
\end{align}
The quantity (\ref{lax_bracket_intermediate}) depends on several other canonical brackets involving the fields $j_\sigma$ and $\fJ_\tau$. We have collected the details of the computations of these intermediate quantities in appendix \ref{app:pb}. Here we simply quote the results; at the end of the calculations one finds
\begin{align}\label{js_brackets_body}
    \{ \fJ_{\tau,1} ( \sigma ) , \fJ_{\tau,2} ( \sigma' ) \} &= [ \fJ_{\tau, 2} , C_{12} ]  \delta ( \sigma - \sigma' ) \, , \nonumber \\
    \left\{ \fJ_{\tau, 1} ( \sigma ) , j_{\sigma, 2} ( \sigma' ) \right\} &= [ j_{\sigma, 2}, C_{12} ] \delta ( \sigma - \sigma' ) - C_{12} \delta' ( \sigma - \sigma' ) \, , \nonumber \\
    \left\{ j_{\sigma, 1} ( \sigma ) , j_{\sigma, 2} ( \sigma' ) \right\} &= 0 \, .
\end{align}
Here we have introduced the Casimir $C_{12}$  defined as
\begin{align}\label{casimir_defn}
    C_{12} = \gamma^{AB} T_A \otimes T_B \, .
\end{align}
We now substitute the three expressions (\ref{js_brackets_body}) into the bracket (\ref{lax_bracket_intermediate}) and use the fact that
\begin{align}\label{j_FJ_swap}
    \left\{ j_{\sigma, 1} ( \sigma ) , \fJ_{\tau, 2} ( \sigma' ) \right\} = \left\{ \fJ_{\tau, 1} ( \sigma ) , j_{\sigma, 2} ( \sigma' ) \right\} \, ,
\end{align}
which is shown around equation (\ref{app_swap_j_FJ}), to obtain
\begin{align}\label{lax_pb_intermediate}
    \left\{ \fL_{\sigma, 1} ( \sigma, z ) , \fL_{\sigma, 2} ( \sigma', z' ) \right\} &= \frac{1}{(1 - z^2) ( 1 - z^{\prime 2} ) } \Bigg( ( z + z ' )  \left( [ j_{\sigma, 2} , C_{12} ] \delta ( \sigma - \sigma' ) - C_{12} \partial_\sigma \delta ( \sigma - \sigma' ) \right)   \nonumber \\
    &\qquad \qquad \qquad \qquad \qquad \qquad + z z' [ \fJ_{\tau, 2} , C_{12} ] \delta ( \sigma - \sigma' )  \Bigg) \, .
\end{align}
Collecting terms, we find that (\ref{lax_pb_intermediate}) can be written as
\begin{align}\label{final_lax_bracket}
    \left\{ \fL_{\sigma, 1} ( \sigma, z ) , \fL_{\sigma, 2} ( \sigma', z' ) \right\} = \frac{1}{(1 - z^2) ( 1 - z^{\prime 2} ) } &\Bigg( \left[ ( z + z ' ) j_{\sigma, 2} + z z' \fJ_{\tau, 2} , C_{12} \right] \delta ( \sigma - \sigma' ) \nonumber \\
    &\qquad - ( z + z' ) C_{12} \partial_\sigma \delta ( \sigma - \sigma' )  \Bigg) \, .
\end{align}
We would like to show that (\ref{final_lax_bracket}) takes the form (\ref{maillet_bracket}) for an appropriate choice of the $r$-matrix (and hence its symmetric part, the $s$-matrix). We make a standard ansatz for the $r$-matrix in terms of a twist function $\varphi ( z )$,
\begin{align}\label{r_from_twist}
    r_{12} ( z, z' ) &= \frac{C_{12}}{z - z'} \varphi^{-1} ( z' ) \, ,
\end{align}
which solves the classical Yang-Baxter equation for any choice of $\varphi ( z )$. Furthermore, we choose this $\varphi$ to match the twist function of the undeformed PCM,
\begin{align}\label{pcm_twist}
    \varphi ( z ) &= \frac{ z^2 - 1 }{z^2} \, . 
\end{align}
Since the transpose of the $r$ matrix is
\begin{align}
    r_{21} ( z', z ) &= \frac{C_{12}}{z' - z} \varphi^{-1} ( z ) \, , 
\end{align}
with this choice of twist function we find that the $s$-matrix is
\begin{align}
    s_{12} ( z , z' ) &= r_{12} ( z, z' ) + r_{21} ( z', z ) = \frac{C_{12} ( z + z' ) }{ ( 1 - z^2 ) ( 1 - z^{\prime 2} ) } \, .
\end{align}
Using this ansatz, we now compute the right side of equation (\ref{maillet_bracket}), which is the prescribed form of the Maillet bracket. This gives
\begin{align}\label{maillet_intermediate}
    &[ r_{12} ( z, z' ) \, , \, \CC{L}_{\sigma, 1} ( \sigma, z ) ] \delta ( \sigma - \sigma' ) - [ r_{21} ( z', z ) \, , \, \CC{L}_{\sigma, 2} ( \sigma, z' ) ] \delta ( \sigma - \sigma' )  - s_{12}( z, z' ) \partial_\sigma \delta ( \sigma - \sigma' ) \nonumber \\
    &\qquad = \frac{\varphi^{-1} ( z' ) }{z - z'} [ C_{12} , \fL_{\sigma, 1} ( \sigma, z ) ] \delta ( \sigma - \sigma' )  - \frac{\varphi^{-1}( z)}{z' - z} [ C_{12}, \fL_{\sigma, 2} ( \sigma, z' ) ] \delta ( \sigma - \sigma' ) \nonumber \\
    &\qquad \qquad  - \frac{C_{12} ( z + z' ) }{ ( 1 - z^2 ) ( 1 - z^{\prime 2} ) } \partial_\sigma \delta ( \sigma - \sigma' ) \nonumber \\
    &\qquad = \delta ( \sigma - \sigma' ) \Bigg( \frac{z^{\prime 2}}{ ( z^{\prime 2} - 1 ) ( z - z' ) } [ C_{12} , \frac{j_{\sigma, 1} + z \fJ_{\tau,1}}{1 - z^2} ]
    - \frac{z^2}{ ( z^2 - 1 ) (z' - z)} [ C_{12} , \frac{j_{\sigma, 2} + z \fJ_{\tau, 2} }{ 1 - z^{\prime 2}} ] \Bigg) 
     \nonumber \\
     &\qquad \qquad
    - \frac{C_{12} ( z + z' ) }{ ( 1 - z^2 ) ( 1 - z^{\prime 2} ) } \partial_\sigma \delta ( \sigma - \sigma' ) \, ,
\end{align}
where we have substituted the expression (\ref{lax_connection}) for the Lax connection of the auxiliary field sigma model. 
Equation (\ref{maillet_intermediate}) further simplifies to
\begin{align}
     &[ r_{12} ( z, z' ) \, , \, \CC{L}_{\sigma, 1} ( \sigma, z ) ] \delta ( \sigma - \sigma' ) - [ r_{21} ( z', z ) \, , \, \CC{L}_{\sigma, 2} ( \sigma, z' ) ] \delta ( \sigma - \sigma' )  - s_{12}( z, z' ) \partial_\sigma \delta ( \sigma - \sigma' ) \nonumber \\
     &\qquad = \frac{1}{( 1 - z^2 ) ( 1 - z^{\prime 2} ) } \Bigg( \left[ C_{12} , z z' \fJ_{\tau, 1} + ( z + z' ) j_{\sigma, 1} \right] \delta ( \sigma - \sigma' ) - ( z + z' ) C_{12} \partial_\sigma \delta ( \sigma - \sigma' ) \Bigg) \, .
\end{align}
Moreover, after using $[C_{12}, X_1] = - [ C_{12}, X_2 ] = [ X_2, C_{12} ]$, which is proved around equation (\ref{C12_X1_to_C12_X2}), one finds
\begin{align}\label{final_maillet_comparison}
    &[ r_{12} ( z, z' ) \, , \, \CC{L}_{\sigma, 1} ( \sigma, z ) ] \delta ( \sigma - \sigma' ) - [ r_{21} ( z', z ) \, , \, \CC{L}_{\sigma, 2} ( \sigma, z' ) ] \delta ( \sigma - \sigma' )  - s_{12}( z, z' ) \partial_\sigma \delta ( \sigma - \sigma' ) \nonumber \\
     &\qquad = \frac{1}{( 1 - z^2 ) ( 1 - z^{\prime 2} ) } \Bigg( \left[ z z' \fJ_{\tau, 2}  + ( z + z' ) j_{\sigma, 2 } , C_{12} \right] \delta ( \sigma - \sigma' )  - ( z + z' ) C_{12} \partial_\sigma \delta ( \sigma - \sigma' ) \Bigg) \, .
\end{align}
We now compare equation (\ref{final_maillet_comparison}) with the Poisson bracket (\ref{final_lax_bracket}) of the Lax connection computed above, and find that they agree. We conclude that the spatial component of the AFSM Lax connection obeys the condition (\ref{maillet_bracket}), where the $r$-matrix is given by equations (\ref{r_from_twist}) and (\ref{pcm_twist}), exactly as in the case of the principal chiral model.\footnote{Unlike the AFSM, many other integrable deformations of the PCM lead to modifications of the twist function $\varphi ( z )$. See \cite{Lacroix:2018njs} and, references therein, for examples and a comprehensive discussion.} By the construction reviewed in section \ref{sec:charge_review}, it follows that every member of the family (\ref{deformed_family}) of auxiliary field sigma models possesses an infinite set of conserved charges in involution. Remarkably, this result holds uniformly for every choice of interaction function $E$, and in all cases the $r$-matrix and twist function are identical (including the case $E = 0$, which is the PCM). The analysis of this section shows how the use of auxiliary fiends dramatically simplifies the analysis of integrability for general classes of higher-spin deformations of the PCM. We will now show how simplifications occur also for T-duality.

\section{T-Duality}\label{sec:t_duality}

T-duality is an equivalence between certain pairs of $2d$ sigma models, which was first discovered in the study of string propagation on target spacetimes which feature a torus \cite{Kikkawa:1984cp,Sakai:1985cs,Sathiapalan:1986zb}. 
For instance, the worldsheet theory of a closed string embedded in a spacetime which involves a circle $S^1_R$ of radius $R$ is equivalent to the theory of a closed string propagating on a target spacetime with a circle $S^1_{\widetilde{R}}$ of radius $\widetilde{R} = \frac{\alpha'}{R}$, where $\alpha' = \ell_s^2$ determines the string length. This duality interchanges string states with momentum along the circular dimension, whose energies scale as $\frac{1}{R}$, with states that wind the circle, and thus have energies that scale linearly with $R$. We will return to this example around equation (\ref{momentum_winding}).

T-duality for strings on spacetimes compactified on a torus is referred to as \emph{Abelian} T-duality, since the isometry group of a torus $T^n$ is the Abelian group $U ( 1 )^n$. This was first studied in the flat case \cite{Ginsparg:1986wr,Giveon:1988tt,Shapere:1988zv,Giveon:1989yf}, soon after in more general curved backgrounds \cite{Buscher:1985kb,Buscher:1987qj,Buscher:1987sk,Rocek:1991ps,Giveon:1991jj} and later on further extended to take into account the presence of extra fields beyond the metric, Kalb-Ramond two-form and dilaton \cite{Siegel:1993xq,Siegel:1993th,Bergshoeff:1995as,Hassan:1995je,Hassan:1999bv,Hassan:1999mm,Cvetic:1999zs,Kulik:2000nr,Bandos:2003bz,Benichou:2008it} (see also the reviews \cite{Giveon:1994fu,Alvarez:1994dn}). The study of T-duality in target spacetimes which possess \emph{non-Abelian} isometry groups was initiated in \cite{delaOssa:1992vci} (see also \cite{Nappi:1979ig,Fridling:1983ha,Fradkin:1984ai} for earlier studies) and immediately attracted a lot of interest \cite{Alvarez:1993qi,Gasperini:1993nz,Curtright:1994be,Sfetsos:1994vz,Alvarez:1994np,Alvarez:1994wj,Elitzur:1994ri,Gasperini:1994du,Lozano:1995jx}, as an extension of the gauging procedure considered in \cite{Rocek:1991ps,Giveon:1991jj}. Given a $2d$ sigma model on a background with group $G$ of isometries, T-duality can be understood as the gauging of a subgroup $H \subseteq G$, followed by a modification of the initial action by means of a Lagrange multiplier term which enforces the flatness of the gauge fields. Integrating out the multipliers allows one to set the gauge fields to be pure gauge,\footnote{\label{higher-genus-footnote}While this is true for topologically trivial worldsheets, issues may arise in the case of higher genus surfaces. In the Abelian setting these have been overcome by recasting T-duality as a canonical transformation \cite{Alvarez:1994wj,Curtright:1994be,Lozano:1995jx}, while in the non-Abelian case the problem still represents an open question.} so as to recover the initial model, while integrating out the gauge fields produces a new sigma model. Upon gauge fixing the degrees of freedom of the initial model, the T-dual is then described in terms of Lagrange multipliers, which play the role of dual coordinates.

The case of non-Abelian T-duality is considerably more subtle than its Abelian counterpart. Whereas Abelian T-duality is an exact equivalence between conformal field theories at the quantum level \cite{Rocek:1991ps,Giveon:1991jj}, there is no conclusive argument which demonstrates that its non-Abelian cousin represents a genuine quantum duality \cite{Giveon:1993ai}. One technical challenge is understanding the effects of worldsheet quantum corrections, organized as an expansion in $\alpha'$. These are also present in the context of Abelian T-duality \cite{Tseytlin:1991wr,Panvel:1992he,Bergshoeff:1995cg,Haagensen:1997er,Kaloper:1997ux,Jack:1999av,Parsons:1999ze} and in the non-Abelian setting, or more generally for Poisson-Lie T-duality, have been studied in \cite{Borsato:2020wwk,Hassler:2020tvz,Codina:2020yma}.
Another issue is the corrections due to finite string coupling $g_s$, which remain an open problem in the non-Abelian setting, as mentioned in footnote \ref{higher-genus-footnote}.  Overall, except in certain limiting cases -- such as holographic scenarios at large $N$, where these corrections can be ignored -- the higher-order terms in $\alpha'$, $g_s$ are not sufficiently well-understood to decisively establish whether non-Abelian T-duality survives in the quantum theory. As an additional peculiarity, the non-Abelian setting is affected by a systematic loss of symmetries, since only the isometries which commute with the gauged ones survive the procedure and remain local in the T-dual model, while the others become generically delocalised \cite{Giveon:1993ai,Plauschinn:2013wta,Bugden:2018pzv}. Due to these features, the non-Abelian dualisation procedure is still nowadays regarded as a relation between different theories.

These unusual features have not diminished the interest in understanding non-Abelian T-duality; to the contrary, it has fostered the development of various related lines of investigation. As a prime example, the search for a formulation of sigma models exhibiting manifest duality invariance prompted a re-evaluation of the role of the isometries, leading to novel frameworks such as Poisson-Lie T-duality \cite{Klimcik:1995ux,Klimcik:1995jn,Klimcik:1995dy}, Double Field Theory \cite{Hull:2009mi,Hohm:2010pp} and Exceptional Field Theory \cite{Berman:2010is,Hohm:2013vpa}, as well as to recent proposals for the generalisation of the gauging procedure, known as non-isometric T-duality \cite{Chatzistavrakidis:2015lga,Chatzistavrakidis:2016jci,Chatzistavrakidis:2016jfz,Bouwknegt:2017xfi} and spherical T-duality \cite{Bouwknegt:2014oka,Bouwknegt:2014ima,Bouwknegt:2015eta} (see also \cite{Demulder:2019bha,Thompson:2019ipl,Aldazabal:2013sca,Hohm:2013bwa,Berman:2020tqn,Bugden:2018pzv} for a detailed list of references). The AdS/CFT correspondence represents another important setting where non-Abelian T-duality has been playing a prominent role, as starting from \cite{Sfetsos:2010uq,Lozano:2011kb} many novel examples of supergravity solutions and holographic pairs have been uncovered \cite{Lozano:2012au,Itsios:2012zv,Lozano:2014ata,Kelekci:2014ima,Lozano:2015cra,Lozano:2015bra,Lozano:2016wrs,Lozano:2017ole,Lozano:2019emq,Lozano:2019ywa,Lozano:2021rmk,Ramirez:2021tkd}. See also the reviews \cite{Lozano:2021xxs,Lozano:2022fqk,Ramirez:2022fkc} and \cite{Borsato:2021gma,Borsato:2021vfy} for recent works in related directions. In parallel to all this, after the generalisation of Abelian T-duality to the case of backgrounds enjoying commuting superisometries \cite{Berkovits:2008ic,Ricci:2007eq,Beisert:2008iq,Adam:2009kt,Hao:2009hw,Sfetsos:2010xa,Adam:2010hh,Bakhmatov:2010fp,Dekel:2011qw,Abbott:2015mla,Colgain:2016gdj}, the non-Abelian setting has been re-examined in a series of supersymmetric extensions charachterised by the use of BRST \cite{Grassi:2011zf} and super Poisson-Lie symmetries \cite{Eghbali:2009cp,Eghbali:2011su,Eghbali:2012sr,Eghbali:2013bda,Eghbali:2014coa,Eghbali:2020twu,Eghbali:2023sak,Eghbali:2024vxi}, Double Field Theory \cite{Astrakhantsev:2021rhj,Astrakhantsev:2022mfs,Butter:2023nxm} and the gauging procedure described above in terms of Lagrange multipliers \cite{Borsato:2016pas,Borsato:2017qsx,Borsato:2018idb,Bielli:2021hzg,Bielli:2022gmm,Bielli:2023bnh,Bielli:2024wru}. Last but not least, it has gradually been realised that T-duality exhibits deep connections with integrable deformations of sigma models, such as current-current \cite{Hassan:1992gi,Henningson:1992rn}, $\TT$ \cite{Araujo:2018rho,Sfondrini:2019smd,Apolo:2019zai,Blair:2020ops}, Yang-Baxter (or $\eta$) \cite{Klimcik:2002zj} and $\lambda$ deformations \cite{Sfetsos:2013wia}. In recent years, this has triggered many interesting results \cite{Klimcik:2015gba,Borsato:2016zcf,Osten:2016dvf,Hoare:2016wsk,Borsato:2016pas,Borsato:2017qsx,Borsato:2018idb,vanTongeren:2019dlq,Osten:2019ayq} -- we refer to \cite{Seibold:2020ouf,Hoare:2021dix,Borsato:2023dis} for additional references -- and with this in mind we shall here focus on the interplay between T-duality and higher-spin auxiliary field deformations of sigma models, extending the results of \cite{Bielli:2024khq} and at the same time providing more details on their derivation.

The PCM is a natural setting in which to study non-Abelian T-duality. Indeed, if one wishes to engineer a sigma model that describes propagation on a target spacetime which possesses a Lie group $G$ of isometries, the simplest theory to consider is the PCM for the field $g: \Sigma \to G$, since the group $G$ naturally acts on $g$ by left-multiplication (and, in fact, right multiplication) in a way which is a symmetry of the theory. However, lest these words of motivation give the impression that the PCM is merely a trivial ``toy example'' for non-Abelian T-duality, we should stress that the PCM also describes sectors of interesting string worldsheet models that play an important role for holography, such as the $S^3$ part of the worldsheet theory describing type IIB superstrings on $\mathrm{AdS}_3 \times S^3 \times T^4$, when the fermions have been set to zero. Non-Abelian T-duality of an $SU(2)$ subgroup has indeed been exploited to generate novel supergravity solutions starting from $\mathrm{AdS}_3 \times S^3 \times T^4$ and $\mathrm{AdS}_5 \times S^5$ backgrounds, where $SU(2)$ respectively arises as a subgroup of isometries of the PCM on $S^3$ and of the symmetric space $S^5$ \cite{Sfetsos:2010uq}. The construction was then extended to generic coset spaces, allowing for dualisation of the full isometry groups of $S^3$ and $S^5$, as well as the study of other backgrounds \cite{Lozano:2011kb}.

The discussion of this work is purely classical, so we will not concern ourselves with questions involving quantum corrections. Instead, we will view non-Abelian T-duality as a classical equivalence between the principal chiral model and a theory whose action is
\begin{align}\label{tdual_PCM}
    S_{\text{TD-PCM}} = - \frac{1}{2} \int d^2 \sigma \, \left( \partial_+ \Lambda_C \right) \left( M^{-1} \right)^{CD} \left( \partial_- \Lambda_D \right) \, ,
\end{align}
which we refer to as the T-dual principal chiral model, or T-dual PCM. In equation (\ref{tdual_PCM}), the fundamental degree of freedom is a Lie algebra valued field $\Lambda_A$, and $\left( M^{-1} \right)^{CD}$ denotes the inverse of a matrix $M_{CD} = M_{CD} \left( \Lambda \right)$ which will be introduced shortly in equation (\ref{MAB_defn}).

\subsection{Dualization of the AFSM}\label{sec:dualization}

In this section we carry out the non-Abelian T-dualization of our auxiliary field sigma models, following the procedure of Buscher \cite{Buscher:1987sk,Buscher:1987qj}. For an introduction, see the pedagogical lectures \cite{Thompson:2019ipl}, sections 1-2 of the thesis \cite{Bielli:2023bnh}, or the earlier review \cite{Alvarez:1994dn}.

We begin with the auxiliary field sigma model Lagrangian (\ref{deformed_family}), with a general interaction function $E$. Let us first make a brief remark about the symmetries of this theory. The undeformed PCM is invariant under a $G_L \times G_R$ global symmetry which acts on $g$ as
\begin{align}\label{g_trans}
    g \to g_L g g_R \, ,
\end{align}
where $g_L$ and $g_R$ are two arbitrary elements of the Lie group $G$. These are non-Abelian target-space isometries of the sigma model. In the deformed setting, we would like to determine how the auxiliary fields $v_{\pm}$ transform. A simple way to infer the correct $G_L \times G_R$ transformation rules of $v_\pm$
is to recall that the equations of motion for the AFSM can be written as the conservation of the modified current $\fJ_{\pm} = - \left( j_{\pm} + 2 v_{\pm} \right)$. In order for the sum $j_{\pm} + 2 v_{\pm}$ to transform covariantly under the $G_L \times G_R$ global symmetry, it must be the case that the auxiliary fields $v_{\pm}$ transform in the same way as the left-invariant Maurer-Cartan form $j_{\pm}$. Under the transformation (\ref{g_trans}) of $g$, we see that
\begin{align}
    j_{\pm} &= g^{-1} \partial_{\pm} g \longrightarrow \left( g_L g g_R \right)^{-1} \partial_{\pm} \left( g_L g g_R \right) \nonumber \\
    &= g_R^{-1} g^{-1} g_L^{-1} g_L \left( \partial_{\pm} g \right) g_R \nonumber \\
    &= g_R^{-1} j_{\pm} g_R \, ,
\end{align}
where we used that $g_L$, $g_R$ are constants. Therefore, for consistency, we also require that
\begin{align}
    v_{\pm} \to g_R^{-1} v_{\pm} g_R 
\end{align}
under the $G_L \times G_R$ action (\ref{g_trans}). With this transformation rule for the auxiliary fields, the entire AFSM action (\ref{deformed_family}) is invariant under $G_L \times G_R$. In particular, all of the interaction variables $\nu_k$ are singlets under this global symmetry.

Since we have now seen that the AFSM is invariant under $G_L \times G_R$, we may now proceed to gauge the $G_L$ global symmetry. In addition to the physical field $g$ and the auxiliary field $v_{\pm}$, we also introduce a Lie algebra valued gauge field $\omega_\alpha$ on the worldsheet. All derivatives $\partial_\alpha$ with respect to worldsheet coordinates, such as those appearing in the definition of the Maurer-Cartan form, $j_\alpha = g^{-1} \partial_\alpha g$, will be promoted to covariant derivatives,
\begin{align}
    D_\alpha = \partial_\alpha + \omega_\alpha \, .
\end{align}
We write the covariantized analogue of the Maurer-Cartan form as
\begin{align}
    j_\alpha^{(\omega)} = g^{-1} D_\alpha g \, .
\end{align}
The curvature of the connection $\omega_\alpha$ is given by
\begin{align}\label{omega_curvature}
    F_{+ - } = \partial_+ \omega_- - \partial_- \omega_+ + [ \omega_+ , \omega_- ] \, .
\end{align}
We would like to impose that the curvature (\ref{omega_curvature}) vanish, to ensure that we have a flat connection. To do this, we add a Lagrange multiplier term to the action,
\begin{align}\label{lag_mult_term}
    \mathcal{L}_{\Lambda} = \frac{1}{2} \epsilon^{\alpha \beta} \tr \left( \Lambda F_{\alpha \beta} \right) = \frac{1}{2} \Lambda_A F_{+ -}^A \, ,
\end{align}
where $\Lambda_A$ is the Lie algebra valued Lagrange multiplier, and where we have used that in our conventions $\epsilon^{+ -} = - \epsilon^{-+} = \frac{1}{2}$. Combining the Lagrange multiplier term (\ref{lag_mult_term}) with the AFSM Lagrangian \eqref{deformed_family}, after promoting $j_\alpha$ to $j_\alpha^{(\omega)}$, gives the combined action
\begin{align}\label{AFSM_master}
    S_{\text{master}} &= \int d^2 \sigma \Big{[} 
    \tr\Big(
    \frac{1}{2} j_+^{(\omega)} j_-^{(\omega)} 
    + v_+ v_- 
    +  j_+^{(\omega)} v_- + j_-^{(\omega)} v_+ 
+ \frac{1}{2} \Lambda F_{+ -}\Big)
\non\\
&~~~~~~~~~~~~~~
+ E ( \nu_2 , \nu_3 , \ldots , \nu_N ) \Big{]} \, ,
\end{align}
which we refer to as the \emph{master action}.

As one expects from the gauging of a global symmetry, the $G_L$ action has been promoted to a local symmetry; that is, the action (\ref{AFSM_master}) is invariant under
\begin{align}
    g \to h^{-1} g \, , \qquad \omega \to h^{-1} d h + h^{-1} \omega h \, ,
\end{align}
for any $h: \Sigma \to G$. By choosing $h ( \sigma, \tau ) = g ( \sigma, \tau )$, we can perform such a gauge transformation which sets $g = 1$ everywhere.\footnote{For simplicity, we have chosen to gauge the entire group $G_L$ rather than a subgroup; if we had instead gauged only a proper subgroup $H \subset G_L$, then it would not be possible to gauge-fix $g = 1$ in general.} After performing this gauge-fixing, we find that the covariantized Maurer-Cartan form collapses to $j_\alpha^{(\omega)} = \omega_\alpha$, and the action (\ref{AFSM_master}) becomes
\begin{align}\label{AFSM_master_GF}
    S_{\text{master, GF}} &= \int d^2 \sigma \Big{[} 
    \tr\Big(
    \frac{1}{2} \omega_+ \omega_- 
    + v_+ v_- 
    +  \omega_+v_- 
    + \omega_-v_+ 
+ \frac{1}{2} \Lambda F_{+ -}\Big)
\non\\
&~~~~~~~~~~~~~~
+ E ( \nu_2 , \nu_3 , \ldots , \nu_N ) \Big{]} \, .
\end{align}
Beginning from this master action (\ref{AFSM_master_GF}), we have two options:
\begin{enumerate}[label = (\alph*)]
    \item\label{master_to_afsm} we can integrate out the Lagrange multiplier field $\Lambda$ to recover the original auxiliary field sigma model; or

    \item\label{master_to_tdafsm} we can integrate out the gauge field $\omega_{\pm}$ in order to obtain the T-dual model, for which the multiplier $\Lambda_A$ is the fundamental degree of freedom.
\end{enumerate}
As a consistency check, let us first perform the first procedure \ref{master_to_afsm} and verify that we recover the original model. The equation of motion for $\Lambda_A$ simply imposes
\begin{align}
    F_{+-}^A = 0 \, ,
\end{align}
which means that the connection $\omega$ is flat. We may therefore take $\omega$ to be pure gauge,
\begin{align}
    \omega_{\pm} = \hat{g}^{-1} \partial_{\pm} \hat{g} \, ,
\end{align}
for some function $\hat{g} : \Sigma \to G$. After making this choice, we see that
\begin{align}
    \tr \left( \omega_+ \omega_- \right) &= \tr \left( \hat{g}^{-1} \partial_+ \hat{g} \hat{g}^{-1} \partial_- \hat{g} \right) \, ,
\end{align}
which motivates the definition of a new ``hatted'' Maurer-Cartan form
\begin{align}
    j_{\pm} = \hat{g}^{-1} \partial_{\pm} \hat{g} \, .
\end{align}
In terms of this field $j_{\pm}$, the master action reduces to
\begin{align}
    S_{\text{master, GF}} \xrightarrow{\text{Integrate out } \Lambda} & \int d^2 \sigma \, \Big{[} \tr\Big(
    \frac{1}{2}j_+ j_- 
    + v_+ v_-  
    +j_+ v_- + j_- v_+ \Big)
    + E ( \nu_2 , \ldots , \nu_N ) \Big{]}
    \nonumber \\
    &= S_{\text{AFSM}} \, ,
\end{align}
which is indeed the original action for the auxiliary field sigma model, as claimed.

Next we will proceed in the opposite direction \ref{master_to_tdafsm}, i.e. we will complete the T-dualization procedure by integrating out the gauge fields $\omega_\alpha^A$ using their equations of motion. To do this, it is convenient to expand the field strength in generators as
\begin{align}
    F_{+ -} = \left( \partial_+ \omega_-^A - \partial_- \omega_+^A \right) T_A + \tensor{f}{_A_B^C} \omega_+^A \omega_-^B T_C \, ,
\end{align}
and then integrate by parts in terms involving derivatives of $\omega$ to find
\begin{align}\label{AFSM_master_GF_two}
    S_{\text{master, GF}} &= \int d^2 \sigma \Bigg( \frac{1}{2} \omega_+^A \omega_{-, A} + v_+^A v_{-, A} + \omega_+^A v_{-, A} + \omega_-^A v_{+, A} \nonumber \\
    &\qquad \qquad \qquad + \frac{1}{2} \left( \omega_+^A \partial_- \Lambda_A - \omega_-^A \partial_+ \Lambda_A + \tensor{f}{_A_B^C} \omega_+^A \omega_-^B \Lambda_C \right) + E ( \nu_2 , \ldots , \nu_N ) \Bigg) \, .
\end{align}
The Euler-Lagrange equation that arises from varying (\ref{AFSM_master_GF_two}) with respect to the components $\omega_{\pm}^A$ of the gauge field are
\begin{align}\label{omega_eom}
    0 &= \frac{\delta S_{\text{master, GF}}}{\delta \omega_{+}^A} = \frac{1}{2} \omega_{-, A} + v_{-, A} + \frac{1}{2} \partial_{-} \Lambda_A + \frac{1}{2} \tensor{f}{_A_B^C} \omega_{-}^B \Lambda_C  \, , \nonumber \\
    0 &= \frac{\delta S_{\text{master, GF}}}{\delta \omega_{-}^A} = \frac{1}{2} \omega_{+, A} + v_{+, A} - \frac{1}{2} \partial_{+} \Lambda_A - \frac{1}{2} \tensor{f}{_A_B^C} \omega_{+}^B \Lambda_C \, .
\end{align}
It is natural to define the matrix
\begin{align}\label{MAB_defn}
    M_{AB} = \gamma_{AB} - \tensor{f}{_A_B^C} \Lambda_C \, ,
\end{align}
so that the equations of motion (\ref{omega_eom}) can be written as
\begin{align}
    0 &= M_{BA} \omega_-^B + \partial_- \Lambda_A + 2 v_{-, A} \, , \nonumber \\
    0 &= M_{AB} \omega_+^B - \partial_+ \Lambda_A + 2 v_{+, A} \, ,
\end{align}
whose solutions are
\begin{align}\label{omega_solns}
    \omega_-^B &= - \left( M^{-1} \right)^{AB} \left( \partial_- \Lambda_A + 2 v_{-, A} \right) \, , \nonumber \\
    \omega_+^B &= \left( M^{-1} \right)^{BA} \left( \partial_+ \Lambda_A - 2 v_{+, A} \right) \, .
\end{align}
Substituting the solutions (\ref{omega_solns}) for $\omega_{\pm}$ into the gauge-fixed master action (\ref{AFSM_master_GF_two}), and after some algebra outlined around equation (\ref{remove_gauge_fields_result}), we arrive at the T-dual AFSM action
\begin{align}\hspace{-16pt}\label{TD_AFSM}
    S_{\text{TD-AFSM}} &= \int d^2 \sigma \, \Bigg( \,  \frac{1}{2} ( \partial_- \Lambda_A )  \left( M^{-1} \right)^{AB} ( \partial_+ \Lambda_B ) 
    + v_+^A v_{-, A}
    - 2 v_{-, A} \left( M^{-1} \right)^{AB} v_{+, B} 
    \nonumber \\
    &\qquad \qquad ~~~
    + v_{-, A} \left( M^{-1} \right)^{AB} ( \partial_+ \Lambda_B )  - ( \partial_- \Lambda_A ) \left( M^{-1} \right)^{AB}  v_{+, B} 
    \nonumber \\
    &\qquad \qquad ~~~
    + E ( \nu_2 , \nu_3 , \ldots , \nu_N ) \Bigg) \, .
\end{align}
We also note that the Lagrangian $\mathcal{L}_{\text{TD-AFSM}}$ appearing in the integrand of the action (\ref{TD_AFSM}) can be expressed using the combinations
\begin{align}\label{frak_K_defn}
    \mathfrak{K}_+^B = \left( M^{-1} \right)^{BA} \left( \partial_+ \Lambda_A - 2 v_{+, A} \right) \, , \qquad \mathfrak{K}_-^B = - \left( M^{-1} \right)^{AB} \left( \partial_- \Lambda_A + 2 v_{-, A} \right) \, ,
\end{align}
in terms of which it takes the form
\begin{align}
    \mathcal{L}_{\text{TD-AFSM}} = - \frac{1}{2} \mathfrak{K}_{+, A} M^{BA} \mathfrak{K}_{-, B}  + v_+^A v_{-, A} + E ( \nu_2 , \ldots , \nu_N ) \, ,
\end{align}
as we show around equation (\ref{L_TDAFSM_with_K}). The combination $\mathfrak{K}_{\pm}$, and the related field $k_{\pm}$ to be defined in equation (\ref{kpmA_defn}), were referred to as $\tilde{j}_{\pm}$ and $\widetilde{\mathfrak{J}}_{\pm}$ in \cite{Bielli:2024khq}, respectively, but here we have chosen different notation to avoid confusing them with the right-invariant Maurer-Cartan form $\tilde{j}_{\pm}$ of equation (\ref{left_right_maurer_cartan}).

\subsubsection*{\ul{\it Reduction to T-dual PCM for $E = 0$}}

We obtained the action (\ref{TD_AFSM}) by T-dualizing a deformed auxiliary field sigma model described by an interaction function $E$. However, when $E = 0$, we know that the AFSM reduces to the usual principal chiral model. Therefore, in this limit, the T-dual action (\ref{TD_AFSM}) should reduce to that of the ordinary T-dual PCM, which was given in equation (\ref{tdual_PCM}), after integrating out the auxiliary fields. This turns out to be the case, although it is not immediately obvious from the form of (\ref{TD_AFSM}), so we now pause to verify this explicitly.

The Euler-Lagrange equations associated with $S_{\text{TD-AFSM}}$ for a general interaction function $E$ are derived in appendix \ref{app:td_afsm}. Specializing to the case $E = 0$, the solution to the auxiliary field equations of motion are
\begin{align}\label{vpmA_soln}
    v_{+}^A &\deq \left( M^{-1} \right)^{BA} \partial_+ \Lambda_B \, , 
    \qquad
    v_{-}^A \deq - \left( M^{-1} \right)^{A B} \partial_- \Lambda_B \, .
\end{align}
Substituting the solutions (\ref{vpmA_soln}) into the action (\ref{TD_AFSM}) gives
\begin{align}\hspace{-10pt}
    S_{\text{TD-AFSM}} &\deq \int d^2 \sigma \, \Bigg( \frac{1}{2} ( \partial_- \Lambda_D )  \left( M^{-1} \right)^{CD} ( \partial_+ \Lambda_C ) + 2 \left( \left( M^{-1} \right)^{C A} \partial_- \Lambda_A \right)  \left( M^{-1} \right)_{CD} \left( M^{-1} \right)^{BD} \partial_+ \Lambda_B  \nonumber \\
    &\qquad - \left( \left( M^{-1} \right)^{C E} \partial_- \Lambda_E \right) \left( M^{-1} \right)_{CD} \partial_+ \Lambda^D  - \partial_- \Lambda^C \left( M^{-1} \right)_{CD}  \left( M^{-1} \right)^{BD} \partial_+ \Lambda_B \nonumber \\
    &\qquad - \left( \left( M^{-1} \right)^{DA} \partial_+ \Lambda_D \right) \gamma_{AB} \left( \left( M^{-1} \right)^{B C} \partial_- \Lambda_C \right) \Bigg) \, ,
\end{align}
which simplifies to
\begin{align}
    S_{\text{TD-AFSM}} &\deq \int d^2 \sigma \, ( \partial_- \Lambda_A ) ( \partial_+ \Lambda_B ) \Bigg( \frac{1}{2} \left( M^{-1} \right)^{AB} + 2 \left( M^{-1} \right)^{C A}  \left( M^{-1} \right)_{CD} \left( M^{-1} \right)^{BD}  \nonumber \\
    &~~~ - \left( M^{-1} \right)^{CA} \tensor{\left( M^{-1} \right)}{_C^B}   - \tensor{\left( M^{-1} \right)}{^A_D} \left( M^{-1} \right)^{BD}  - \left( M^{-1} \right)^{B C} \tensor{\left( M^{-1} \right)}{_C^A} \Bigg)  \, .
\end{align}
To complete the derivation, one needs some useful identities for the matrix $M^{-1}$, which are derived in appendix \ref{app:MAB}. In particular, using the relations (\ref{stronger_result_one}) and (\ref{triple_product}), we find
\begin{align}
    S_{\text{TD-AFSM}} &\deq \int d^2 \sigma \, ( \partial_- \Lambda_A ) ( \partial_+ \Lambda_B ) \left( \frac{1}{2} \left( M^{-1} \right)^{AB} - \frac{1}{2} \left( M^{-1} \right)^{AB} - \frac{1}{2} \left( M^{-1} \right)^{BA} \right) \nonumber \\
    &= - \frac{1}{2} \int d^2 \sigma \,  ( \partial_+ \Lambda_B ) \left( M^{-1} \right)^{BA} \left( \partial_- \Lambda_A \right) \, .
\end{align}
This precisely matches (\ref{tdual_PCM}), confirming that our T-dual auxiliary field models correctly reduce to the T-dual PCM in the undeformed limit.

\subsubsection*{\ul{\it Exchange of equation of motion and Maurer-Cartan identity}}

Here we will briefly discuss some aspects of the equations of motion for the T-dual auxiliary field sigma model (\ref{TD_AFSM}), and how they relate to those of the original AFSM (\ref{deformed_family}). 

First let us set expectations for how the T-dual model should behave, using a heuristic argument based on observations about the master action (\ref{AFSM_master_GF_two}). Recall that, when we integrated out $\Lambda$ to perform the procedure \ref{master_to_afsm}, we found
\begin{align}
    \omega_{\pm} = j_{\pm} \, ,
\end{align}
whereas when we integrated out $\omega$ to implement the direction \ref{master_to_tdafsm}, instead, we obtained
\begin{align}
    \omega_\pm = \mathfrak{K}_{\pm} \, ,
\end{align}
where $\mathfrak{K}_{\pm}$ is defined in (\ref{frak_K_defn}). This suggests the rough identification
\begin{align}\label{j_to_frak_K}
    j_\pm = \mathfrak{K}_{\pm} \, .
\end{align}
In the original auxiliary field sigma model, when $j_{\pm}$ is related to the fundamental field as $j_{\pm} = g^{-1} \partial_{\pm} g$ and the equations of motion are satisfied, then $j_{\pm}$ is flat but not conserved, while the modified field $\CC{J}_{\pm}$ is conserved but not flat. Thus the identification (\ref{j_to_frak_K}) suggests that, when the appropriate conditions in the T-dual theory are met, the field $\mathfrak{K}_{\pm}$ should be flat but not conserved, while some other related field $k_{\pm}$ should be conserved but not flat. 

Indeed, this turns out to be the case. As we show in appendix \ref{app:td_afsm}, the equation of motion for $\Lambda_A$ is equivalent to the condition
\begin{align}\label{TD_Lambda_eom_body}
    \partial_+ \CC{K}_- - \partial_- \CC{K}_+ + [ \CC{K}_+ , \CC{K}_- ] = 0 \, ,
\end{align}
which expresses the flatness of $\mathfrak{K}_{\pm}$. 

Furthermore, the equation of motion for the auxiliary field $v_{\pm}$ takes the form
\begin{align}
    0 \deq \mathfrak{K}_{\pm} + v_{\pm} + \sum_{n = 2}^{N} n E_n \tr ( v_{\pm}^n ) v_{\mp}^{A_1} \ldots v_{\mp}^{A_{n-1}} T^B \tr ( T_{(B} T_{A_1} \ldots T_{A_{n-1} )} ) \, .
\end{align}
Therefore, using an identical argument to the one around equation (\ref{fundamental_commutator_identity}), which relies on the generalized Jacobi identity discussed in appendix \ref{app:jac} and in particular its implication (\ref{gen_jac_body}), one finds that
\begin{align}\label{tdual_fundamental_commutator}
    [ v_{\mp} , \mathfrak{K}_{\pm} ] \deq [ v_{\pm} , v_{\mp} ] \, .
\end{align}
This is exactly the same structure which we found in our study of the ordinary AFSM, except with $j_{\pm}$ replaced by $\mathfrak{K}_{\pm}$. As a consequence, versions of all of the commutator identities which we derived in section \ref{sec:implications} for the AFSM will also hold for the T-dual AFSM. If we define the related fields
\begin{align}\label{kpmA_defn}
    k_\pm = - \left( \mathfrak{K}_{\pm} + 2 v_{\pm} \right) \, ,
\end{align}
then we note that
\begin{align}
    [ k_+ , k_- ] &= [ \CC{K}_+ + 2 v_+ , \CC{K}_- + 2 v_- ]  \nonumber \\
    &= [ \CC{K}_+ , \CC{K}_- ] + 2 [ \CC{K}_+ , v_- ] + 2 [ v_+ , \CC{K}_- ] + 4 [ v_+ , v_- ] \nonumber \\
    &\deq  [ \CC{K}_+ , \CC{K}_- ] - 2 [ v_+ , v_- ] - 2 [ v_+ , v_- ] + 4 [ v_+ , v_- ] \nonumber \\
    &= [ \CC{K}_+ , \CC{K}_- ] \, ,
\end{align}
in complete analogy with equation (\ref{imp_two}). Likewise, we find
\begin{align}
    [ \CC{K}_+ , k_- ] \deq [ k_+ , \CC{K}_- ] \, ,
\end{align}
by steps identical to those around equation (\ref{imp_one}).

Furthermore, when equations (\ref{TD_Lambda_eom_body}) and (\ref{tdual_fundamental_commutator}) are satisfied, this field obeys the conservation equation
\begin{align}
    \partial_\alpha k^\alpha \deq 0 \, ,
\end{align}
as we show around equation (\ref{k_cons_final}).

In this sense, the dualization procedure has exchanged the equation of motion and Maurer-Cartan identity.\footnote{See \cite{Ivanov:1987yv} where the interchange of the Maurer-Cartan identity and equation of motion is interpreted as giving a dual description of the PCM which is related to the Drinfeld double.} In the AFSM, the combination $\CC{J}_{\pm}$, which is built from both $j_{\pm}$ and auxiliary fields, is conserved but not flat. In the T-dual model, the analogous combination $\CC{K}_{\pm}$, which involves auxiliaries in a similar way as $\CC{J}_{\pm}$ via its definition (\ref{frak_K_defn}), is instead flat but not conserved, while it is $k_{\pm}$ which is conserved but not flat.

\subsubsection*{\ul{\it Lax representation}}

The above observation concerning the interchange of the equation of motion and Maurer-Cartan identity in the T-dual auxiliary field sigma model suggests the possibility of building a Lax connection for the TD-AFSM which takes the form
\begin{align}\label{tdual_lax}
    \CC{L}_{\pm} = \frac{\CC{K}_{\pm} \pm z k_{\pm}}{1 - z^2} \, ,
\end{align}
by analogy with (\ref{lax_connection}) for the AFSM. One could then proceed to show that the flatness of the Lax connection (\ref{tdual_lax}) for any $z \in \mathbb{C}$ is equivalent to the equations of motion for the model, then compute the Poisson bracket of the spatial component of the Lax connection to argue for strong integrability, as we did in section \ref{sec:maillet}.

Even though one could carry out this procedure to establish weak and strong integrability of the T-dual AFSM, we will not perform these steps here, since we will present a simpler argument for the classical integrability of this model in section \ref{sec:canonical}. There, we will see that the T-dual AFSM is equivalent to the ordinary AFSM up to a symplectomorphism, or a change of phase space coordinates which preserves the form of Hamilton's equations and the Poisson brackets. This result immediately implies integrability of the T-dual model without explicitly carrying out the steps which we performed for the AFSM in section \ref{sec:integrability}.

\subsection{A Commuting Square of Dualities and Deformations}\label{sec:comm_square}

In section \ref{sec:dualization}, we have constructed a family of models by first coupling the principal chiral model to auxiliary fields, and then T-dualizing the result. It is natural to wonder whether these two procedures can be performed in the opposite order. That is, can one first T-dualize the undeformed PCM, and then couple the resulting theory to auxiliary fields? And if so, do the results of performing these procedures in the two different orders agree?

To answer this question, we must first specify what we mean by coupling the T-dual PCM to auxiliary fields. That is, we would like to write down an auxiliary field model which extends the undeformed T-dual PCM action (\ref{tdual_PCM}) in a way which is analogous to the extension of the principal chiral model to the auxiliary field sigma model (\ref{deformed_family}). 

Let us enumerate the properties that we would like such an extension to satisfy:

\begin{enumerate}[label = (\roman*)]
    \item The Lagrangian $\mathcal{L}_{\text{AF-TDSM}}$ should depend on the field $\Lambda_A$, which appears in the undeformed Lagrangian $\mathcal{L}_{\text{TD-PCM}}$, in addition to a vector auxiliary field $\tilde{v}_\alpha$.

    \item The desired $\mathcal{L}_{\text{AF-TDSM}}$ should include only terms of the following types: a kinetic term for $\Lambda_A$ which is quadratic in its derivatives; a term quadratic in auxiliary fields $\tilde{v}_\alpha$; terms which linearly couple $\tilde{v}_\alpha$ to $\partial_\alpha \Lambda$; and an interaction function $E$ that depends on a collection of scalars $\tilde{\nu}_k$.

    \item\label{v_to_Lambda_property} When $E = 0$, the auxiliary field equation of motion must be $\tilde{v}_\pm = - \partial_\pm \Lambda$, and after integrating out $\tilde{v}_\alpha$ using this equation of motion, $\mathcal{L}_{\text{AF-TDSM}}$ must reduce to $\mathcal{L}_{\text{TD-PCM}}$.

    \item\label{td_current_property} When $\tilde{v}_\pm = - \partial_\pm \Lambda$, each variable $\tilde{\nu}_k$ must reduce to a bilinear of the form (\ref{Ok_defn}) which is a product of the spin-$k$ conserved currents in the T-dual PCM.
\end{enumerate}
We have selected these properties because they are identical to the ones satisfied by the auxiliary field deformation (\ref{deformed_family}) of the principal chiral model. Assumption \ref{td_current_property} implies that in the special case where $E = E ( \tilde{\nu}_2 )$, the family of auxiliary field T-dual sigma models is equivalent to all deformations of $\mathcal{L}_{\text{TD-PCM}}$ by functions of the energy-momentum tensor.

We note that the Lagrangian in the integrand of the action (\ref{TD_AFSM}), which is obtained from T-dualizing the AFSM, does not satisfy properties \ref{v_to_Lambda_property} and \ref{td_current_property}. However, it is straightforward to guess a Lagrangian which satisfies all of these assumptions:
\begin{align}\hspace{-15pt}\label{AFTDSM}
    \mathcal{L}_{\text{AF-TDSM}} &= \frac{1}{2} \left( \partial_+ \Lambda_C \right) \left( M^{-1} \right)^{(CD)} \left( \partial_- \Lambda_D \right) - \frac{1}{2}  \left( \partial_+ \Lambda_C \right) \left( M^{-1} \right)^{[CD]} \left( \partial_- \Lambda_D \right) \nonumber \\
    &~~~
    + \tilde{v}_{-, C} \left( M^{-1} \right)^{(CD)} \tilde{v}_{+, D} 
    + \left( \partial_+ \Lambda_C \right) \left( M^{-1} \right)^{(CD)} \tilde{v}_{-, D} + \tilde{v}_{+, C} \left( M^{-1} \right)^{(CD)} \left( \partial_- \Lambda_D \right)
    \nonumber \\
    &~~~
    + E \left( \tilde{\nu}_2 , \ldots , \tilde{\nu}_N \right) \, ,
\end{align}
where
\begin{align}\label{td_nu_defn}
    \tilde{\nu}_k = ( - 1 )^k \tr \Big( \underbrace{\left( M^{-1} \right)^T \tilde{v}_+ \cdot \ldots \cdot \left( M^{-1} \right)^T \cdot \tilde{v}_+ }_{k \text{ copies of } \left( M^{-1} \right)^T \cdot \tilde{v}_+} \Big) \tr \Big( \underbrace{ \left( M^{-1} \right) \cdot \tilde{v}_- \cdot \ldots \cdot \left( M^{-1} \right) \cdot \tilde{v}_-  }_{k \text{ copies of } \left( M^{-1} \right) \cdot \tilde{v}_-} \Big) \, .
\end{align}
As a sanity check, let us verify that (\ref{td_nu_defn}) obeys property \ref{v_to_Lambda_property}. Varying $\tilde{v}_-^A$ gives
\begin{align}\label{vmA_variation_intermediate}
    \left( M^{-1} \right)^{(A D)} \tilde{v}_{+, D} + ( \partial_+ \Lambda_D ) \left( M^{-1} \right)^{(D A)} = 0 \, ,
\end{align}
whose solution is
\begin{align}
    \tilde{v}_+^A \deq - \partial_+ \Lambda^A \, ,
\end{align}
where as usual we use $\deq$ to indicate that two quantities are equal when the auxiliary field equations of motion are satisfied. By similar reasoning we find
\begin{align}\label{my_aux_soln}
    \tilde{v}_-^A \deq - \partial_- \Lambda^A \, .
\end{align}
When we plug these in, we find
\begin{align}\hspace{-13pt}
    \mathcal{L}_{\text{AF-TDSM}} &\deq \frac{1}{2} \left( \partial_+ \Lambda_C \right) \left( M^{-1} \right)^{(CD)} \left( \partial_- \Lambda_D \right) - \frac{1}{2}  \left( \partial_+ \Lambda_C \right) \left( M^{-1} \right)^{[CD]} \left( \partial_- \Lambda_D \right) + \partial_- \Lambda_C \left( M^{-1} \right)^{(CD)} \partial_+ \Lambda_D \nonumber \\
    &\qquad - \left( \partial_+ \Lambda_C \right) \left( M^{-1} \right)^{(CD)} \partial_- \Lambda_D  - \partial_+ \Lambda_D \left( M^{-1} \right)^{(CD)} \left( \partial_- \Lambda_D \right) \, \nonumber \\
    &= - \frac{1}{2} \partial_+ \Lambda_C \left( M^{-1} \right)^{(CD)} \partial_+ \Lambda_D - \frac{1}{2}  \left( \partial_+ \Lambda_C \right) \left( M^{-1} \right)^{[CD]} \left( \partial_- \Lambda_D \right) \nonumber \\
    &= \mathcal{L}_{\text{TD-PCM}} \, .
\end{align}
Thus, when $E = 0$, the auxiliary field action (\ref{AFTDSM}) reduces to the T-dual PCM.

It may seem puzzling that we have now constructed two different theories, (\ref{TD_AFSM}) and (\ref{AFTDSM}), which involve an interaction function $E$ of several variables and which both reduce to the T-dual PCM when $E = 0$. This suggests that the two models might in fact be equivalent, in the sense that they agree after a field redefinition. If this were true, then it would give an answer to the question which we posed at the beginning of this section, namely whether the processes of T-duality and coupling to auxiliary fields commute. Indeed, since (\ref{TD_AFSM}) was obtained from one order of these operations (couple to auxiliaries, then T-dualize) while (\ref{AFTDSM}) arose from the other order (T-dualize, then couple to auxiliaries), a proof of the physical equivalence of these two theories would amount to a demonstration of the commutativity of the two operations.

We will now see that such an equivalence does, in fact, hold, as indicated in the diagram
\begin{align}\label{comm_digram}
\begin{tikzcd}[ampersand replacement=\&]
    {S_{\text{PCM}}} \&\&\& {S_{\text{TD-PCM}}} \\
    \\
    \&\&\& {S_{\text{AF-TDSM}}} \\
    {S_{\text{AFSM}}} \&\& {S_{\text{TD-AFSM}}}
    \arrow["{\text{T-dualize}}", from=1-1, to=1-4]
    \arrow["{\text{Auxiliaries}}"{description}, from=1-1, to=4-1]
    \arrow["{\text{Auxiliaries}}"{description}, from=1-4, to=3-4]
    \arrow["{\text{T-dualize}}"', from=4-1, to=4-3]
    \arrow["{\substack{\large \text{Field} \\ \normalsize \text{Redefinition}}}"', curve={height=18pt}, tail reversed, from=4-3, to=3-4]
\end{tikzcd} \, .
\end{align}
To see this, consider the field redefinition
\begin{align}\label{v_field_redef}
    v_+^B = - \left( M^{-1} \right)^{AB} \tilde{v}_{+, A} \, , \qquad v_-^B = \left( M^{-1} \right)^{BA} \tilde{v}_{-, A} \, .
\end{align}
Let us substitute this redefinition into the Lagrangian for the T-dual AFSM,
\begin{align}\label{LTD_AFSM}
    \mathcal{L}_{\text{TD-AFSM}} \big\vert &= \frac{1}{2} ( \partial_- \Lambda_C )  \left( M^{-1} \right)^{CD} ( \partial_+ \Lambda_D ) - 2 v_{-, C} \left( M^{-1} \right)^{CD} v_{+, D} 
    + v_+^A \gamma_{AB} v_-^B
       \nonumber \\
    &\qquad 
    + v_{-, C} \left( M^{-1} \right)^{CD} \partial_+ \Lambda_D
- \partial_- \Lambda_C \left( M^{-1} \right)^{CD} v_{+, D}  + E ( \nu_2 , \ldots , \nu_N ) \, ,
\end{align}
where $\big\vert$ refers to evaluation after substituting (\ref{v_field_redef}), which gives
\begin{align}\label{field_redefn_massage}
    \mathcal{L}_{\text{TD-AFSM}} \big\vert &= \frac{1}{2} ( \partial_- \Lambda_C )  \left( M^{-1} \right)^{CD} ( \partial_+ \Lambda_D ) + 2 \left( \left( M^{-1} \right)^{CA} \tilde{v}_{-, A} \right) \left( M^{-1} \right)^{CD} \left( \left( M^{-1} \right)^{BD} \tilde{v}_{+, B}  \right)  \nonumber \\
    &\qquad 
- \left( M^{-1} \right)^{CA} \tilde{v}_{+, C} \gamma_{AB} \left( \left( M^{-1} \right)^{BD} \tilde{v}_{-, D} \right) + \left( \left( M^{-1} \right)^{CA} \tilde{v}_{-, A} \right) \left( M^{-1} \right)^{CD} \partial_+ \Lambda_D 
\nonumber \\
    &\qquad 
 + \partial_- \Lambda_C \left( M^{-1} \right)^{CD} \left( M^{-1} \right)^{AD} \tilde{v}_{+, A}  
+ E ( \nu_2 , \ldots , \nu_N ) \, .
\end{align}
To simplify the products of the $M^{-1}$ matrices, we must appeal to the identities (\ref{stronger_result_one}) and (\ref{triple_product}) which are proved in appendix \ref{app:MAB}. Using these and simplifying, (\ref{field_redefn_massage}) becomes
\begin{align}\label{td_AFSM_change_vars_before_E}
    \mathcal{L}_{\text{TD-AFSM}} \big\vert &=  \frac{1}{2} ( \partial_- \Lambda_C )  \left( M^{-1} \right)^{CD} ( \partial_+ \Lambda_D ) + \tilde{v}_{+, A} \left( M^{-1} \right)^{(AB)}  \tilde{v}_{-, B} + \tilde{v}_{-, A} \left( M^{-1} \right)^{(AB)}  \partial_+ \Lambda_B  \nonumber \\
    &\qquad + \tilde{v}_{+, A} \left( M^{-1} \right)^{(AC)} \partial_- \Lambda_B + E ( \nu_2 , \ldots , \nu_N ) \, .
\end{align}
We therefore see that all of the terms in (\ref{td_AFSM_change_vars_before_E}) match those in (\ref{AFTDSM}) with the possible exception of the interaction function, which in (\ref{td_AFSM_change_vars_before_E}) depends on the variables $\nu_k$ rather than the variables $\tilde{\nu}_k$. However, it turns out that these are also equal. Beginning with
\begin{align}
    \nu_k = \tr ( v_+^k ) \tr ( v_-^k ) \, ,
\end{align}
and making the field redefinition (\ref{v_field_redef}) gives
\begin{align}
    \nu_k \big\vert &= \tr \Big( \underbrace{ ( - 1 ) \left( M^{-1} \right)^T \tilde{v}_+ \cdot \ldots \cdot ( - 1 ) \left( M^{-1} \right)^T \cdot \tilde{v}_+ }_{k \text{ times}} \Big) \tr \Big( \underbrace{ \left( M^{-1} \right) \cdot \tilde{v}_- \cdot \ldots \cdot \left( M^{-1} \right) \cdot \tilde{v}_-  }_{k \text{ times}} \Big) \nonumber \\
    &= \tilde{\nu}_k \, .
\end{align}
Therefore, we have proven that
\begin{align}
    \mathcal{L}_{\text{TD-AFSM}} \big\vert = \mathcal{L}_{\text{AF-TDSM}} \, ,
\end{align}
after making the field redefinition (\ref{TD_AFSM}). This confirms the physical equivalence of the two models and the commutativity of the diagram (\ref{comm_digram}).

\subsection{Canonical Transformation and Integrability}\label{sec:canonical}

It has been known since the work of \cite{Alvarez:1994wj,Lozano:1995jx} that the undeformed PCM and its T-dual theory (\ref{tdual_PCM}) are related by a canonical transformation, which means that the Hamiltonian formulations of the two theories are equivalent up to a change of variables. We will also refer to a canonical transformation as a symplectomorphism or symplectic diffeomorphism, i.e. a diffeomorphism between phase space manifolds which preserves the symplectic structure.

The existence of such a canonical transformation is a powerful tool for inferring properties about the model in one duality frame from the model in the other frame. For instance, as we have already reviewed, the principal chiral model is strongly integrable in the sense that it possesses an infinite tower of conserved charges $Q_n$ which obey
\begin{align}
    \left\{ Q_n , Q_m \right\} = 0 \, .
\end{align}
Let $\widetilde{Q}_n$ be the images of these charges under the symplectic diffeomorphism which maps the PCM to the T-dual PCM. Since the Poisson brackets are preserved under any canonical transformation, we immediately conclude that
\begin{align}\label{canonical_transform_charges}
    \big\{ \widetilde{Q}_n , \widetilde{Q}_m \big\} = 0 \, ,
\end{align}
which implies that the T-dual PCM, with Lagrangian (\ref{tdual_PCM}), is also strongly integrable.

In this section, we will prove that the deformed auxiliary field sigma models are also related to their non-Abelian T-duals by a canonical transformation. In fact, perhaps surprising, they are related by the \emph{same} canonical transformation as the one which relates the undeformed PCM and its T-dual. Thus the symplectomorphism which implements non-Abelian T-duality in the AFSM is completely blind to the choice of interaction function.

We remind the reader that we have established that the AFSM is related to its T-dual by a canonical transformation only for the case in which the entire group $G_L$ has been gauged. The analysis for the gauging of a proper subgroup $H \subset G_L$ will appear elsewhere.

\subsubsection*{\ul{\it Generalities on canonical transformation approach}}

Let us first review some basic aspects of the symplectomorphism which relates the principal chiral model and its T-dual in the undeformed setting. In this case, the canonical transformation is associated with a generating function of type 1. The generating function approach to canonical transformations is perhaps more familiar in the setting of particle mechanics, rather than in field theory. In the simpler particle context, suppose that one has a mechanical system with $N$ generalized coordinates $q^i$, $i = 1 , \ldots , N$, and their corresponding conjugate momenta $p_i$. Time evolution is governed by Hamilton's equations
\begin{align}
    \dot{q}^i = \frac{\partial H}{\partial p^i} \, , \qquad \dot{p}^i = - \frac{\partial H}{\partial q^i} \, ,
\end{align}
where $H = H ( q^i , p^i , t )$ is the Hamiltonian. We might then wish to change to a new set of phase space variables $Q^i ( t )$ and $P^i ( t )$, which are related to the old ones as
\begin{align}
    Q^i = Q^i ( q_j, p_j, t ) \, , \qquad P^i = P^i ( q_j , p_j , t ) \, .
\end{align}
This change of variables is said to be canonical if there exists a new Hamiltonian function $\widetilde{H} ( Q^i, P^i, t )$, which need not agree with the old one, such that
\begin{align}
    \dot{Q}^i = \frac{\partial \widetilde{H}}{\partial P^i} \, , \qquad \dot{P}^i = - \frac{\partial \widetilde{H}}{\partial Q^i} \, ,
\end{align}
i.e. so that Hamilton's equations have the same form when written in terms of the new coordinates. If such a new Hamiltonian $\widetilde{H}$ exists, this canonical transformation is guaranteed to preserve Poisson brackets (in fact, the preservation of Poisson brackets can be taken as an alternative but equivalent definition of a canonical transformation).

One way to generate valid canonical transformations is to use the generating function approach, where a symplectomorphism is obtained from a function $F$ that depends on a subset of the old and new phase space coordinates $( q^i, p^i, Q^i, P^i )$, and can depend explicitly on time. Generating functions are classified into four types based on which phase space coordinates they depend on; a discussion of these types and their properties can be found in standard textbooks \cite{goldstein2002classical}. We will be interested in type 1 generating functions, which take the form
\begin{align}
    F = F ( q^i , Q^i , t ) \, .
\end{align}
That is, a type 1 generating function can depend on the ``old'' coordinates $q^i$ and the ``new'' coordinates $Q^i$, but not on the conjugate momenta (and, again, may have explicit time dependence). Given such a generating function, the new Hamiltonian $\widetilde{H} ( Q^i, P^i )$ and the old Hamiltonian $H ( q^i , p^i )$ are related by
\begin{align}\label{type_1_ham}
    \widetilde{H} ( Q^i, P^i ) = H ( q^i , p^i ) + \frac{\partial F}{\partial t} \, ,
\end{align}
where one must take care in interpreting (\ref{type_1_ham}) as the new and old Hamiltonian are expressed in different sets of variables. For a generating function that does not depend explicitly on time, one has the relations
\begin{align}\label{type_1_ham_no_time}
    \widetilde{H} ( Q^i, P^i ) = H ( q^i , p^i ) \, ,
\end{align}
which means that the Hamiltonian transforms as a scalar under the symplectomorphism associated with the generating function $F$.

Given a type 1 generating function $F ( q^i, Q^i )$, the momenta are determined by
\begin{align}
    p_i = \frac{\partial F}{\partial q^i} \, , \qquad P^i = - \frac{\partial F}{\partial Q^i}\ .
\end{align}
As a simple example, if we choose the generating function
\begin{align}\label{particle_swap_p_q}
    F ( q^i, Q^i ) = q^i Q^i \, , 
\end{align}
then the symplectomorphism acts as
\begin{align}
    Q^i = p^i \, , \qquad P^i = - q^i \, , 
\end{align}
which has the effect of exchanging the generalized coordinates and their canonical momenta. A similar exchange of fields and momenta occurs for the type 1 generating function that implements T-duality in the principal chiral model. We turn to a discussion of such canonical transformations in field theories next.

\subsubsection*{\ul{\it Canonical transformations in field theory; generating function for T-dual PCM}}

In field theory, the situation is somewhat more complicated because the Hamiltonian
\begin{align}
    H [ \phi^i , \pi_i ] = \int d^d x \, \mathcal{H} ( \phi^i , \pi_i ) \, ,
\end{align}
is now a functional of the fundamental fields $\phi^i$ and their conjugate momenta $\pi_i$, obtained as the spatial integral of a Hamiltonian density $\mathcal{H}$. The appropriate generalization of the symplectomorphisms considered above is a change of phase space variables which preserves the Hamilton equations of motion,
\begin{align}\label{field_theory_hamilton}
    \dot{\phi}^i = \frac{\delta H}{\delta \pi_i} \, , \qquad \dot{\pi}_i = - \frac{\delta H}{\delta \phi^i} \, ,
\end{align}
and in particular this involves \emph{functional} derivatives of the integrated Hamiltonian $H$. For instance, the analogue of equation (\ref{type_1_ham_no_time}), which expresses the condition that the Hamiltonian is a scalar under a syplectomorphism, in the context of a field theoretical canonical transformation from old fields $( \phi^i , \pi_i )$ to new fields $( \Phi^i, \Pi_i )$ is
\begin{align}
    \widetilde{H} [ \Phi^i , \Pi_i ] = H [ \phi^i , \pi_i ] \, ,
\end{align}
which is an equality of functionals. One can consider quite general transformations
\begin{align}\label{integrated_canonical_transformations}
    \int d^d x \, \Phi^i &= \int d^d x \, \mathcal{F}_{i, 1} \left( \phi^j , \partial_j \phi^k , \pi_j, x_j, t \right) \, , \nonumber \\
    \int d^d x \, \Pi^i &= \int d^d x \, \mathcal{F}_{i, 2} \left( \phi^j , \partial_j \phi^k , \pi_j, x_j, t \right) \, ,
\end{align}
which leave (\ref{field_theory_hamilton}) invariant. See for instance \cite{musicki} for a discussion of such transformations.

However, for the present purposes, we will restrict to a smaller class of field-theoretic symplectomorphisms, which we call ``local transformations'' or ``point transformations'' and which only involve changes of variables relating $\phi^i$, $\pi_i$, $\Phi^i$, and $\Pi_i$ at a fixed spacetime point, rather than involving integrals as in (\ref{integrated_canonical_transformations}).

For point transformations, one can introduce the analogue of a type 1 generating function considered in particle mechanics,
\begin{align}\label{type_1_functional}
    F = F [ \phi^i , \Phi^i , t ] = \int d^d x \, \mathcal{F} ( \phi^i, \Phi^i, t ) \, .
\end{align}
From here onwards, we will consider only time-independent generating functions. In this case, under a symplectic diffeomorphism induced by a type 1 generating functional, the canonical momenta satisfy
\begin{align}\label{type_1_field_theory_momenta}
    \pi_i = \frac{\delta F}{\delta \phi^i} \, , \qquad \Pi_i = - \frac{\delta F}{\delta \Phi^i} \, ,
\end{align}
and the Hamiltonian \emph{density} (rather than the integrated Hamiltonian) obeys
\begin{align}
    \widetilde{\mathcal{H}} ( \Phi^i , \Pi_i ) = \mathcal{H} ( \phi^i , \pi_i ) \, .
\end{align}
This restricted class of canonical transformations in field theory is sufficient for our purposes, since it includes the symplectomorphism that relates the PCM to its non-Abelian T-dual (and, as we shall see, the AFSM to its T-dual).

In particular, we now specialize to the principal chiral model, whose fundamental fields $\phi^\mu$ are the local coordinates on the target space Lie group, and the canonical momenta are written $\pi_\mu$. In the T-dual, the fundamental fields $\Phi^i$ are taken to be the Lagrange multiplier fields $\Lambda_A$, which carry a Lie algebra index $A$, and their momenta will be written as $\Pi^A$.
Then consider the type 1 generating functional
\begin{align}\label{TD_PCM_gen_F}
    F [ \phi^\mu , \Lambda_A ] = - \int d \sigma \, \tr \left( \Lambda j_\sigma \right) \, .
\end{align}
Note that $j_\sigma$ depends on $\phi^\mu$ but not on the conjugate momenta $\pi_\mu$, as we remarked around equation (\ref{jsigma_no_pie}), so this generating functional is indeed type 1. Using (\ref{type_1_field_theory_momenta}), we check around equation (\ref{final_F_to_pi_check}) that the relations for the momenta implied by (\ref{TD_PCM_gen_F}) are 
\begin{align}\label{tdual_canonical}
    \pi_\mu &= j_\mu^A \partial_\sigma \Lambda_A - f_{AB}{}^C j_\nu^A j_\mu^B\Lambda_C  \partial_\sigma \phi^\nu \, , \nonumber \\
    \Pi^A &= j_\mu^A \partial_\sigma \phi^\mu \, .
\end{align}
One can find by explicit computation that the relations (\ref{tdual_canonical}) render the Hamiltonians of the principal chiral model and its non-Abelian T-dual equal \cite{Lozano:1995jx}; see also the review \cite{Lozano:1996sc}.

It is interesting to observe that the structure of the generating functional (\ref{TD_PCM_gen_F}) is similar to that of the generating function (\ref{particle_swap_p_q}) in particle mechanics, which interchanges coordinates with momenta. The na\"ive extension of the generating function $F_1 = q^i Q^i$ to the PCM and its T-dual would involve a product of the two fundamental fields $\phi^\mu$ and $\Lambda_A$, but the generating functional relevant for the PCM instead involves the spatial derivative of the field $\phi^\mu$, namely $\partial_\sigma \phi^\mu$. In the canonical formalism such spatial derivatives are still functions of the fundamental fields but not of their momenta, so this does not spoil the property that the generating functional is type 1. However, one still has the issue that $\partial_\sigma \phi^\mu$ and $\Lambda_A$ carry different indices. One can remedy this by converting the ``curved'' $\mu$ index to a ``flat'' tangent space index $A$ using the vielbein $j_\mu^A$ to define the combination
\begin{align}
    \left( \partial_\sigma \phi^\mu \right) j_\mu^A = j_\sigma^A \, ,
\end{align}
which is precisely the combination that appears in the integrand $\mathcal{F} = j_\sigma^A \Lambda_A$ of the generating functional (\ref{tdual_canonical}). In a loose sense, this matches our physical intuition for the action of Abelian T-duality on string sigma models, which interchanges momentum and winding modes along a target space circle. The expansion of a coordinate $X ( \sigma, \tau )$ for a closed string on a target space circle of radius $R$, in a sector with momentum $p = \frac{k}{R}$ and winding $w$, is
\begin{align}\label{momentum_winding}
    X ( \sigma, \tau ) = x + 2 \alpha' \frac{k}{R} \tau + 2 R w \sigma + \ldots \, , \qquad k , w \in \mathbb{Z} \, , 
\end{align}
where $\ldots$ captures oscillator terms. Therefore the winding contribution is associated with $\partial_\sigma X$ and the momentum contribution is determined by $\partial_\tau X$, which is related to the canonical momentum. Thus a generating functional of the form (\ref{TD_PCM_gen_F}) which (roughly) interchanges worldsheet spatial gradients $\partial_\sigma \phi$ with momenta has the right flavor for T-duality.

Next we will discuss the analogue of this canonical transformation which relates higher-spin auxiliary field deformations of the principal chiral model to their non-Abelian T-duals.

\subsubsection*{\ul{\it Canonical transformation for AFSM and its T-dual}}

Our present goal is to demonstrate that the Hamiltonian densities of the auxiliary field sigma model and its T-dual agree after performing the symplectomorphism associated with the generating functional (\ref{TD_PCM_gen_F}), or equivalently, when the momenta are related to the fields by (\ref{tdual_canonical}). We have already computed the Hamiltonian density $\mathcal{H}_{\text{AFSM}}$ of the AFSM in equation (\ref{final_afsm_ham}). To begin, we must carry out the same procedure for the T-dual AFSM by performing the Legendre transform of the Lagrangian appearing in the action (\ref{TD_AFSM}). This process is carried out in appendix \ref{app:tdual_details}, and the result is
\begin{align}\label{final_tdpcm_ham}
    \mathcal{H}_{\text{TD-AFSM}} &= \frac{1}{2} \Pi_C M^{AC} \tensor{M}{_A^D} \Pi_D - \Pi_A \tensor{f}{^A_B_C} \Lambda^C \left( \partial_\sigma \Lambda^B - 2 v_\tau^B \right) + \frac{1}{2} ( \partial_\sigma \Lambda^A ) ( \partial_\sigma \Lambda_A ) + v_\tau^A v_{\tau, A} \nonumber \\
    &\qquad + v_\sigma^A v_{\sigma, A}  + 2 v_\sigma^A \Pi_A  - 2 v_\tau^A \partial_\sigma \Lambda_A - E ( \nu_2 , \ldots , \nu_N ) \, .
\end{align}
Next we must evaluate (\ref{final_tdpcm_ham}) when the momentum $\Pi_A$ satisfies (\ref{tdual_canonical}). We will decorate the symbol for the Hamiltonian density with a star to indicate that these relations are satisfied:
\begin{align}
    \mathcal{H}_{\text{TD-AFSM}}^\ast = \mathcal{H}_{\text{TD-AFSM}} \left( \Pi_A = j_\mu^A \partial_\sigma \phi^\mu \right) \, .
\end{align}
Substituting the expression (\ref{tdual_canonical}) for $\Pi_A$ into (\ref{final_tdpcm_ham}) yields
\begin{align}
    \mathcal{H}_{\text{TD-AFSM}}^\ast &= \frac{1}{2} j_{\sigma, C} M^{AC} \tensor{M}{_A^D} j_{\sigma, D} - j_{\sigma, A} \tensor{f}{^A_B_C} \Lambda^C \left( \partial_\sigma \Lambda^B - 2 v_\tau^B \right) + \frac{1}{2} ( \partial_\sigma \Lambda^A ) ( \partial_\sigma \Lambda_A ) + v_\tau^A v_{\tau, A} \nonumber \\
    &\qquad + v_\sigma^A v_{\sigma, A}  + 2 v_\sigma^A j_{\sigma, A}  - 2 v_\tau^A \partial_\sigma \Lambda_A - E ( \nu_2 , \ldots , \nu_N ) \, .
\end{align}
It is convenient to expand $M$ explicitly using its definition (\ref{MAB_defn}), which gives
\begin{align}
    M^{AC} \tensor{M}{_A^D} = \gamma^{CD} - \tensor{f}{^C_A_B} \Lambda^B \tensor{f}{^A^D_E} \Lambda^E \, .
\end{align}
Thus we have
\begin{align}\label{first_canonical_two}
    \mathcal{H}_{\text{TD-AFSM}}^\ast &= \frac{1}{2}\, j_{\sigma, C} j_{\sigma}^{C} - \frac{1}{2} j_{\sigma, C} j_{\sigma, D} \tensor{f}{^C_A_B} \Lambda^B \tensor{f}{^A^D_E} \Lambda^E - j_{\sigma, A} \tensor{f}{^A_B_C} \Lambda^C \partial_\sigma \Lambda^B  + 2 j_{\sigma, A} \tensor{f}{^A_B_C} \Lambda^C v_\tau^B \nonumber \\
    &\qquad + \frac{1}{2} ( \partial_\sigma \Lambda^A ) ( \partial_\sigma \Lambda_A ) + v_\tau^A v_{\tau, A} + v_\sigma^A v_{\sigma, A}  + 2 v_\sigma^A j_{\sigma, A}  - 2 v_\tau^A \partial_\sigma \Lambda_A - E ( \nu_2 , \ldots , \nu_N ) \, .
\end{align}
We now turn to evaluating $\mathcal{H}_{\text{PCM}}^\ast$. Plugging in the expression for $\pi_\mu$ in (\ref{tdual_canonical}) to the Hamiltonian (\ref{final_afsm_ham}) gives
\begin{align}\label{second_canonical_one}
    \mathcal{H}_{\text{AFSM}}^\ast &= \frac{1}{2} j^\mu_A \left( j_\mu^B \partial_\sigma \Lambda_B - f_{BCD} \Lambda^B j_\sigma^C j_\mu^D \right) j^{\nu, A} \left( j_\nu^E \partial_\sigma \Lambda_E - f_{EFG} \Lambda^E j_\sigma^F j_\nu^G \right) + \frac{1}{2} j_\sigma^A j_{\sigma, A} + v_\tau^A v_{\tau, A} \nonumber \\
    &\qquad + v_\sigma^A v_{\sigma, A} + 2 \left( j_\sigma^A v_{\sigma, A} - j_\mu^A \left( j^{\mu, B} \partial_\sigma \Lambda_B - f_{BCD} \Lambda^B j_\sigma^C j^{\mu, D} \right) v_{\tau, A} \right)  - E ( \nu_2 , \ldots , \nu_N )  \, .
\end{align}
We can cancel vielbeins and inverse vielbeins using
\begin{align}
    j^\mu_A j_\mu^B = \tensor{\delta}{_A^B} \, ,
\end{align}
which gives
\begin{align}\label{second_canonical_two}
    \mathcal{H}_{\text{AFSM}}^\ast &= \frac{1}{2} \left( \partial_\sigma \Lambda_A - f_{BCA} \Lambda^B j_\sigma^C \right)  \left( \partial_\sigma \Lambda^A - \tensor{f}{_E_F^A} \Lambda^E j_\sigma^F \right) + \frac{1}{2} j_\sigma^A j_{\sigma, A} + v_\tau^A v_{\tau, A} \nonumber \\
    &\qquad + v_\sigma^A v_{\sigma, A} + 2 \left( j_\sigma^A v_{\sigma, A} - \left( \partial_\sigma \Lambda^A - \tensor{f}{_B_C^A} \Lambda^B j_\sigma^C \right) v_{\tau, A} \right)  - E ( \nu_2 , \ldots , \nu_N )  \, \nonumber \\
    &= \frac{1}{2} \partial_\sigma \Lambda^A \partial_\sigma \Lambda_A - \partial_\sigma \Lambda_A \tensor{f}{_B_C^A} \Lambda^B j_\sigma^C + \frac{1}{2} f_{BCA} \Lambda^B   \tensor{f}{_E_F^A} \Lambda^E j_\sigma^F j_\sigma^C + \frac{1}{2} j_\sigma^A j_{\sigma, A} + v_\tau^A v_{\tau, A} \nonumber \\
    &\qquad + v_\sigma^A v_{\sigma, A} + 2 j_\sigma^A v_{\sigma, A} - 2 \partial_\sigma \Lambda^A v_{\tau, A} + 2 \tensor{f}{_B_C^A} \Lambda^B j_\sigma^C v_{\tau, A}   - E ( \nu_2 , \ldots , \nu_N ) \, .
\end{align}
Comparing (\ref{first_canonical_two}) and (\ref{second_canonical_two}), after using the antisymmetry of the structure constants and relabeling indices, we see that
\begin{align}\label{equal_after_canonical}
    \mathcal{H}_{\text{TD-AFSM}}^\ast = \mathcal{H}_{\text{AFSM}}^\ast \, .
\end{align}
Therefore, this canonical transformation indeed makes the Hamiltonians equal.

We reiterate that the analysis in this section proves that the auxiliary field sigma model with interaction function $E$, in the Hamiltonian formulation, is equivalent to the T-dual AFSM with the same interaction function $E$; they are the same physical theory, written in different phase space variables. In particular, this immediately implies that the TD-AFSM is classically integrable. One can see that there exists an infinite set of Poisson-commuting conserved charges in the TD-AFSM by pushing forward the charges of the AFSM under this symplectomorphism, as we alluded to around equation (\ref{canonical_transform_charges}).

Alternatively, one can argue indirectly for the involution of an infinite set of charges for the TD-AFSM by noting that the Poisson bracket of the spatial component of the Lax connection for the TD-AFSM must also take the Maillet form. Schematically, the reason for this is the following. Consider the image $\widetilde{\CC{L}}_\sigma$ of the spatial component of the Lax connection for the AFSM, given in equation (\ref{spatial_lax}), under the canonical transformation which relates the AFSM to its T-dual. We have already seen that $\CC{L}_\sigma$ satisfies equation (\ref{maillet_bracket}) for a particular choice of the $r$-matrix, so let us then push forward this entire equation under the symplectomorphism which implements the T-duality. The Poisson bracket is invariant under such a symplectomorphism, and the matrices $r_{12}$ and $s_{12}$ are also invariant under this map as they depend only on the spectral parameters $z$, $z'$ but not on the phase space variables. Thus one immediately concludes that 
\begin{align}\label{maillet_bracket_TD}
    \left\{ \widetilde{\CC{L}}_{\sigma, 1} ( \sigma, z ) \, , \, \widetilde{\CC{L}}_{\sigma, 2} ( \sigma', z' ) \right\} &= [ r_{12} ( z, z' ) \, , \, \widetilde{\CC{L}}_{\sigma, 1} ( \sigma, z ) ] \delta ( \sigma - \sigma' ) - [ r_{21} ( z', z ) \, , \, \widetilde{\CC{L}}_{\sigma, 2} ( \sigma, z' ) ] \delta ( \sigma - \sigma' ) \nonumber \\
    &\qquad - s_{12}( z, z' ) \partial_\sigma \delta ( \sigma - \sigma' ) \, ,
\end{align}
which establishes that the Poisson bracket of the Lax connection for the T-dual AFSM also takes the Maillet form. 

Of course, one could also prove (\ref{maillet_bracket_TD}) by a direct evaluation of the left and right sides using the expression for the Lax connection of the TD-AFSM. However, we will not carry out these procedures explicitly, since the integrability of the TD-AFSM is now clear.

\section{Adding a Wess-Zumino Term}\label{sec:extensions}

The auxiliary field deformations of the principal chiral model which we have constructed in this work represent just one entry in a long list of other integrable deformations of the PCM; see \cite{Zarembo:2017muf,Orlando:2019his,Seibold:2020ouf,Klimcik:2021bjy,Hoare:2021dix} for reviews. As a general rule, whenever one has multiple deformations of a given theory at one's disposal -- all of which preserve some desirable structure such as integrability -- it is advantageous to attempt to activate all of these deformations simultaneously. If one succeeds in doing so, then one generates an even larger class of deformed models, and one can sometimes even learn more about each of the original deformations by viewing them as limiting cases within this broader family.\footnote{Other examples where this philosophy of constructing multiple simultaneous deformations has been useful include Fateev's $2$-parameter deformation of the $SU(2)$ PCM \cite{FATEEV1996509}, Lukyanov's $4$-parameter generalization \cite{Lukyanov:2012zt}, bi-Yang-Baxter deformations \cite{Klimcik:2014bta} and their combination with Wess-Zumino term and TsT transformations \cite{Delduc:2017fib}, and (in a somewhat different context) the deformed Inozemtsev spin chain of \cite{Klabbers:2023jqx}.}

It is interesting to investigate whether our higher-spin auxiliary field deformations can be combined with other integrable deformations of the PCM in this way. We will consider one example of such a simultaneous deformation in this work, namely the doubly-deformed models obtained by combining our auxiliary field model with a Wess-Zumino term. In this case, we find that the resulting doubly-deformed model also admits a Lax representation for its equations of motion. Other extensions and generalizations of our auxiliary formalism, including combination with Yang-Baxter deformations \cite{yb_toappear} and applications to symmetric and semi-symmetric space sigma models \cite{toappear}, will appear elsewhere.

\subsection{Definition of AFSM with WZ Term}\label{sec:wess_zumino}

The addition of a Wess-Zumino (WZ) term \cite{WESS197195,Novikov:1982ei} gives an integrable deformation of the principal chiral model which has a long history. For a particular choice of coefficient, the PCM with WZ term gives the Wess-Zumino-Witten (WZW) model, which is an exact conformal field theory at the quantum level \cite{Witten:1983ar} (in contrast, the PCM is only classically conformally invariant). It was shown in \cite{ABDALLA1982181} that the PCM with WZ term is classically integrable for any value of the coupling constant. In the context of string sigma models, the Wess-Zumino term describes the coupling to a target-space $H$ flux which is closed but not exact, and therefore cannot be captured using a conventional $B$-field coupling for a globally defined Kalb-Ramond two-form $B_2$.

First let us recall the formulation of the principal chiral model with Wess-Zumino term, before deforming by coupling to auxiliary fields. It is convenient to introduce two constants $\hay$ and $\kay$ which multiply the PCM kinetic term and the WZ term, respectively, so that
\begin{align}\label{PCM_WZ}
    S_{\text{PCM-WZ}} = - \frac{\hay}{2} \int d^2 \sigma \, \tr ( j_+ j_- ) + \frac{\kay}{6} \int_{\mathcal{M}_3} d^3 x \, \epsilon^{ijk} \tr \left( j_i [ j_j, j_k ] \right) \, .
\end{align}
A few comments are in order. First, the integral in the WZ term is performed over a $3$-manifold $\mathcal{M}_3$ whose boundary is the worldsheet, $\partial \mathcal{M}_3 = \Sigma$. We use lowercase middle Latin letters like $i$, $j$, $k$ for indices on $\mathcal{M}_3$, in contrast with the early Latin letters like $\alpha$, $\beta$ for indices on $\Sigma$. Second, in order for this action to give rise to a well-defined $2d$ theory, it must be invariant under variations which modify the field only in the $3d$ bulk $\mathcal{M}_3$ but not on the $2d$ boundary $\Sigma$. That is, given a variation $g \to g e^{\epsilon} = g ( 1 + \epsilon )$, under which
\begin{align}\label{delta_j_3d}
    \delta j_i = \partial_i \epsilon + [ j_i , \epsilon ] \, ,
\end{align}
we must have that $\delta S_{\text{WZ}} = 0$ assuming that $\epsilon \big\vert_{\Sigma} = 0$, where $\big\vert_{\Sigma}$ indicates the restriction to the $2d$ boundary. We review why this is indeed the case in appendix \ref{app:wz}, where we compute the variation of the WZ term. Finally, in order to facilitate comparison to the ordinary principal chiral model, note that the PCM limit is recovered by taking $\kay = 0$ and $\hay = 1$.

Next we will deform the action (\ref{PCM_WZ}) by coupling to auxiliary fields. We do this in the obvious way: we replace the PCM part of the action with the auxiliary field sigma model (\ref{deformed_family}), while leaving the Wess-Zumino term unchanged, to arrive at
\begin{align}\label{family_with_WZ}
    S_{\text{AFSM-WZ}} &= \hay \int d^2 \sigma \, \left( \frac{1}{2} \mathrm{tr} ( j_+ j_- ) + \mathrm{tr} ( v_+ v_- ) + \mathrm{tr} ( j_+ v_- + j_- v_+ ) + E ( \nu_2 , \ldots , \nu_N )  \right) \nonumber \\
    &\qquad + \frac{\kay}{6} \int_{\mathcal{M}_3} d^3 x \, \epsilon^{ijk} \tr \left( j_i [ j_j, j_k ] \right) \, .
\end{align}
Here we define each argument of the interaction function $E$ by $\nu_k = \tr ( v_+^k ) \tr ( v_-^k )$, exactly as before. The action (\ref{family_with_WZ}) is doubly-deformed in the sense that it reduces to each of our singly-deformed families of models in appropriate limits: for $\hay = 1$ and $\kay = 0$ it reproduces the auxiliary field sigma model (\ref{deformed_family}), while for $E = 0$ it coincides with the PCM with Wess-Zumino term (\ref{PCM_WZ}) after integrating out the auxiliary fields.

The equation of motion for the $g$-field arising from (\ref{family_with_WZ}) is derived in appendix \ref{app:wz} and takes the form
\begin{align}\label{AFSM_WZ_g_eom}
    \left( \hay - \kay \right) \partial_+ j_- + \left( \hay + \kay \right) \partial_- j_+ + 2 \hay \left( \partial_+ v_- + \partial_- v_+ + [ j_+ , v_- ] + [ j_- ,v_+ ] \right) = 0 \, ,
\end{align}
which explicitly depends on  $\hay$ and $\kay$.
Conversely, since the auxiliary fields $v_{\pm}$ do not appear in the Wess-Zumino term, their equation of motion is (\ref{v_eom}), exactly as in the case of the auxiliary field sigma model without Wess-Zumino term.

In particular, since the analysis of section \ref{sec:implications} relied only on the form of the auxiliary field equations of motion, all of the implications which we recorded in that section are also applicable here. For instance, it is still the case that 
\begin{align}\label{fundamental_commutator_identity_repeat}
    [ v_{\mp} , j_{\pm} ] \deq [ v_{\pm} , v_{\mp} ]
\end{align}
in the AFSM with WZ term. Similarly, if we again define
\begin{align}\label{WZ_frak_J}
    \mathfrak{J}_{\pm} = - \left( j_\pm + 2 v_{\pm} \right) \, ,
\end{align}
then the relations
\begin{align}\label{imps_wz}
    [ \CC{J}_+ , j_- ] &\deq [ j_+ , \CC{J}_- ] \, , \nonumber \\
    [ \CC{J}_+ , \CC{J}_- ] &\deq [ j_+ , j_- ] \, ,
\end{align}
also still hold. In terms of $\CC{J}_\pm$, the $g$-field equation of motion (\ref{AFSM_WZ_g_eom}) can be written as
\begin{align}\label{PCM_WZ_eom_dotted}
    \partial_+ \left( \hay \mathfrak{J}_- + \kay j_- \right) + \partial_- \left( \hay \mathfrak{J}_+ - \kay j_+ \right) \deq 0 \, ,
\end{align}
when the auxiliary field equation of motion is satisfied. We note that $\CC{J}_{\pm}$ is no longer conserved, unless $\kay = 0$.

\subsubsection*{\ul{\it The WZW point}}

For the undeformed PCM, the choice $\hay = \pm \kay$ gives a special point in the parameter space which defines the Wess-Zumino-Witten (WZW) model. Since $\CC{J}_{\pm} = j_{\pm}$ in the undeformed theory, one can see from equation (\ref{PCM_WZ_eom_dotted}) that the assignments $\hay = \kay$ and $\hay = - \kay$ cause the term proportional to $j_+$ and to $j_-$, respectively, to drop out of the conservation equation, which means that the remaining component is chirally conserved. This leads to an enhancement of the symmetries of the theory, which is related to the fact that the WZW model is a conformal field theory at the quantum level.

In the case of a general interaction function $E \neq 0$, no such chiral conservation law emerges in the limit $\hay = \pm \kay$, and both of the terms in equation (\ref{PCM_WZ_eom_dotted}) remain non-zero. This may be viewed as a signal that a generic member of our deformed class of models is not conformal, although at least for certain choices of interaction functions such as the one corresponding to the $\TT$ deformation, it may be that a form of ``dressed'' conformal symmetry still remains for the theory. However, even for classically conformal choices of interaction functions, such as $E ( \nu_k ) \sim \sqrt[k]{\nu_k}$, there is still no chirally conserved current; for $k = 2$, which corresponds to the root-$\TT$ deformation, this observation appeared in \cite{Borsato:2022tmu}.

Nonetheless, some simplifications still occur at the would-be WZW points $\hay = \pm \kay$ for general choices of interaction function. One example, as we will see below in equation (\ref{WZ_lax_commutator_intermediate}), is that the commutator $[\CC{L}_+ , \CC{L}_-]$ of the components of the Lax connection vanishes at this point of parameter space.

\subsection{Zero-Curvature Representation}

Let us now show that the auxiliary field sigma model with Wess-Zumino term is also weakly integrable, in the sense that it satisfies the condition (\ref{aux_weak_integrability}). In this case, the appropriate Lax connection is
\begin{align}\label{AFSM_WZ_lax}
    \mathfrak{L}_{\pm} = \frac{ \left( j_{\pm} \mp \frac{\kay}{\hay} \mathfrak{J}_{\pm} \right) \pm z \left( \mathfrak{J}_{\pm} \mp \frac{\kay}{\hay} j_{\pm} \right) }{1 - z^2} \, .
\end{align}
As expected, this candidate Lax connection reduces to the AFSM Lax connection (\ref{lax_connection}) upon taking $\kay = 0$ and $\hay = 1$. We will proceed as in \ref{sec:weak_integrability} by computing the curvature $d_{\CC{L}} \CC{L}$ of (\ref{AFSM_WZ_lax}) and showing that it vanishes for all values of the spectral parameter $z$ if and only if the equation of motion (\ref{PCM_WZ_eom_dotted}) is satisfied. In this calculation, we will freely use the Maurer-Cartan identity for $j_{\pm}$, which holds identically due to the definition of $j$.

We first compute the commutator of $\CC{L}_+$ with $\CC{L}_-$, which is
\begin{align}\label{WZ_lax_commutator_intermediate}
    [ \mathfrak{L}_+ , \mathfrak{L}_- ] &= \frac{1}{( 1 - z^2 )^2} \left[ j_{+} - \frac{\kay}{\hay} \mathfrak{J}_{+} + z \left( \mathfrak{J}_{+} - \frac{\kay}{\hay} j_{+} \right)  ,  j_{-} + \frac{\kay}{\hay} \mathfrak{J}_{-} - z \left( \mathfrak{J}_{-} + \frac{\kay}{\hay} j_{-} \right)  \right] \nonumber \\
    &= \frac{1}{( 1 - z^2 )^2}  \Bigg( [j_+ , j_-] + \frac{\kay}{\hay} \left( [ j_+ , \mathfrak{J}_- ] - [ \mathfrak{J}_+ , j_- ]  \right) - \frac{\kay^2}{\hay^2} [ \mathfrak{J}_+ , \mathfrak{J}_- ] \nonumber \\
    &\qquad \qquad \qquad + z \left( [ \mathfrak{J}_+ , j_- ] + \frac{\kay}{\hay} \left( [ \mathfrak{J}_+ , \mathfrak{J}_- ] - [ j_+ , j_- ] \right) - \frac{\kay^2}{\hay^2} [ j_+ , \mathfrak{J}_- ] \right) \nonumber \\
    &\qquad \qquad \qquad - z \left( [ j_+ , \mathfrak{J}_- ] - \frac{\kay}{\hay} \left( [ \mathfrak{J}_+ , \mathfrak{J}_- ] - [ j_+ , j_- ] \right) - \frac{\kay^2}{\hay^2} \left( [ \mathfrak{J}_+ , j_- ] \right) \right) \nonumber \\
    &\qquad \qquad \qquad - z^2 \left( [ \mathfrak{J}_+ , \mathfrak{J}_-] + \frac{\kay}{\hay} \left( [ \mathfrak{J}_+ , j_- ] - [ j_+ , \mathfrak{J}_- ] \right) - \frac{\kay^2}{\hay^2} [ j_+ , j_- ]  \right) \Bigg) \, .
\end{align}
We have emphasized that the implications (\ref{imps_wz}) of the auxiliary field equation of motion still hold for the AFSM with WZ term, just as in the usual AFSM. We use these relations to simplify the commutators involving $\CC{J}_{\pm}$ in (\ref{WZ_lax_commutator_intermediate}), finding
\begin{align}\label{WZ_lax_commutator_intermediate_two}
    [ \mathfrak{L}_+ , \mathfrak{L}_- ] &\deq \frac{1}{( 1 - z^2 )^2}  \Bigg( \left( 1 - \frac{\kay^2}{\hay^2} \right) [j_+ , j_-] - z^2 \left( 1 - \frac{\kay^2}{\hay^2} \right) [j_+ , j_-] \Bigg) \nonumber \\
    &= \frac{1 - \frac{\kay^2}{\hay^2}}{1 - z^2} [ j_+ , j_- ] \, .
\end{align}
As we alluded to above, we note that the commutator (\ref{WZ_lax_commutator_intermediate_two}) vanishes at the points $\hay = \pm \kay$. Having obtained this useful intermediate result, we are now prepared to compute the curvature two-form of the Lax connection, which is given by
\begin{align}\label{WZ_dL_intermediate}
    d_{\CC{L}} \mathfrak{L} &= \partial_+ \mathfrak{L}_- - \partial_- \mathfrak{L}_+ + [ \mathfrak{L}_+ , \mathfrak{L}_- ] \nonumber \\
    &\deq \frac{1}{1 - z^2} \Bigg( \partial_+ \left(  j_{-} + \frac{\kay}{\hay} \mathfrak{J}_{-} - z \left( \mathfrak{J}_{-} + \frac{\kay}{\hay} j_{-} \right) \right) - \partial_- \left( j_{+} - \frac{\kay}{\hay} \mathfrak{J}_{+} + z \left( \mathfrak{J}_{+} - \frac{\kay}{\hay} j_{+} \right)  \right) \nonumber \\
    &\qquad \qquad \qquad + \left( 1 - \frac{\kay^2}{\hay^2} \right) [ j_+ , j_- ]  \Bigg) \nonumber \\
    &= \frac{1}{1 - z^2} \Bigg( \partial_+ j_- - \partial_- j_+ + [ j_+ , j_- ] - \frac{\kay^2}{\hay^2} [ j_+ , j_- ] - z \frac{\kay}{\hay} \partial_+ j_- + z \frac{\kay}{\hay} \partial_- j_+ \nonumber \\
    &\qquad \qquad + \left( \frac{\kay}{\hay} - z \right) \partial_+ \mathfrak{J}_- + \left( \frac{\kay}{\hay} - z \right) \partial_- \mathfrak{J}_+ \Bigg) \, .
\end{align}
We now use the Maurer-Cartan identity for $j_{\pm}$, which causes the first three terms in parentheses in the penultimate line of (\ref{WZ_dL_intermediate}) to vanish. This identity also allows us to replace $[j_+, j_-]$ in the fourth term of the penultimate line, which gives
\begin{align}\label{PCM_WZ_flat_lax_final}
    d_{\CC{L}} \mathfrak{L} & \deq \frac{1}{1 - z^2} \Bigg( \frac{\kay^2}{\hay^2} \left( \partial_+ j_- - \partial_- j_+ \right) - z \frac{\kay}{\hay} \partial_+ j_- + z \frac{\kay}{\hay} \partial_- j_+ + \left( \frac{\kay}{\hay} - z \right) \partial_+ \mathfrak{J}_- + \left( \frac{\kay}{\hay} - z \right) \partial_- \mathfrak{J}_+ \Bigg) \nonumber \\
    &= \frac{1}{1 - z^2} \left( \left( \frac{\kay}{\hay} - z \right) \partial_+ \left( \frac{\kay}{\hay} j_- + \mathfrak{J}_- \right) + \left( \frac{\kay}{\hay} - z \right) \partial_- \left( \mathfrak{J}_+ - \frac{\kay}{\hay} j_+ \right)  \right) \nonumber \\
    &= \frac{\frac{\kay}{\hay} - z }{\hay ( 1 - z^2 ) } \left( \partial_+ \left( \hay \mathfrak{J}_- + \kay j_- \right) + \partial_- \left( \hay \mathfrak{J}_+ - \kay j_+ \right) \right) \, .
\end{align}
The last line of (\ref{PCM_WZ_flat_lax_final}) vanishes, for arbitrary $z \in \mathbb{C}$, if and only if
\begin{align}
    \partial_+ \left( \hay \mathfrak{J}_- + \kay j_- \right) + \partial_- \left( \hay \mathfrak{J}_+ - \kay j_+ \right) \deq 0 \, ,
\end{align}
which is the equation of motion (\ref{PCM_WZ_eom_dotted}). We conclude that every member of our family of higher-spin auxiliary field models with Wess-Zumino term satisfies the weak integrability condition (\ref{aux_weak_integrability}), that is, that their classical equations of motion are equivalent to the flatness of the Lax connection in equation (\ref{AFSM_WZ_lax}).

The arguments of this section serve as a proof-of-concept that, at least in one example, our auxiliary field deformations can be combined with other deformations of the PCM in a way that still admits a zero-curvature representation for the equations of motion. In order to conclusively establish classical integrability of the doubly-deformed models (\ref{family_with_WZ}), we would also need to demonstrate the existence of an infinite set of conserved charges in involution.\footnote{In the special case where the interaction function $E$ depends only on $\nu_2$, involution of the charges can also be shown from the $4d$ Chern-Simons presentation of the AFSM developed in \cite{Fukushima:2024nxm}.} For instance, one could do this by computing the Poisson bracket of the spatial component of the Lax connection and showing that it takes the non-ultralocal Maillet form, as we did in section \ref{sec:maillet} for the AFSM without Wess-Zumino term. We leave this task to future work, since our present goal is simply to argue for the feasibility of combining the AFSM with other integrable deformations, rather than to carry out a detailed analysis.

\section{Conclusion}\label{sec:conclusion}

In this work, we have generalized the auxiliary field deformations of \cite{Ferko:2024ali} to a larger family which is parameterized by a general multivariate function $E$ of higher-spin combinations of auxiliary fields. We have argued that, at least to leading order, this class of models includes deformations of the PCM by higher-spin current bilinears of the form considered in \cite{Smirnov:2016lqw}. Every model in this family is classically integrable, both in the weak sense of admitting a Lax representation for its equations of motion, and in the strong sense of possessing an infinite set of conserved charges in involution. We have also constructed the non-Abelian T-duals of the deformed models in this family, studied their features, and argued that they are classically integrable. Finally, we have constructed the ``doubly-deformed'' models obtained by simultaneously activating our auxiliary field deformations along with a Wess-Zumino term, and exhibited a Lax connection for these theories.

There remain several interesting directions for future research. Perhaps the most obvious is to combine our higher-spin auxiliary field deformations with other integrable deformations of the PCM, such as $\eta$- and $\lambda$-deformations, as we briefly mentioned at the end of section \ref{sec:wess_zumino}. Another direction is to seek explicit expressions for the deformed higher-spin conserved currents, generalizing the conserved quantities $\mathcal{J}_{s\pm} = \tr ( j_\pm^s )$ in the PCM, in an auxiliary field sigma model with arbitrary interaction function. Still other avenues for research include attempting to couple multiple sigma models via sequential auxiliary field deformations, similar to those introduced in \cite{Ferko:2022dpg}, or seeking analogues of these higher-spin deformations in the related $4d$ and $6d$ auxiliary field formalisms, which would extend results on stress tensor deformations of these theories presented in \cite{Ferko:2023wyi} and \cite{Ferko:2024zth}, respectively.

Below we briefly sketch a few additional lines of inquiry for further investigation. We hope to return to some of these in future work.

\subsubsection*{\ul{\it W-gravity}}

It seems likely that the class of models considered in this work includes deformations of the PCM by general functions of both the energy-momentum tensor and the higher-spin conserved currents. Such deformations can often be endowed with a geometrical interpretation. This was first established for the $2d$ $\TT$ deformation, which is classically equivalent to enacting a certain field-dependent change of coordinates in the undeformed theory \cite{Conti:2018tca}, and a similar statement applies for so-called $T \overbar{T}_s$ deformations, which are built from both the stress tensor and higher-spin currents of the type that we consider here \cite{Conti:2019dxg}. These observations were later elaborated and extended in various directions, both to higher dimensions \cite{Conti:2022egv,Morone:2024ffm} and to other stress tensor deformations like root-$\TT$ \cite{Babaei-Aghbolagh:2024hti,Tsolakidis:2024wut}; in these more general cases, the deformations are typically viewed as a coupling to some gravitational sector, rather than just a change of coordinates. See also \cite{Tolley:2019nmm,Ferko:2024yhc} for related work.

It would be very interesting if generic deformations by collections of higher-spin currents could be given a unified geometrical interpretation which treats all such flows on similar footing. One framework which may be useful for such an endeavor is $W$-gravity or $W$-geometry; see \cite{Hull:1993kf,Pope:1991ig} for reviews. Within the formalism of $W$-gravity, the ordinary metric tensor is promoted to a $W$-metric which, in addition to a linearized spin-$2$ metric fluctuation $h_{\alpha \beta}$ that couples to the conserved stress tensor $T_{\alpha \beta}$, contains a spin-$3$ component $W_{\alpha \beta \gamma}$ that couples to a totally symmetric spin-$3$ current $J_{\alpha \beta \gamma}$, and so on for all positive integer spins. It is conceivable that, just as a stress tensor deformation can be interpreted as coupling a theory to a field-dependent metric, a higher-spin deformation might be classically equivalent to a coupling to a field-dependent $W$-metric. The use of $W$-gravity may also help in obtaining higher-spin extensions of Hilbert's energy-momentum tensor to then engineer $\TT$-like flows based on the operators in equation (\ref{OK_defn_intro}).

\subsubsection*{\ul{\it Twisted cylinder boundary conditions}}

Although the collection of integrable deformations considered in this article is quite large -- for instance, in the PCM with gauge group $SU(N)$, they are in correspondence with functions of $N - 1$ variables 
-- they also act quite simply on the integrable structure of the theory, as evidenced by the fact that they do not affect the twist function $\varphi ( z )$. It also seems that these deformations are, in some sense, fairly ``universal'' in that they can be applied to any member of a large class of integrable sigma models, such as the principal chiral model with Wess-Zumino term, the PCM subject to a Yang-Baxter deformation  \cite{yb_toappear}, or to semi-symmetric space sigma models \cite{toappear}.

This suggests that there should be a uniform physical interpretation for the effect of our deformations which applies to many different models. One possibility for such an interpretation, at least for theories on the cylinder $\Sigma = S^1 \times \mathbb{R}$, is in terms of twisted boundary conditions in the spatial direction.\footnote{We thank Alessandro Sfondrini for this suggestion.} Some evidence for this conjecture is provided by the fact that fairly general deformations of $2d$ QFTs by combinations of higher-spin conserved charges can be interpreted by implementing twists of the fields' boundary conditions along the cylinder, as shown in \cite{Hernandez-Chifflet:2019sua}; a generalized flow equation for deformed charges under such a flow was also derived in this work (see also \cite{Cordova:2021fnr,Camilo:2021gro} for analyses of these higher-spin deformations using thermodynamic Bethe ansatz). It may be that our auxiliary field deformations are likewise implementing such twisted boundary conditions, which might explain the fact that these deformations leave $\varphi ( z )$ unchanged, much like homogeneous Yang-Baxter deformations \cite{Kawaguchi:2014qwa,vanTongeren:2015soa,Borsato:2021fuy} (which are also related to such twisted boundary conditions).

\subsubsection*{\ul{\it Holographic Interpretation}}

The special class of auxiliary field models where the interaction function $E$ depends only on $\nu_2$ is known \cite{Ferko:2024ali} to contain all deformations of the principal chiral model by functions of the energy-momentum tensor. At least two such deformations, the irrelevant $\TT$ deformation and its marginal root-$\TT$ cousin, admit interpretations in terms of modified variational principles in $\mathrm{AdS}_3 / \mathrm{CFT}_2$ holography. In the $\TT$ case, this variational principle was written down in \cite{Guica:2019nzm}, whereas the analogue for root-$\TT$ appeared in \cite{Ebert:2023tih} (see also \cite{Tian:2024vln,Ebert:2024fpc,Babaei-Aghbolagh:2024hti}).

It is natural to wonder whether the higher-spin deformations considered in this work can also be understood from a holographic perspective. One way to proceed might be to write the $\mathrm{AdS}_3$ gravity theory as an $SL ( 2 , \mathbb{R} ) \times SL ( 2 , \mathbb{R} )$ Chern-Simons theory (deformations by $\TT$ have also been studied in this context \cite{Llabres:2019jtx,Ebert:2022ehb}). The advantage of this re-framing is that one can straightforwardly promote the bulk gauge group to $SL ( N , \mathbb{R} ) \times SL ( N , \mathbb{R} )$, which has higher-spin degrees of freedom in addition to the spin-$2$ metric; see the review \cite{Campoleoni:2024ced}, or the papers \cite{Campoleoni:2010zq,Henneaux:2010xg,Campoleoni:2012hp} and references therein, for more details. It would be intriguing if a modified variational principle for such a higher-spin bulk Chern-Simons theory could engineer deformations of the boundary theory of the type considered in this work.

\section*{Acknowledgements}
We thank Peter Bouwknegt, Mattia Cesaro, Chris Hull, Johanna Knapp, Yolanda Lozano, Tommaso Morone, Jock McOrist, Silvia Penati, Anayeli Ramirez, Savdeep Sethi, Alessandro Sfondrini, Dmitri Sorokin, Roberto Tateo, and Martin Wolf for fruitful discussions. 
D.\,B. is supported by Thailand NSRF via PMU-B, grant number B13F670063.
C.\,F. is supported by U.S. Department of Energy grant DE-SC0009999 and funds from the University of California. 
L.\,S. is supported by a postgraduate scholarship at the University of Queensland.
G.\,T.-M. has been supported by the Australian Research Council (ARC) Future Fellowship FT180100353, ARC Discovery
Project DP240101409, and the Capacity Building Package of the University of Queensland. This research was supported in part by grant NSF PHY-2309135 to the Kavli Institute for Theoretical Physics (KITP).

\appendix

\section{Derivation of Equations of Motion}\label{app:eom}

To streamline the presentation in the body of this article, in this appendix we collect the derivation of the Euler-Lagrange equations for the various models considered in this work.

\subsection{Higher-Spin AFSM}\label{app:AFSM}

The main class of models which we study in this manuscript are defined by Lagrangians of the form (\ref{deformed_family}), which we repeat here for convenience:
\begin{align}\label{deformed_action}
    S = \int d^2 \sigma \, \left( \frac{1}{2} \mathrm{tr} ( j_+ j_- ) + \mathrm{tr} ( v_+ v_- ) + \mathrm{tr} ( j_+ v_- + j_- v_+ ) + E ( \nu_2 , \nu_3 , \ldots , \nu_N ) \right) \, ,
\end{align}
where $E$ is a differentiable function of the variables $\nu_k = \tr ( v_+^k ) \tr ( v_-^k )$.

\subsubsection*{\ul{\it Physical field equation of motion}}

Let us first derive the equation of motion for the group-valued field $g : \Sigma \to G$. To do this, we demand stationarity of the action under an infinitesimal variation of $g$, which can be parameterized by an element $\epsilon \in \mathfrak{g}$ of the Lie algebra associated with $G$ as
\begin{align}
    g \to g e^{\epsilon} = g \left( 1 + \epsilon \right) \, ,
\end{align}
where we work to leading order in $\epsilon$. Thus the variation of the field is $\delta g = g \epsilon$.

Using the fact that $g g^{-1} = 1$ so $\delta ( g g^{-1} ) = ( \delta g ) g^{-1} + g \delta g^{-1} = 0$, one can compute the variation of the inverse group element, which gives
\begin{align}
    \delta g^{-1} = - \epsilon g^{-1} \, .
\end{align}
The action (\ref{deformed_action}) is written in terms of the left-invariant Maurer-Cartan form $j_{\pm} = g^{-1} \partial_{\pm} g$, so we would like to also find the variation of this object. One finds
\begin{align}\label{afsm_variation}
    \delta j_{\pm} &= \delta \left( g^{-1} \partial_{\pm} g \right) \nonumber \\
    &= \left( \delta g^{-1} \right) \partial_{\pm} g + g^{-1} \partial_{\pm} \left( \delta g \right) \nonumber \\
    &= - \epsilon g^{-1} \partial_{\pm} g + g^{-1} \partial_{\pm} \left( g \epsilon \right) \nonumber \\
    &= - \epsilon j_{\pm} + g^{-1} \left( \partial_{\pm} g \right) \epsilon + g^{-1} g \partial_{\pm} \epsilon \nonumber \\
    &= [ j_{\pm}, \epsilon ] + \partial_{\pm} \epsilon \, .
\end{align}
We can now find the variation of the action (\ref{deformed_action}) under a fluctuation $\delta g = g \epsilon$:
\begin{align}
    \delta S = \int d^2 \sigma \, &\Bigg( \frac{1}{2} \tr \Big( \left( [ j_+, \epsilon ] + \partial_+ \epsilon \right) j_- + j_+ \left( [ j_-, \epsilon ] + \partial_- \epsilon \right) \Big) \nonumber \\
    &\qquad + \tr \Big( \left( [ j_+, \epsilon ] + \partial_+ \epsilon \right) v_- + \left( [ j_-, \epsilon ] + \partial_- \epsilon \right) v_+ \Big) \Bigg) \, .
\end{align}
Note that, for any Lie algebra valued quantities $X$, $Y$, and $Z$, one has
\begin{align}\label{cycle_trace_commutators}
    \tr \left( [X, Y] Z \right) &= \tr \left( X Y Z - Y X Z \right) \nonumber \\
    &= \tr \left( Y Z X - Y X Z \right) \nonumber \\
    &= \tr ( [ Z, X ] Y) \, ,
\end{align}
where we have used cyclicity of the trace. We repeatedly use the identity (\ref{cycle_trace_commutators}) to move factors of $\epsilon$ outside of commutators in (\ref{afsm_variation}), and also integrate by parts, to find
\begin{align}\label{afsm_var_intermediate}
    \delta S = \int d^2 \sigma \, &\Bigg( \frac{1}{2} \tr \Big( \epsilon  [ j_-, j_+ ] - \epsilon \partial_+ j_- + \epsilon [ j_+, j_- ] - \epsilon \partial_- j_+ \Big)  \nonumber \\
    &\qquad + \tr \Big( \epsilon [ v_-, j_+ ] - \epsilon \partial_+ v_- + \epsilon [ v_+, j_- ] - \epsilon \partial_- v_+ \Big) \Bigg) \, .
\end{align}
The two terms involving commutators $[j_{\pm}, j_{\mp}]$ in the first line of (\ref{afsm_var_intermediate}) cancel, and after collecting terms, we are left with
\begin{align}\label{afsm_var_intermediate_two}
    \delta S = - \int d^2 \sigma \, \tr \left( \epsilon \left( \frac{1}{2} \partial_+ j_- + \frac{1}{2} \partial_- j_+ + \partial_+ v_- + \partial_- v_+ - [ v_- , j_+ ] - [ v_+, j_- ] \right) \right) \, .
\end{align}
Demanding that (\ref{afsm_var_intermediate_two}) vanish for arbitrary $\epsilon$ gives
\begin{align}
    \partial_+ \left( j_- + 2 v_- \right) + \partial_- \left( j_+ + 2 v_+ \right) = 2 \left( [ v_-, j_+ ] + [ v_+ , j_- ] \right) \, ,
\end{align}
which justifies the $g$-field equation of motion (\ref{j_eom}) quoted in the main text.

\subsubsection*{\ul{\it Auxiliary field equation of motion}}

Next let us derive the algebraic equation of motion associated with the auxiliary fields $v_{\pm}$. It is helpful to expand these fields in a basis of generators $T_A$ for the Lie algebra $\mathfrak{g}$,
\begin{align}
    v_{\pm} = v_{\pm}^A T_A \, .
\end{align}
Recall that the interaction function $E$ depends on a collection of variables $\nu_n = \tr ( v_+^n ) \tr ( v_-^n )$. The derivative of one of these variables with respect to $v_{\pm}^A$ is given by
\begin{align}\label{nu_deriv}
    \frac{\partial \nu_n}{\partial v_{\pm}^A} &= \frac{\partial}{\partial v_{\pm}^A} \left( \tr ( v_+^n ) \tr ( v_-^n ) \right) \nonumber \\
    &= \tr ( v_{\mp}^n ) \frac{\partial}{\partial v_{\pm}^A} \left( v_{\pm}^{A_1} \ldots v_{\pm}^{A_n} \tr ( T_{A_1} \ldots T_{A_n} ) \right) \nonumber \\
    &= n \tr ( v_{\mp}^n ) v_{\pm}^{A_1} \ldots v_{\pm}^{A_{n-1}} \tr ( T_{(A} T_{A_1} \ldots T_{A_{n-1} )} ) \, .
\end{align}
Therefore, when we compute the Euler-Lagrange equation
\begin{align}
    0 &= \frac{\partial \mathcal{L}}{\partial v_{\pm}^A} \nonumber \\
    &= j_{\mp A} + v_{\mp A} + \sum_{n=2}^{N} \frac{\partial E}{\partial \nu_n} \frac{\partial \nu_n}{\partial v_{\pm}^A} \, ,
\end{align}
we may substitute (\ref{nu_deriv}) for each of the $\nu_k$ derivatives to find
\begin{align}\label{higher_spin_aux_eom}
    0 \deq j_{\pm}^A + v_{\pm}^A + \sum_{n = 2}^{N} n E_n \tr ( v_{\pm}^n ) v_{\mp}^{A_1} \ldots v_{\mp}^{A_{n-1}} \gamma^{AB} \tr ( T_{(B} T_{A_1} \ldots T_{A_{n-1} )} ) \, .
\end{align}
We remind the reader that we use the symbol $\deq$ to emphasize that this equation holds when the auxiliary field equation of motion is satisfied (in the present setting, this is a tautology, as equation (\ref{higher_spin_aux_eom}) \emph{is} the auxiliary field equation of motion). We have also defined
\begin{align}
    E_n = \frac{\partial E}{\partial \nu_n} \, .
\end{align}
Finally, we may also contract equation (\ref{higher_spin_aux_eom}) against $T_A$ to write
\begin{align}\label{higher_spin_aux_eom_contracted}
    0 \deq j_{\pm} + v_{\pm} + \sum_{n = 2}^{N} n E_n \tr ( v_{\pm}^n ) v_{\mp}^{A_1} \ldots v_{\mp}^{A_{n-1}} T^B \tr ( T_{(B} T_{A_1} \ldots T_{A_{n-1} )} ) \, .
\end{align}
This agrees with the $v_{\pm}$ equation of motion (\ref{v_eom}) which we mentioned before. 

\subsection{T-dual AFSM}\label{app:td_afsm}

Next we find the Euler-Lagrange equations of the T-dual auxiliary field action (\ref{TD_AFSM}). 

\subsubsection*{\ul{\it Physical field equation of motion}}

We begin with the equation of motion for the Lie algebra valued field $\Lambda_A$. The action involves inverses of the matrix $M_{AB} = \gamma_{AB} - \tensor{f}{_A_B^C} \Lambda_C$, so we will need the formula for the variation of the inverse of a matrix,
\begin{align}\label{inverse_variation}
    \delta \tensor{\left( M^{-1} \right)}{_C^D} = - M^{-1}_{CA} \cdot \delta M^{AB} \cdot \tensor{\left( M^{-1} \right)}{_B^D} \, .
\end{align}
Therefore one has
\begin{align}\label{Minv_variation}
    \frac{\partial \left( M^{-1} \right)^{CD}}{\partial \Lambda^A} 
    &= -  \tensor{\left( M^{-1} \right)}{^C_E} \cdot \frac{\partial M^{EB}}{\partial \Lambda^A} \cdot \tensor{\left( M^{-1} \right)}{_B^D}  \nonumber \\
    &=  \tensor{\left( M^{-1} \right)}{^C_E} \cdot \tensor{f}{^E^B_A} \cdot \tensor{\left( M^{-1} \right)}{_B^D} \, .
\end{align}
Using (\ref{Minv_variation}), we find
\begin{align}\label{tdafsm_dl_dLambda}
    \frac{\partial \mathcal{L}}{\partial \Lambda^A}
    &= \frac{1}{2} \left( \partial_- \Lambda_C \right) \left( \tensor{\left( M^{-1} \right)}{^C_E} \cdot \tensor{f}{^E^B_A} \cdot \tensor{\left( M^{-1} \right)}{_B^D} \right) \left( \partial_+ \Lambda_D \right) \nonumber \\
    &\quad - 2 v_{-, C} \left( \tensor{\left( M^{-1} \right)}{^C_E} \cdot \tensor{f}{^E^B_A} \cdot \tensor{\left( M^{-1} \right)}{_B^D} \right) v_{+, D}  \nonumber \\
    &\quad + v_{-, A} \left( \tensor{\left( M^{-1} \right)}{^C_E} \cdot \tensor{f}{^E^B_A} \cdot \tensor{\left( M^{-1} \right)}{_B^D} \right) \left( \partial_+ \Lambda_D \right) 
    \nonumber \\
    &\quad
    - \left( \partial_- \Lambda_C \right) \left( \tensor{\left( M^{-1} \right)}{^C_E} \cdot \tensor{f}{^E^B_A} \cdot \tensor{\left( M^{-1} \right)}{_B^D} \right) v_{+, B} \, .
\end{align}
Note that in the T-dual AFSM -- to be contrasted with the AF-TDSM (\ref{AFTDSM}), obtained by coupling the undeformed T-dual PCM to auxiliaries -- the scalars $\nu_k$ are independent of $M^{-1}$, and hence $\frac{\partial \mathcal{L}}{\partial \Lambda^A}$ does not receive any contribution from the interaction function.

Throughout our analysis of the T-dual AFSM equations of motion, it will be convenient to use the fields $\CC{K}_\pm$ defined in (\ref{frak_K_defn}). In terms of these quantities, (\ref{tdafsm_dl_dLambda}) can be written as
\begin{align}
    \frac{\partial \mathcal{L}}{\partial \Lambda^A} &= - \frac{1}{2} \mathfrak{K}_{-, E} \tensor{f}{^E^B_A} \mathfrak{K}_{+, B}  + 2 \left( M^{-1} \right)_{AE} v_-^A \tensor{f}{^E^B_A} \left( M^{-1} \right)_{BD} v_+^D \nonumber \\
    &\qquad - 2 v_{-, C} \tensor{\left( M^{-1} \right)}{^C_E} \cdot \tensor{f}{^E^B_A} \tensor{\left( M^{-1} \right)}{_B^D} v_{+, D} \, .
\end{align}
and one also finds
\begin{align}
    \frac{\partial \mathcal{L}}{\partial ( \partial_+ \Lambda_A) } &= \frac{1}{2} \left( M^{-1} \right)^{BA} \left( \partial_- \Lambda_B + 2 v_{-, B} \right) = - \frac{1}{2} \mathfrak{K}_-^A \, , \nonumber \\
    \frac{\partial \mathcal{L}}{\partial ( \partial_- \Lambda_A ) } &= \frac{1}{2} \left( M^{-1} \right)^{AB} \left(  \partial_+ \Lambda_B  - 2 v_{+ , B} \right) = \frac{1}{2} \mathfrak{K}_+^A \, .
\end{align}
The equation of motion for $\Lambda_A$ is therefore
\begin{align}\label{TD_Lambda_eom}
    0 &= \partial_+ \left( \frac{\partial \mathcal{L}_{\text{TD-PCM}}}{\partial ( \partial_+ \Lambda^A ) } \right) + \partial_- \left( \frac{\partial \mathcal{L}_{\text{TD-PCM}}}{\partial ( \partial_- \Lambda^A ) } \right)  - \frac{\partial \mathcal{L}_{\text{TD-PCM}}}{\partial \Lambda^A} \nonumber \\
    &= - \frac{1}{2} \partial_+ \mathfrak{K}_-^A + \frac{1}{2} \partial_- \mathfrak{K}_+^A + \frac{1}{2} \mathfrak{K}_{-, E} \tensor{f}{^E^B_A} \mathfrak{K}_{+, B} \, .
\end{align}
Since we have that
\begin{align}
    \left( [ \mathfrak{K}_+ , \mathfrak{K}_- ] \right)^A = \mathfrak{K}_+^E \mathfrak{K}_-^B  \tensor{f}{_E_B^A} \, ,
\end{align}
one finds that the equation of motion can be written as
\begin{align}\label{TD_frakK_flat}
    0 &= \partial_+ \mathfrak{K}_- - \partial_- \mathfrak{K}_+ + [ \mathfrak{K}_+ , \mathfrak{K}_- ] \, ,
\end{align}
which is the flatness of $\mathfrak{K}_{\pm}$. Note that the auxiliary field equation of motion was not used in arriving at this conclusion.

\subsubsection*{\ul{\it Auxiliary field equation of motion}}

Now let us vary the auxiliary fields $v_{\pm}^A$. One finds
\begin{align}
    \frac{\partial \mathcal{L}}{\partial v_+^B} &= - 2 v_{-, A} \left( M^{-1} \right)^{AB} + v_{-, B} - \left( \partial_- \Lambda_A \right) \left( M^{-1} \right)^{AB} + \frac{\partial E}{\partial v_+^B} \, , \nonumber \\
    \frac{\partial \mathcal{L}}{\partial v_-^B} &= - 2 \left( M^{-1} \right)^{BA} v_{+, A} + \left( M^{-1} \right)^{BA} \left( \partial_+ \Lambda_A \right) + v_{+, B} + \frac{\partial E}{\partial v_-^B} \, .
\end{align}
The associated equations of motion can therefore be written as
\begin{align}
    0 = \mathfrak{K}_\pm +  v_{\pm} + T^B \frac{\partial E}{\partial v_\mp^B} \, ,
\end{align}
where we have again used the definition (\ref{frak_K_defn}) of $\mathfrak{K}_\pm$. 

The quantities
\begin{align}
    \frac{\partial E}{\partial v_{\pm}^B} = \sum_{n=2}^{N} \frac{\partial E}{\partial \nu_n} \frac{\partial \nu_n}{\partial v_{\pm}^A} \, ,
\end{align}
can be expressed in terms of (\ref{nu_deriv}), precisely as in the case of the derivation of the auxiliary field equation of motion for the ordinary AFSM. In particular, this means that the $v_\pm$ Euler-Lagrange equation can be written as
\begin{align}
    0 \deq \mathfrak{K}_{\pm} + v_{\pm} + \sum_{n = 2}^{N} n E_n \tr ( v_{\pm}^n ) v_{\mp}^{A_1} \ldots v_{\mp}^{A_{n-1}} T^B \tr ( T_{(B} T_{A_1} \ldots T_{A_{n-1} )} ) \, ,
\end{align}
which takes exactly the same form as (\ref{higher_spin_aux_eom_contracted}) with $j_{\pm}$ replaced by $\mathfrak{K}_{\pm}$. 

\subsection{AFSM with Wess-Zumino Term}\label{app:wz}

We now turn to the equations of motion for the auxiliary field sigma model along with a Wess-Zumino term, whose action is given in (\ref{family_with_WZ}) and reproduced here:
\begin{align}\label{afsm_wz_app}
    S_{\text{AFSM-WZ}} &= \hay \int d^2 \sigma \, \left( \frac{1}{2} \mathrm{tr} ( j_+ j_- ) + \mathrm{tr} ( v_+ v_- ) + \mathrm{tr} ( j_+ v_- + j_- v_+ ) + E ( \nu_2 , \ldots , \nu_N )  \right) \nonumber \\
    &\qquad + \frac{\kay}{6} \int_{\mathcal{M}_3} d^3 x \, \epsilon^{ijk} \tr \left( j_i [ j_j, j_k ] \right) \, .
\end{align}
Up to an overall multiplicative factor, the variation of the first line of (\ref{afsm_wz_app}) is identical to that of the ordinary higher-spin auxiliary field sigma model, which we have considered in appendix \ref{app:AFSM}. We therefore focus on the variation of the Wess-Zumino term, in the second line of (\ref{afsm_wz_app}), and its associated contribution to the equations of motion.

We have seen in appendix \ref{app:AFSM} that an infinitesimal variation $g \to g e^{\epsilon}$ for $\epsilon \in \mathfrak{g}$ leads to a variation of the Maurer-Cartan form
\begin{align}
    \delta j_i = [ j_i , \epsilon ] + \partial_i \epsilon \, ,
\end{align}
where now we use middle lowercase Latin letters like $i$, $j$, $k$ which label indices on $\mathcal{M}_3$. The variation of the Wess-Zumino term is therefore
\begin{align}
    \delta S_{\text{WZ}} &= \delta \left( \frac{1}{6} \int_{\mathcal{M}_3} d^3 x \, \epsilon^{ijk} \tr \left( j_i [ j_j, j_k ] \right) \right) \nonumber \\
    &= \frac{1}{2} \int_{\mathcal{M}_3} d^3 x \, \epsilon^{ijk} \tr \left( \left( \partial_i \epsilon + [ j_i , \epsilon ] \right) [ j_j , j_k ] \right) \, .
\end{align}
We use the product rule to rewrite 
\begin{align}
    \left( \partial_i \epsilon \right) [ j_j , j_k ] = \partial_i \left( \epsilon [ j_j , j_k ] \right) - \epsilon \partial_i [ j_j , j_k ]  \, ,
\end{align}
which yields
\begin{align}
    \delta S_{\text{WZ}} = \frac{1}{2} \int_{\mathcal{M}_3} d^3 x \, \epsilon^{ijk} \tr \left( \partial_i \left( \epsilon [ j_j , j_k ] \right) - \epsilon \partial_i [ j_j , j_k ] + [ j_i , \epsilon ] [ j_j, j_k ] \right) \, .
\end{align}
We evaluate the final term using cyclicity of the trace and the symmetries of the Levi-Civita tensor $\epsilon^{ijk}$, along with the ability to relabel dummy indices:
\begin{align}\label{levi_civita_commutator_product_identity}
    \epsilon^{ijk} \tr ( [ j_i , \epsilon ] [ j_j , j_k ] ) &= \epsilon^{ijk} \tr \left( \left( j_i \epsilon - \epsilon j_i \right) \left( j_j j_k - j_k j_j \right) \right) \nonumber \\
    &= \epsilon^{ijk} \tr \left( j_i \epsilon j_j j_k - j_i \epsilon j_k j_j - \epsilon j_i j_j j_k + \epsilon j_i j_k j_j \right) \nonumber \\
    &= \epsilon^{ijk} \tr \left( \epsilon \left( j_i [ j_j, j_k ] - [ j_j , j_k ] j_i ] \right) \right) \nonumber \\
    &= \epsilon^{ijk} \tr \left( \epsilon [ j_i , [ j_j , j_k ] ] \right) \, .
\end{align}
However,
\begin{align}
    \epsilon^{ijk} [ j_i , [ j_j , j_k ] ] = 0
\end{align}
by the Jacobi identity. Therefore this term vanishes, and we are left with
\begin{align}\label{delta_WZ_intermediate}
    \delta S_{\text{WZ}} &= \frac{1}{2} \int_{\mathcal{M}_3} d^3 x \, \epsilon^{ijk} \tr \left( \partial_i \left( \epsilon [ j_j , j_k ] \right) - \epsilon \partial_i [ j_j , j_k ] \right) \, \nonumber \\
    &= \frac{1}{2} \int_{\mathcal{M}_3} d^3 x \, \epsilon^{ijk} \tr \left( \partial_i \left( \epsilon [ j_j , j_k ] \right) + \epsilon \partial_i \left( \partial_j j_k - \partial_k j_j \right) \right) \, ,
\end{align}
where we have replaced the commutator using the Maurer-Cartan identity in the form
\begin{align}
    [ j_j , j_k ] = - \left( \partial_j j_k - \partial_k j_j  \right) \, .
\end{align}
But now the final term of (\ref{delta_WZ_intermediate}) vanishes by symmetry, since the second derivatives like $\partial_i \partial_j j_k$ are symmetric in the exchange $i \leftrightarrow j$ while the Levi-Civita symbol is antisymmetric:
\begin{align}
    \epsilon^{ijk}  \partial_i \partial_j j_k = \epsilon^{ijk} \partial_i \partial_k j_j = 0 \, .
\end{align}
Therefore we conclude
\begin{align}\label{final_delta_SWZ}
    \delta S_{\text{WZ}} = \frac{1}{2} \int_{\mathcal{M}_3} d^3 x \, \epsilon^{ijk} \partial_i \tr \left( \epsilon [ j_j , j_k ] \right) \, ,
\end{align}
which is the integral of a total derivative and thus localizes to the boundary $\partial \mathcal{M}_3 = \Sigma$.

Explicitly, using Stokes' theorem and writing the result in light-cone coordinates, the $3d$ integral (\ref{final_delta_SWZ}) is equivalent to the $2d$ integral
\begin{align}\label{delta_SWZ_localized}
    \delta S_{\text{WZ}} &= - \frac{1}{2} \int d^2 \sigma \, \tr ( \epsilon [ j_+ , j_- ] ) \nonumber \\
    &= + \frac{1}{2} \int d^2 \sigma \, \tr \left( \epsilon \left( \partial_+ j_- - \partial_- j_+ \right) \right) \, ,
\end{align}
where we have used the Maurer-Cartan identity in going to the second line.

Let us pause to make one remark. We commented, in the paragraph following equation (\ref{PCM_WZ}), that well-definedness of the PCM with Wess-Zumino term requires that the variation of $S_{\text{WZ}}$ must vanish for any fluctuation $\epsilon$ with the property that $\epsilon \Big\vert_{\Sigma} = 0$. Indeed, this is now manifest from the form (\ref{delta_SWZ_localized}) of $\delta S_{\text{WZ}}$, which is written as an integral over $\Sigma$. Therefore, if the parameter $\epsilon$ of the $g$-variation has support only away from the boundary of $\mathcal{M}_3$, the Wess-Zumino term $S_{\text{WZ}}$ is invariant under such a fluctuation, as required.

Let us conclude by recording the full equation of motion arising from the combined action (\ref{afsm_wz_app}). To do this, we add up the contribution from variation of the AFSM term, which takes the form (\ref{afsm_var_intermediate_two}) multiplied by a factor of $\hay$, and the variation of the Wess-Zumino term, which is given by (\ref{delta_SWZ_localized}) multiplied by the constant $\kay$. The result is
\begin{align}
    \delta S = \int &d^2 \sigma \, \tr \Bigg( \epsilon \Bigg( \hay \left( [ v_- , j_+ ] + [ v_+, j_- ] - \frac{1}{2} \partial_+ j_- - \frac{1}{2} \partial_- j_+  - \partial_+ v_- - \partial_- v_+  \right) \nonumber \\
    &\qquad \qquad \qquad + \frac{\kay}{2} \left( \partial_+ j_- - \partial_- j_+ \right) \Bigg) \Bigg) \, ,
\end{align}
and demanding that this vanish for any $\epsilon$ gives
\begin{align}
    \left( \hay - \kay \right) \partial_+ j_- + \left( \hay + \kay \right) \partial_- j_+ + 2 \hay \left( \partial_+ v_- + \partial_- v_+ + [ j_+ , v_- ] + [ j_- ,v_+ ] \right) = 0 \, .
\end{align}
This gives the $g$-field equation of motion for the AFSM with WZ term. The auxiliary field equation of motion is unaffected by the Wess-Zumino coupling, and thus still takes the same form (\ref{higher_spin_aux_eom_contracted}) as in the AFSM.

\section{Generalized Jacobi Identity}\label{app:jac}

The goal of this appendix is to state some simple results concerning traces of products of generators of a Lie algebra. Although these facts are well-known and elementary -- indeed, they simply follows from ad-invariance of the trace, or equivalently, the fact that invariant tensors are inert under changes of coordinates -- they turn out to be important in developing some consequences of the equations of motion for auxiliary field sigma models, as we see in section \ref{sec:implications}. Following the terminology of \cite{vanRitbergen:1998pn}, we refer to the first of these statements as the ``generalized Jacobi identity'' for Lie algebras.
\begin{lemma}[Generalized Jacobi Identity]\label{lem:gen_jacobi}
    Fix a representation $R$ of a semi-simple Lie algebra $\mathfrak{g}$ associated to a Lie group $G$. Choose a basis $T_A$ of generators for $\mathfrak{g}$ in the representation $R$, let $\Tr$ denote the trace in this representation, and let $\tensor{f}{_A_B^C}$ be the structure constants which satisfy $[ T_A, T_B ] = \tensor{f}{_A_B^C} T_C$. Then for any trace of $n$ generators, one has
    \begin{align}\label{generalized_jacobi}
        \sum_{i=1}^{n} \tensor{f}{_C_{A_i}^B} \Tr \left( T_{A_1} \ldots T_{A_{i-1}} T_B T_{A_{i+1}} \ldots T_{A_n} \right) = 0 \, ,
    \end{align}
    where in each term of the sum, the index $A_i$ of the $i$-th generator within the trace has been replaced by the index $B$.
\end{lemma}

\begin{proof}
    First recall that, if $g \in G$ is an arbitrary element of the Lie group $G$, then any trace of a product of generators $T_{A_i} \in \mathfrak{g}$ is invariant under the simultaneous replacements
    \begin{align}\label{adjoint_action_G}
        T_{A_i} \to \mathrm{Ad}_g \left( T_{A_i} \right) \, .
    \end{align}
    In the case of a matrix Lie group, we may view this as a consequence of the cyclicity of the trace; under the adjoint action (\ref{adjoint_action_G}), each generator $T_{A_i}$ transforms as
    \begin{align}
        T_{A_i} \to g T_{A_i} g^{-1} \, , 
    \end{align}
    and thus
    \begin{align}\label{invariance_trace_generators}
        \Tr \left( T_{A_1} T_{A_2} \ldots T_{A_n} \right) &\to \Tr \left( g T_{A_1} g^{-1} g T_{A_2} g^{-1} \ldots g T_{A_n} g^{-1} \right) \nonumber \\
        &= \Tr \left( T_{A_1} T_{A_2} \ldots T_{A_n} \right) \, .
    \end{align}
    In these expressions, products like $g T_{A} g^{-1}$ are simply products of matrices; such an expression is interpreted as enacting a change of coordinates on the matrix $T_{A}$, parameterized by the matrix $g$. For an abstract Lie group where such a matrix product is not available, the expression $g T_{A} g^{-1}$ is not (strictly speaking) a product of any three quantities. However, for simplicity we will still use the notation $g T_{A} g^{-1}$ for $\mathrm{Ad}_g \left( T_{A} \right)$, and we note that even in such an abstract Lie group, the invariance of any trace under (\ref{adjoint_action_G}) still holds.

    Let us now specialize to the case where $g \in G$ is infinitesimal and can be parameterized as $g = e^{\epsilon} = 1 + \epsilon$ for $\epsilon \in \mathfrak{g}$, where we work to leading order near the identity, so that $g^{-1} = e^{-\epsilon} = 1 - \epsilon$. Using product notation, the invariance condition (\ref{invariance_trace_generators}) then reads
    \begin{align}\label{epsilons_ad_invariance}
        \Tr \left( T_{A_1} T_{A_2} \ldots T_{A_n} \right) = \Tr \Big( ( 1 + \epsilon ) T_{A_1} ( 1 - \epsilon ) ( 1 + \epsilon ) T_{A_2} ( 1 - \epsilon ) \ldots ( 1 + \epsilon ) T_{A_n} ( 1 - \epsilon ) \Big) \, .
    \end{align}
    We then collect all terms of order $\epsilon$ that arise from expanding the various factors in (\ref{epsilons_ad_invariance}). For each generator $T_{A_i}$, there will be one pair of terms involving
    \begin{align}
        \epsilon T_{A_i} - T_{A_i} \epsilon = [ \epsilon, T_{A_i} ] = \epsilon^C [ T_C, T_{A_i} ] = \epsilon^C \tensor{f}{_C_{A_i}^B} T_B \, ,
    \end{align}
    which appears inside of the trace, along with the other generators $T_{A_j}$ for $j \neq i$. One could have also inferred this result using the abstract notation, without appealing to the product structure of a matrix Lie group, since
    \begin{align}
        \mathrm{ad}_\epsilon T_{A_i} = [ \epsilon, T_{A_i} ] \, ,
    \end{align}
    where we write $\mathrm{ad}_X$ for the adjoint action of $X \in \mathfrak{g}$ to distingusih from $\mathrm{Ad}_g$ for $g \in G$.
    
    Demanding that the resulting $\mathcal{O} ( \epsilon )$ terms vanish for any $\epsilon$, we conclude that
    \begin{align}\label{generalized_jacobi_proved}
        \sum_{i=1}^{n} \tensor{f}{_C_{A_i}^B} \Tr \left( T_{A_1} \ldots T_{A_{i-1}} T_B T_{A_{i+1}} \ldots T_{A_n} \right) = 0 \, ,
    \end{align}
    which is what we set out to prove.    
\end{proof}
The result (\ref{generalized_jacobi}) holds generically for any trace of a product of generators in any semi-simple Lie algebra. However, one can say more by restricting attention to a \emph{symmetrized} trace of generators, which is the case that we encounter in section \ref{sec:implications}. We state the relevant consequence of Lemma \ref{lem:gen_jacobi} as a corollary.

\begin{corollary}
    Let
    \begin{align}
        M_{A_1 \ldots A_n} = \tr \left( T_{(A_1} \ldots T_{A_n)} \right) \, .
    \end{align}
    be a totally symmetrized trace of generators $T_{A_i}$ for a semi-simple Lie algebra $\mathfrak{g}$ in a representation $R$, and let $\Tr$ and $\tensor{f}{_A_B^C}$ be the trace and structure constants as before. Then
    \begin{align}
        0 = \tensor{f}{_C_{( A_1}^B} M_{A_{2} \ldots A_{n} ) B} \, .
    \end{align}
\end{corollary}

\begin{proof}
    By Lemma \ref{lem:gen_jacobi}, we have
    \begin{align}
        \sum_{i=1}^{n} \tensor{f}{_C_{A_i}^B} \Tr \left( T_{A_1} \ldots T_{A_{i-1}} T_B T_{A_{i+1}} \ldots T_{A_n} \right) = 0 \, ,
    \end{align}
    and symmetrizing this equation over the indices $A_1$, $\ldots$, $A_n$ gives
    \begin{align}
        \sum_{i=1}^{n} \tensor{f}{_C_{( A_i}^B} M_{A_1 \ldots A_{i-1} | B |  A_{i+1} \ldots A_n ) } = 0 \, .
    \end{align}
    We may now use the complete symmetry of the tensor $M$ to move the index $B$ to the final slot in every term of the sum, finding
    \begin{align}
        \sum_{i=1}^{n} \tensor{f}{_C_{( A_i}^B} M_{A_1 \ldots A_{i-1} A_{i+1} \ldots A_n ) B } = 0 \, ,
    \end{align}
    and subsequently use the symmetrization over the indices $A_i$ to arrange these indices in ascending order in every term,
    \begin{align}
        \sum_{i=1}^{n} \tensor{f}{_C_{( A_1}^B} M_{A_2 \ldots A_n ) B } = 0 \, ,
    \end{align}
    which is the desired result.
    
\end{proof}

\section{Details of Poisson Bracket Calculations}\label{app:pb}

In this appendix, we present the intermediate steps of computations whose results are used in the proof of strong integrability for higher-spin auxiliary field sigma models in section \ref{sec:maillet}. Although these manipulations are standard and well-known in the literature, we include a considerable level of detail for completeness and to make the present work self-contained.

\subsection{Proofs of Basic Results}\label{sec:app_basic}

Before beginning the main calculations of the required Poisson brackets, let us collect a few simple observations which will be useful later.

\subsubsection*{\ul{\it Variant of Maurer-Cartan identity}}

It will be helpful to record an identity which is a consequence of the Maurer-Cartan identity (\ref{mc_identity}). We begin with the usual identity in the form
\begin{align}
    \partial_\alpha j_\beta - \partial_\beta j_\alpha + [ j_\alpha, j_\beta ] = 0 \, .
\end{align}
By pushing forward all of the indices on the two-dimensional spacetime $\Sigma$ to those on the target space $G$, this identity can be written as
\begin{align}
    \partial_\mu j_\nu - \partial_\nu j_\mu + [ j_\mu , j_\nu ] = 0 \, ,
\end{align}
where $j_\mu \partial_\alpha \phi^\mu = j_\alpha$. Expanding as
\begin{align}
    j_\mu = j_\mu^A T_A
\end{align}
then gives
\begin{align}
    \left( \partial_\mu j_\nu^A \right) T_A - \left( \partial_\nu j_\mu^A \right) T_A + j_\mu^A j_\nu^B [ T_A, T_B ] = 0 \, .
\end{align}
Evaluating the commutator in terms of structure constants,
\begin{align}
    [ T_A, T_B ] = \tensor{f}{_A_B^C} T_C \, ,
\end{align}
we find
\begin{align}
    \left( \partial_\mu j_\nu^A \right) T_A - \left( \partial_\nu j_\mu^A \right) T_A + j_\mu^A j_\nu^B \tensor{f}{_A_B^C} T_C = 0 \, ,
\end{align}
or after collecting terms and re-indexing,
\begin{align}
    \left( \partial_\mu j_\nu^A - \partial_\nu j_\mu^A + j_\mu^B j_\nu^C \tensor{f}{_B_C^A} \right) T_A = 0 \, .
\end{align}
This equation must hold independently for the component along each generator $T_A$, so
\begin{align}\label{mc_push_forward}
    \partial_\mu j_\nu^A - \partial_\nu j_\mu^A  = - j_\mu^B j_\nu^C \tensor{f}{_B_C^A} \, ,
\end{align}
for each $A$.

\subsubsection*{\ul{\it Delta function identities}}

Poisson brackets for continuum field theories necessarily involve Dirac delta functions, as in the fundamental Poisson bracket (\ref{fundamental_poisson}) for the auxiliary field sigma model. Therefore it will be expedient to briefly record some standard identities concerning $\delta$ functions.

Functional derivatives are defined such that
\begin{align}
    \frac{\delta f ( x ) }{\delta f ( y ) } = \delta  ( x - y ) \, .
\end{align}
One also has the functional derivative identity
\begin{align}\label{deriv_of_delta_identity}
    \frac{\delta ( \partial_x f ( x ) ) }{\delta f ( y ) } = - \partial_y \delta ( y - x ) \, ,
\end{align}
which can be proved by integration by parts,
\begin{align}
    \frac{\delta ( \partial_x f ( x ) ) }{\delta f ( y ) } &= \frac{\delta}{\delta f ( y ) } \int dz \, \left( \partial_z f ( z ) \right) \delta ( z - x ) \nonumber \\
    &= - \frac{\delta}{\delta f ( y ) } \int dz \, f ( z ) \partial_z  \delta ( z - x ) \nonumber \\
    &= - \int dz \, \delta ( z - y ) \partial_z  \delta ( z - x ) \nonumber \\
    &= - \partial_y \delta ( y - x ) \, .
\end{align}
The composition of the delta function with another function $F(x)$ behaves as
\begin{align}
    \delta ( F ( x ) ) = \sum_i \frac{\delta ( x - x_i )}{| F' ( x_i ) | } \, ,
\end{align}
where the sum runs over all simple roots of $F(x)$, which implies that $\delta ( x - y ) = \delta ( y - x )$, and by exploiting a change of variables, that
\begin{align}\label{delta_flip_sign}
    \partial_x \delta ( x - y ) = - \partial_y \delta ( y - x ) \, .
\end{align}
Finally, we will need the identity
\begin{align}\label{nice_delta_identity}
    f ( y ) \partial_x \delta ( x - y ) &= f ( x ) \partial_x \delta ( x - y ) + f' ( x ) \delta ( x - y ) \, .
\end{align}
To prove (\ref{nice_delta_identity}), we assume that the expression on the left side sits under an integral, and then Taylor expand the function $f ( y )$ to leading order near $y = x$ as
\begin{align}\label{derive_nice}
    \int dy \, f ( y ) \partial_x \delta ( x - y ) &= \int dy \,  \left( f ( x ) + f'(x) \cdot ( y - x ) \right) \partial_x \delta ( x - y ) \nonumber \\
    &= \int dy \,  \left( f ( x ) \partial_x \delta ( x - y ) - f'
     ( x ) \cdot ( y - x ) \cdot \partial_y \delta ( y - x ) \right) \nonumber \\
    &=  \int dy \,  \left( f ( x ) \partial_x \delta ( x - y ) + \delta ( y - x ) \partial_y \left( f'
     ( x ) \cdot ( y - x ) \right) \right) \nonumber \\
     &= \int dy \,  \left( f ( x ) \partial_x \delta ( x - y ) +  f' ( x )  \delta ( x - y )  \right) \, .
\end{align}
In the second line we have used (\ref{delta_flip_sign}) to exchange the $x$-derivative for a $y$-derivative, in the third step we have integrated by parts with respect to $y$, and in the final step we use $\delta ( y - x ) = \delta ( x - y )$. Higher-order terms in the Taylor expansion of $f(y)$ near $y = x$ vanish, since they are proportional to $(y - x)$ and multiply a delta function $\delta ( y - x )$ after integration by parts. Comparing the integrands of the first and last lines of (\ref{derive_nice}) gives the distributional identity (\ref{nice_delta_identity}).

\subsubsection*{\ul{\it Properties of tensor product quantities}}

In equation (\ref{doubled_X}), we have introduced the ``doubled'' quantities $X_1$ and $X_2$, which map any element $X \in \mathfrak{g}$ into $\mathfrak{g} \otimes \mathfrak{g}$ by right- or left-tensor-multiplying with the identity matrix.\footnote{Of course, the identity matrix $1$ is not an element of the Lie algebra $\mathfrak{g}$, but we are implicitly conflating $\mathfrak{g}$ with its universal enveloping algebra $U ( \mathfrak{g} )$. Properly speaking, $X_1$ and $X_2$ belong to $U ( \mathfrak{g} ) \otimes U ( \mathfrak{g} )$.}

These doubled quantities interact in a straightforward way with the Poisson bracket. If $X = X^A ( \phi , \pi ) T_A$, $Y = Y^A ( \phi , \pi ) T_A$ are any Lie algebra valued quantities which depend on the phase space variables $\phi^\mu$ and $\pi_\mu$, then the Poisson bracket of $X_1$ and $Y_2$ satisfies
\begin{align}\label{poisson_brackets_doubled_identity}
    &\left\{ X_1 ( \sigma ) , Y_2 ( \sigma' ) \right\} \nonumber \\
    &\qquad = \left\{ X ( \sigma ) \otimes 1 , 1 \otimes Y ( \sigma' ) \right\} \nonumber \\
    &\qquad = \int d \sigma'' \Bigg( \left( \frac{\delta X ( \sigma )}{\delta \phi^\rho ( \sigma'' ) } \otimes 1 \right) \cdot \left( 1 \otimes \frac{\delta Y ( \sigma' ) }{\delta \pi_\rho ( \sigma'' ) } \right)  - \left( \frac{\delta X ( \sigma )}{\delta \pi_\rho ( \sigma'' ) } \otimes 1 \right) \cdot \left( 1 \otimes \frac{\delta Y ( \sigma' ) }{\delta \phi^\rho ( \sigma'' ) } \right) \Bigg) \nonumber \\
    &\qquad = \int d \sigma'' \Bigg( \frac{\delta X^A ( \sigma )}{\delta \phi^\rho ( \sigma'' ) } \frac{\delta Y^B ( \sigma' ) }{\delta \pi_\rho ( \sigma'' ) }  - \frac{\delta X^A ( \sigma )}{\delta \pi_\rho ( \sigma'' ) } \frac{\delta Y^B ( \sigma' ) }{\delta \phi^\rho ( \sigma'' ) } \Bigg) T_A \otimes T_B \nonumber \\
    &\qquad = \left\{ X^A, Y^B \right\} T_A \otimes T_B \, .
\end{align}
Also recall from equation (\ref{casimir_defn}) that we have defined the Casimir element
\begin{align}
    C_{12} = \gamma^{AB} T_A \otimes T_B \in \mathfrak{g} \otimes \mathfrak{g} \, .
\end{align}
Note that, by the bilinearity of the tensor product, for any $X_1$ one has
\begin{align}\label{C12_X1}
    [ C_{12} , X_1 ] &= \gamma^{AB} [ T_A \otimes T_B , X^C T_C \otimes 1 ] \nonumber \\
    &= \gamma^{AB} X^C \left( \left( [ T_A, T_C ] \right) \otimes T_B \right) \nonumber \\
    &= \gamma^{AB} X^C \tensor{f}{_A_C^D} \left( T_D \otimes T_B \right) \, ,
\end{align}
while on the other hand a similar manipulation gives
\begin{align}\label{C12_X2}
    [ C_{12}, X_2 ] &= \gamma^{AB} [ T_A \otimes T_B , 1 \otimes  X^C T_C ] \nonumber \\
    &= \gamma^{AB} X^C \left( T_A \otimes \left( [T_B, T_C ] \right) \right) \nonumber \\
    &= \gamma^{AB} X^C \tensor{f}{_B_C^D} \left( T_A \otimes T_D \right) \, .
\end{align}
Comparing the last line of (\ref{C12_X1}) to the last line of (\ref{C12_X2}), using the symmetry of $\gamma^{AB}$ and anti-symmetry of $\tensor{f}{_A_B^C}$, and relabeling indices, we conclude
\begin{align}\label{C12_X1_to_C12_X2}
    [ C_{12} , X_1 ] = - [ C_{12}, X_2 ] \, .
\end{align}
There is one more simple observation which will be useful. Consider general quantities $X^A ( \sigma ) $ and $Y^B ( \sigma' )$ which have the property that their Poisson bracket is antisymmetric under the simultaneous interchange of $\sigma$ with $\sigma'$ and of $A$ with $B$, that is
\begin{align}\label{symmetry_assumption}
    \left\{ X^A ( \sigma ) , Y^B ( \sigma' ) \right\} = - \left\{ X^B ( \sigma' ) , Y^A ( \sigma ) \right\} \, .
\end{align}
Then using the above identity (\ref{poisson_brackets_doubled_identity}), along with antisymmetry of the Poisson bracket and the symmetry assumption (\ref{symmetry_assumption}), we see that
\begin{align}\label{steps_to_swap_bracket}
    \left\{ X_1 ( \sigma ) , Y_2 ( \sigma' ) \right\} &= \left\{ X^A ( \sigma ) , Y^B ( \sigma' ) \right\} T_A \otimes T_B \nonumber \\
    &= - \left\{ X^B ( \sigma' ) , Y^A ( \sigma ) \right\} T_A \otimes T_B \nonumber \\
    &= + \left\{ Y^A ( \sigma ) , X^B ( \sigma' )  \right\} T_A \otimes T_B \nonumber \\
    &= \left\{ Y_1 ( \sigma ) , X_2 ( \sigma' ) \right\} \, .
\end{align}

\subsection{Computation of Brackets Involving $j_\alpha^A$ and $\fJ_\alpha^A$}\label{app:uncontracted_brackets}

We now begin the calculations of each of the canonical (Poisson or Dirac) brackets which appear in the arguments of section \ref{sec:maillet}.

\subsubsection*{Calculation of $\{ \fJ_\tau^A, \fJ_\tau^B \}$ }

Let us first compute the Poisson bracket $\{ \fJ_\tau^A ( \sigma ) , \fJ_\tau^B ( \sigma' ) \}$. Using (\ref{J_to_pi}) this is
\begin{align}
    \{ \fJ_\tau^A ( \sigma ) , \fJ_\tau^B ( \sigma' ) \} = \eta^{AC} \eta^{BD} \left\{ j^\mu_C ( \phi ( \sigma ) ) \pi_\mu ( \sigma ) , j^\nu_D ( \phi ( \sigma' ) )  \pi_\nu ( \sigma' ) \right\} \, .
\end{align}
Using the definition (\ref{bracket_defn}) of the Poisson bracket, we then have
\begin{align}\label{J_J_bracket_step_one}
    \{ \fJ_\tau^A ( \sigma ) , \fJ_\tau^B ( \sigma' ) \} &= \eta^{AC} \eta^{BD} \int d \sigma'' \Bigg( \frac{\delta \left( j^\mu_C ( \phi (  \sigma ) ) \pi_\mu ( \sigma ) \right) }{\delta \phi^\rho ( \sigma'' )} \frac{\delta \left( j^\nu_D ( \phi ( \sigma' ) )  \pi_\nu ( \sigma' ) \right) }{\delta \pi_\rho ( \sigma'' ) } \nonumber \\
    &\hspace{100pt} - \frac{ \delta \left( j^\nu_D ( \phi (  \sigma' ) )   \pi_\nu ( \sigma' ) \right) }{\delta \phi^\rho ( \sigma'' )} \frac{\delta \left( j^\mu_C ( \phi ( \sigma ) ) \pi_\mu ( \sigma ) \right) }{\delta \pi_\rho ( \sigma'' ) } \Bigg) \nonumber \\
    &= \eta^{AC} \eta^{BD} \int d \sigma'' \Bigg( \partial_\rho j^\mu_C ( \phi ( \sigma ) ) \pi_\mu ( \sigma ) j^\rho_D ( \phi ( \sigma' ) ) \delta ( \sigma' - \sigma'' ) \delta ( \sigma - \sigma'' ) \nonumber \\
    &\hspace{100pt} - \partial_\rho j^\nu_D ( \phi ( \sigma' ) ) \pi_\nu ( \sigma' ) j^\rho_C ( \phi ( \sigma ) ) \delta ( \sigma - \sigma'' ) \delta ( \sigma' - \sigma'' ) \Bigg) \nonumber \\
    &= \eta^{AC} \eta^{BD} \delta ( \sigma - \sigma' ) \pi_\mu ( \sigma ) \Big( \partial_\rho j^\mu_C ( \phi ( \sigma ) ) j_D^\rho ( \phi ( \sigma' ) ) - \partial_\rho j^\mu_D ( \phi ( \sigma' ) ) j^\rho_C ( \phi ( \sigma ) ) \Big) \, .
\end{align}
Because all quantities in the last line of (\ref{J_J_bracket_step_one}) appear multiplying a delta function $\delta ( \sigma - \sigma' )$, from here onward we will not distinguish between quantities evaluated at $\sigma$ and those evaluated at $\sigma'$. Next we will insert a resolution of the identity in the form
\begin{align}
    \tensor{\delta}{^\nu_\mu} = j_E^\nu j^E_\mu \, ,
\end{align}
which gives
\begin{align}\label{fJ_fJ_intermediate}
    \{ \fJ_\tau^A ( \sigma ) , \fJ_\tau^B ( \sigma' ) \} &= \eta^{AC} \eta^{BD} \delta ( \sigma - \sigma' ) \pi_\nu  j_E^\nu j^E_\mu \left( \partial_\rho j^\mu_C  j_D^\rho  - \partial_\rho j^\mu_D j^\rho_C \right) \nonumber \\
    &= \eta^{AC} \eta^{BD} \delta ( \sigma - \sigma' ) \pi_\nu \left( j_E^\nu j^E_\mu \partial_\rho j^\mu_C  j_D^\rho - j_E^\nu j^E_\mu \partial_\rho j^\mu_D j^\rho_C \right) \, .
\end{align}
Note that we have
\begin{align}
    0 = \partial_\rho \tensor{\delta}{^E_D} = \partial_\rho \left( j^E_\mu j^\mu_D \right) \, ,
\end{align}
and thus
\begin{align}\label{switch_j_identity}
    j^E_\mu \partial_\rho j^\mu_D  = - j^\mu_D \partial_\rho j^E_\mu \, .
\end{align}
Applying the identity (\ref{switch_j_identity}) twice in (\ref{fJ_fJ_intermediate}), we find
\begin{align}
    \{ \fJ_\tau^A ( \sigma ) , \fJ_\tau^B ( \sigma' ) \} &= - \eta^{AC} \eta^{BD} \delta ( \sigma - \sigma' ) \pi_\nu \left( j_E^\nu j^\mu_C \partial_\rho j^E_\mu  j_D^\rho - j_E^\nu j^\mu_D \partial_\rho j^E_\mu j^\rho_C \right) \nonumber \\
    &= - \eta^{AC} \eta^{BD} \delta ( \sigma - \sigma' ) \pi_\nu \left( \partial_\rho j^E_\mu -   \partial_\mu j^E_\rho \right) j_E^\nu j^\mu_C  j_D^\rho \, ,
\end{align}
where in the last step we have relabeled indices. We now apply equation (\ref{mc_push_forward}) to obtain
\begin{align}
    \{ \fJ_\tau^A ( \sigma ) , \fJ_\tau^B ( \sigma' ) \} &= \eta^{AC} \eta^{BD} \delta ( \sigma - \sigma' ) \pi_\nu \left( j_\rho^F j_\mu^G \tensor{f}{_F_G^E} \right) j_E^\nu j^\mu_C  j_D^\rho \, \nonumber \\
    &=  \eta^{AC} \eta^{BD} \delta ( \sigma - \sigma' ) \pi_\nu  \tensor{f}{_F_G^E} j_E^\nu \tensor{\delta}{_C^G} \tensor{\delta}{_D^F} \nonumber \\
    &= \eta^{AC} \eta^{BD} \delta ( \sigma - \sigma' ) \pi_\nu \tensor{f}{_D_C^E} j^\nu_E \, .
\end{align}
Then using the formula (\ref{J_to_pi}), we find
\begin{align}
    \{ \fJ_\tau^A ( \sigma ) , \fJ_\tau^B ( \sigma' ) \} &= - \eta^{AC} \eta^{BD} \delta ( \sigma - \sigma' ) \fJ_\tau^F \tensor{f}{_D_C_F} \nonumber \\
    &= - \delta ( \sigma - \sigma' ) \tensor{f}{^B^A_F} \fJ_\tau^F \, ,
\end{align}
which finally using the antisymmetry of the structure constants yields
\begin{align}\label{final_Jt_Jt_bracket}
    \{ \fJ_\tau^A ( \sigma ) , \fJ_\tau^B ( \sigma' ) \} = \tensor{f}{^A^B_C} \fJ_\tau^C \delta ( \sigma - \sigma' ) \, .
\end{align}
Equation (\ref{final_Jt_Jt_bracket}) is our final result for the first of our useful brackets.

\subsubsection*{Calculation of $\{ \fJ_\tau^A, j_\sigma^B \}$}

Next we would like to compute
\begin{align}
    \left\{ \fJ_\tau^A ( \sigma ) , j_\sigma^B ( \sigma' ) \right\} = - \gamma^{AC} \left\{ j^\mu_C ( \phi ( \sigma ) ) \pi_\mu ( \sigma ) , j_\nu^B ( \phi ( \sigma' ) ) \frac{\partial \phi^\nu ( \sigma' )}{\partial \sigma'} \right\} \, .
\end{align}
Using the definition of the Poisson bracket in equation (\ref{bracket_defn}), we have
\begin{align}
    \left\{ \fJ_\tau^A ( \sigma ) , j_\sigma^B ( \sigma' ) \right\} &= - \gamma^{AC} \int d \sigma'' \Bigg( \frac{\delta \left( j^\mu_C ( \phi ( \sigma ) ) \pi_\mu ( \sigma ) \right) }{\delta \phi^\rho ( \sigma'' )} \frac{\delta \left( j_\nu^B ( \phi ( \sigma' ) ) \frac{\partial \phi^\nu ( \sigma' )}{\partial \sigma'} \right) }{\delta \pi_\rho ( \sigma'' ) } \nonumber \\
    &\hspace{100pt} - \frac{ \delta \left( j_\nu^B ( \phi ( \sigma' ) ) \frac{\partial \phi^\nu ( \sigma' )}{\partial \sigma'} \right) }{\delta \phi^\rho ( \sigma'' )} \frac{\delta \left( j^\mu_C ( \phi ( \sigma ) ) \pi_\mu ( \sigma ) \right) }{\delta \pi_\rho ( \sigma'' ) } \Bigg) \, .
\end{align}
The first term vanishes, and in the second we use the identity (\ref{deriv_of_delta_identity}) to write
\begin{align}\label{JT_js_bracket_intermediate}
    \left\{ \fJ_\tau^A ( \sigma ) , j_\sigma^B ( \sigma' ) \right\} &= \gamma^{AC} \int d \sigma'' \Bigg( \frac{ \delta \left( j_\nu^B ( \phi ( \sigma' ) ) \frac{\partial \phi^\nu ( \sigma' )}{\partial \sigma'} \right) }{\delta \phi^\rho ( \sigma'' )} \frac{\delta \left( j^\mu_C ( \phi ( \sigma ) ) \pi_\mu ( \sigma ) \right) }{\delta \pi_\rho ( \sigma'' ) } \Bigg) \, \nonumber \\
    &= \gamma^{AC} \int d \sigma'' \, j^\rho_C ( \phi ( \sigma ) ) \delta ( \sigma - \sigma'' ) \Bigg( \partial_\rho j_\nu^B ( \phi ( \sigma' ) ) \frac{\partial \phi^\nu ( \sigma' )}{\partial \sigma'} \delta ( \sigma' - \sigma'' ) \nonumber \\
    &\hspace{200pt} - j_\rho^B ( \phi ( \sigma' ) ) \partial^{(\sigma'')} \delta ( \sigma'' - \sigma' ) \Bigg) \nonumber \\
    &= \gamma^{AC} j_C^\rho ( \phi ( \sigma ) ) \partial_\rho j_\nu^B ( \phi ( \sigma' ) ) \frac{\partial \phi^\nu ( \sigma' )}{\partial \sigma'} \delta ( \sigma' - \sigma ) \nonumber \\
    &\qquad - \gamma^{AC} j_C^\rho ( \phi ( \sigma ) ) j_\rho^B ( \phi ( \sigma' ) ) \partial_\sigma \delta ( \sigma - \sigma' ) \, . 
\end{align}
Furthermore, using the delta function identity (\ref{nice_delta_identity}), we can express
\begin{align}
    j_\rho^B ( \phi ( \sigma' ) ) \partial_\sigma \delta ( \sigma - \sigma' ) = j_\rho^B ( \phi ( \sigma ) ) \partial_\sigma \delta ( \sigma - \sigma' ) + \partial_\sigma j_\rho^B ( \phi ( \sigma ) ) \delta ( \sigma - \sigma' ) \, .
\end{align}
Using this result in (\ref{JT_js_bracket_intermediate}) gives
\begin{align}\label{JT_js_bracket_intermediate_two}
    \left\{ \fJ_\tau^A ( \sigma ) , j_\sigma^B ( \sigma' ) \right\} &= \gamma^{AC} j_C^\rho ( \phi ( \sigma ) ) \partial_\rho j_\nu^B ( \phi ( \sigma' ) ) \frac{\partial \phi^\nu ( \sigma' )}{\partial \sigma'} \delta ( \sigma' - \sigma ) \nonumber \\
    &\qquad - \gamma^{AC} j_C^\rho ( \phi ( \sigma ) ) \left( j_\rho^B ( \phi ( \sigma ) ) \partial_\sigma \delta ( \sigma - \sigma' ) + \partial_\sigma j_\rho^B ( \phi ( \sigma ) ) \delta ( \sigma - \sigma' ) \right) \nonumber \\
    &= \gamma^{AC} j_C^\rho \partial_\rho j_\nu^B \partial_\sigma \phi^\nu \delta ( \sigma' - \sigma ) - \gamma^{AC} j_C^\rho ( \phi ( \sigma ) ) j_\rho^B ( \phi ( \sigma ) ) \partial_\sigma \delta ( \sigma - \sigma ' ) \nonumber \\
    &\qquad - \gamma^{AC} j_C^\rho \partial_\sigma j_\rho^B \delta ( \sigma - \sigma' ) \, ,
\end{align}
where we no longer distinguish between fields evaluated at $\sigma$ and those evaluated at $\sigma'$ in terms that multiply $\delta ( \sigma' - \sigma )$. We use $j^B_\rho ( \phi ( \sigma ) ) j^\rho_C ( \phi ( \sigma ) ) = \tensor{\delta}{^B_C}$ and collect terms to write
\begin{align}
    \left\{ \fJ_\tau^A ( \sigma ) , j_\sigma^B ( \sigma' ) \right\} = \gamma^{AC} j_C^\rho \partial_\sigma \phi^\nu \delta ( \sigma - \sigma' ) \left( \partial_\rho j_\nu^B - \partial_\nu j_\rho^B \right) - \gamma^{AC} \tensor{\delta}{_C^B} \partial_\sigma \delta ( \sigma - \sigma' ) \, .
\end{align}
Again applying the Maurer-Cartan identity in the form (\ref{mc_push_forward}), we find
\begin{align}
    \left\{ \fJ_\tau^A ( \sigma ) , j_\sigma^B ( \sigma' ) \right\} = \gamma^{AC} j_C^\rho \partial_\sigma \phi^\nu \delta ( \sigma - \sigma' ) \left( - j_\rho^D j_\nu^E \tensor{f}{_D_E^B} \right) - \gamma^{AB} \partial_\sigma \delta ( \sigma - \sigma' ) \, ,
\end{align}
and again using $j^\rho_C j^D_\rho = \tensor{\delta}{_C^D}$ gives
\begin{align}
    \left\{ \fJ_\tau^A ( \sigma ) , j_\sigma^B ( \sigma' ) \right\} = - \gamma^{AC} \partial_\sigma \phi^\nu \delta ( \sigma - \sigma' ) j_\nu^E \tensor{f}{_C_E^B}  - \gamma^{AB} \partial_\sigma \delta ( \sigma - \sigma' ) \, .
\end{align}
Simplifying using antisymmetry of the structure constants and $\partial_\sigma \phi^\nu j_\nu^E = j_\sigma^E$, we conclude
\begin{align}\label{JT_js_final}
    \left\{ \fJ_\tau^A ( \sigma ) , j_\sigma^B ( \sigma' ) \right\} = \tensor{f}{^A^B_C} j_\sigma^C \delta ( \sigma - \sigma' )   - \gamma^{AB} \partial_\sigma \delta ( \sigma - \sigma' ) \, ,
\end{align}
which is the final form of the brackets we wished to derive in this subsection.

Let us also remark on the behavior of (\ref{JT_js_final}) under the simultaneous replacements
\begin{align}\label{simult_replacements}
    A \longleftrightarrow B \, , \quad \sigma \longleftrightarrow \sigma' \, .
\end{align}
Since $\tensor{f}{^A^B_C} = - \tensor{f}{^B^A_C}$ and $\delta ( \sigma - \sigma' ) = \delta ( \sigma' - \sigma )$, the first term on the right side of (\ref{JT_js_final}) suffers a sign flip under these replacements. Similarly, since
\begin{align}
    \partial_\sigma \delta ( \sigma - \sigma' ) = - \partial^{(\sigma')} \delta ( \sigma' - \sigma ) \, ,
\end{align}
by equation (\ref{delta_flip_sign}), while $\gamma^{AB} = \gamma^{BA}$, the second term on the right side of (\ref{JT_js_final}) is also multiplied by $-1$ under the replacements (\ref{simult_replacements}). We conclude that
\begin{align}
    \left\{ \fJ_\tau^A ( \sigma ) , j_\sigma^B ( \sigma' ) \right\} = - \left\{ \fJ_\tau^B ( \sigma ' ) , j_\sigma^A ( \sigma ) \right\} \, ,
\end{align}
and as a result of the general argument around equation (\ref{steps_to_swap_bracket}), it follows that
\begin{align}\label{app_swap_j_FJ}
    \left\{ j_{\sigma, 1} ( \sigma ) , \fJ_{\tau, 2} ( \sigma' ) \right\} = \left\{ \fJ_{\tau, 1} ( \sigma ) , j_{\sigma, 2} ( \sigma' ) \right\} \, .
\end{align}

\subsubsection*{Calculation of $\{ j_\sigma^A, j_\sigma^B \}$}

Finally, let us consider the bracket
\begin{align}
    \left\{ j_\sigma^A ( \sigma ) , j_\sigma^B ( \sigma' ) \right\} &= \left\{ j_\mu^A ( \phi ( \sigma ) ) \frac{\partial \phi^\mu ( \sigma ) }{\partial \sigma} , j_\nu^B ( \phi ( \sigma' ) ) \frac{\partial \phi^\nu ( \sigma' )}{\partial \sigma'} \right\} \, .
\end{align}
Both of the entries in the brackets are functions only of the canonical fields $\phi^\mu$ but not of the conjugate momenta $\pi_\mu$. That is,
\begin{align}\label{jsigma_no_pie}
    \frac{\delta}{\delta \pi^\rho ( \sigma'' ) } \left( j_\mu^A ( \phi ( \sigma ) ) \frac{\partial \phi^\mu ( \sigma ) }{\partial \sigma} \right) = 0 = \frac{\delta}{\delta \pi^\rho ( \sigma'' ) } \left( j_\nu^B ( \phi ( \sigma' ) ) \frac{\partial \phi^\nu ( \sigma' )}{\partial \sigma'} \right) \, .
\end{align}
Therefore, using the definition (\ref{bracket_defn}), we find that this bracket vanishes:
\begin{align}\label{js_js_bracket_final}
    \left\{ j_\sigma^A ( \sigma ) , j_\sigma^B ( \sigma' ) \right\} = 0 \, .
\end{align}

\subsection{Contractions with Generators}

We will now translate the results derived in section \ref{app:uncontracted_brackets} to statements about Poisson brackets for ``doubled'' tensor product quantities.

First let us contract the equation (\ref{final_Jt_Jt_bracket}) with $T_A \otimes T_B$ to find
\begin{align}\label{start_Jt_Jt_bracket_doubled}
    \{ \fJ_\tau^A ( \sigma ) , \fJ_\tau^B ( \sigma' ) \} T_A \otimes T_B = \tensor{f}{^A^B_C} T_A \otimes T_B \,\fJ_\tau^C \delta ( \sigma - \sigma' ) \, .
\end{align}
Using the bilinearity of the tensor product and the cyclic symmetry of $\tensor{f}{^A^B_C}$, we have
\begin{align}\label{structure_doubled_intermediate}
    \tensor{f}{^A^B_C} T_A \otimes T_B &= T_A \otimes \left( \tensor{f}{^A^B_C} T_B \right) \nonumber \\
    &= \gamma^{AD} \gamma^{BE} T_A \otimes \left( \tensor{f}{_D_E_C} T_B \right) \nonumber \\
    &= \gamma^{AD} \gamma^{BE} T_A \otimes \left( \tensor{f}{_C_D_E} T_B \right) \, \nonumber \\
    &= \gamma^{AD} T_A \otimes \left( \tensor{f}{_C_D^B} T_B \right) \, .
\end{align}
Recall that our conventions for the structure constants are $[ T_C, T_D ] = \tensor{f}{_C_D^B} T_B$, so
\begin{align}\label{structure_constants_tensor_product_doubled}
    \tensor{f}{^A^B_C} T_A \otimes T_B = T_A \otimes [ T_C, T^A ] \, .
\end{align}
Using (\ref{structure_constants_tensor_product_doubled}) on the right side of (\ref{start_Jt_Jt_bracket_doubled}), and simplifying the left side using (\ref{poisson_brackets_doubled_identity}), we have 
\begin{align}
    \{ \fJ_{\tau,1} ( \sigma ) , \fJ_{\tau,2} ( \sigma' ) \} &= \left( T_A \otimes [ T_C, T^A ] \right) \fJ_\tau^C \delta ( \sigma - \sigma' ) \nonumber \\
    &= \left( T_A \otimes [ \fJ_\tau^C  T_C, T^A ] \right) \delta ( \sigma - \sigma' ) \nonumber \\
    &= \left( T_A \otimes [ \fJ_\tau , T^A ] \right) \delta ( \sigma - \sigma' ) \, .
\end{align}
On the other hand, one trivially has
\begin{align}\label{trivial_commutator_identity}
    T_A \otimes [ \fJ_\tau , T^A ] &= T_A \otimes ( \fJ_\tau  T^A ) - T_A \otimes (  T^A \fJ_\tau  ) \nonumber \\
    &= \left( 1 \otimes \fJ_\tau \right) \left( T_A \otimes T^A \right) - \left( T_A \otimes T^A \right) \left( 1 \otimes \fJ_\tau \right) \nonumber \\
    &= [ \fJ_{\tau, 2} , C_{12} ] \, .
\end{align}
We conclude that
\begin{align}\label{Jt_Jt_contracted_final}
    \{ \fJ_{\tau,1} ( \sigma ) , \fJ_{\tau,2} ( \sigma' ) \} = [ \fJ_{\tau, 2} , C_{12} ]  \delta ( \sigma - \sigma' ) \, .
\end{align}
Next let us contract equation (\ref{JT_js_final}) with $T_A \otimes T_B$, which yields
\begin{align}\label{JT_js_doubled_start}
    \left\{ \fJ_\tau^A ( \sigma ) , j_\sigma^B ( \sigma' ) \right\} T_A \otimes T_B = \left( \tensor{f}{^A^B_C} j_\sigma^C \delta ( \sigma - \sigma' )   - \gamma^{AB} \partial_\sigma \delta ( \sigma - \sigma' ) \right) T_A \otimes T_B \, ,
\end{align}
where using (\ref{poisson_brackets_doubled_identity}) on the left and (\ref{structure_doubled_intermediate}) in the first term on the right gives
\begin{align}
    \left\{ \fJ_{\tau, 1} ( \sigma ) , j_{\sigma, 2} ( \sigma' ) \right\} &= \left( T_A \otimes [ T_C, T^A ] \right) j_\sigma^C \delta ( \sigma - \sigma' ) - C_{12} \partial_\sigma \delta ( \sigma - \sigma' ) \, ,
\end{align}
and then further simplifying using manipulations similar to (\ref{trivial_commutator_identity}), we conclude
\begin{align}\label{Jt_js_contracted}
    \left\{ \fJ_{\tau, 1} ( \sigma ) , j_{\sigma, 2} ( \sigma' ) \right\} &= [ j_{\sigma, 2}, C_{12} ] \delta ( \sigma - \sigma' ) - C_{12} \partial_\sigma \delta ( \sigma - \sigma' ) \, .
\end{align}
Finally, by contracting (\ref{js_js_bracket_final}) against $T_A \otimes T_B$, we have
\begin{align}\label{js_js_contracted}
    \left\{ j_{\sigma, 1} ( \sigma ) , j_{\sigma, 2} ( \sigma' ) \right\} = 0 \, .
\end{align}

\section{Matrix Identities for $M = 1 + \mathrm{ad}_\Lambda$}\label{app:MAB}

In our analysis of the T-dual auxiliary field sigma model (TD-AFSM) in section \ref{sec:t_duality}, we often need to use certain identities for the matrix
\begin{align}\label{M_app_defn}
    M_{AB} = \gamma_{AB} - \tensor{f}{_A_B^C} \Lambda_C \, ,
\end{align}
where $\Lambda = \Lambda^C T_C$ is a Lie algebra valued field which is the fundamental degree of freedom in the T-dual model. We remind the reader that capital Latin letters (like $A$, $B$, $C$) are used for Lie algebra indices, as (for example) when writing a basis $T_A$ of generators for $\mathfrak{g}$, and that our conventions are such that $\gamma_{AB} = \Tr ( T_A T_B )$ and $[ T_A, T_B ] = \tensor{f}{_A_B^C} T_C$.

It is convenient to alternately use both the definition (\ref{M_app_defn}) which carries explicit Lie algebra indices, as well as the index-free definition
\begin{align}\label{index_free}
    M = 1 + \mathrm{ad}_\Lambda \, .
\end{align}
To see that these definitions are equivalent, note for any Lie algebra valued quantity $X$, the index-free operator (\ref{index_free}) acts as 
\begin{align}
    M X = X + \mathrm{ad}_\Lambda X = X + [ \Lambda, X ] \, , 
\end{align}
and after expanding in a basis of generators $T_A$,
\begin{align}
    \left( M X \right)^A T_A &= X^A T_A + [ \Lambda^A T_A , X^B T_B ] \nonumber \\
    &= X^A T_A + \Lambda^A X^B \tensor{f}{_A_B^C} T_C \nonumber \\
    &= \left(  \gamma_{AB} - \tensor{f}{_A_B^C} \Lambda_C \right) X^B T^A \, .
\end{align}
As we will see, the required identities for the operator $M$ are in fact more general, and hold for any linear operator on $\mathfrak{g}$ of the form
\begin{align}\label{general_1_plus_A}
    M = 1 + \mathcal{A} \, ,
\end{align}
where $\mathcal{A}$ is antisymmetric, or in index notation,
\begin{align}
    M_{AB} = \gamma_{AB} + \mathcal{A}_{AB} \, , \qquad \mathcal{A}_{AB} = - \mathcal{A}_{BA} \, .
\end{align}
Indeed, we recover the case (\ref{index_free}) by taking $\mathcal{A}_{AB} = - \tensor{f}{_A_B^C} \Lambda_C$, or in index-free notation, $\mathcal{A} = \mathrm{ad}_\Lambda$. In this appendix, we will derive the required identities for this operator by working with the general expression (\ref{general_1_plus_A}) and using antisymmetry of $\mathcal{A}$.

Let us first set some more notation. We will also be interested in the transpose of the matrix $M$ in its Lie algebra indices, which can be written as
\begin{align}
    \left( M^T \right)_{AB} = \gamma_{AB} - \mathcal{A}_{AB} \, , 
\end{align}
due to the antisymmetry of $\mathcal{A}$. In index-free notation, we instead have
\begin{align}\label{MT_defn}
    M^T = 1 - \mathcal{A} \, .
\end{align}
The matrix inverse of $M_{AB}$ will be denoted $\left( M^{-1} \right)^{AB}$, and satisfies
\begin{align}
    \left( M^{-1} \right)^{AB} M_{BC} = \tensor{\delta}{^A_C} = M^{AB} \left( M^{-1} \right)_{BC} \, .
\end{align}
In index-free notation, we write this inverse as
\begin{align}
    M^{-1} = \frac{1}{1 + \mathcal{A}} \, ,
\end{align}
whose action on an element $X \in \mathfrak{g}$ is defined by the Taylor series expansion of $\frac{1}{1 + x}$,
\begin{align}\label{ad_taylor}
     \left( \frac{1}{1 + \mathcal{A}} \right) X = \sum_{n=0}^{\infty} ( - 1 )^n \mathcal{A}^n X \, .
\end{align}
It will also be useful to separate the Taylor series expansion (\ref{ad_taylor}) into even and odd parts:
\begin{align}\label{ad_taylor_even_odd}
     \frac{1}{1 + \mathcal{A}} &= \sum_{n = 0}^{\infty} \left( \mathcal{A} \right)^{2n} - \sum_{n = 0}^{\infty} \left( \mathcal{A} \right)^{2n+1} \nonumber \\
     &= \frac{1}{1 - \mathcal{A}^2} - \frac{\mathcal{A}}{1 - \mathcal{A}^2} \, .
\end{align}
Using the superscript ${}^{(S)}$ for the symmetric part of an operator and ${}^{(A)}$ for the antisymmetric part, we can then write
\begin{align}\label{minv_idxfree_defns}
    M^{-1} &= \left( M^{-1} \right)^{(S)} + \left( M^{-1} \right)^{(A)} \, , \nonumber \\
    \left( M^{-1} \right)^{(S)} &= \frac{1}{1 - \mathcal{A}^2} \, , \nonumber \\
    \left( M^{-1} \right)^{(A)} &= - \frac{\mathcal{A}}{1 - \mathcal{A}^2} \, .
\end{align}
In matrix notation, we write the components of $\left( M^{-1} \right)^{(S)}$ as $\left( M^{-1} \right)^{(AB)}$ and the components of $\left( M^{-1} \right)^{(A)}$ as $\left( M^{-1} \right)^{[AB]}$, as in our conventions
\begin{align}
    \left( M^{-1} \right)^{(AB)} &= \frac{1}{2} \left( \left( M^{-1} \right)^{AB} + \left( M^{-1} \right)^{BA} \right) \, , \nonumber \\
    \left( M^{-1} \right)^{[AB]} &= \frac{1}{2} \left( \left( M^{-1} \right)^{AB} - \left( M^{-1} \right)^{BA} \right) \, .
\end{align}
Let us now turn to the identities of interest. First consider the combination $\left( M^{- 1} \right)^T \cdot M^{-1}$. Since $\left( M^{- 1} \right)^T = \left( M^T \right)^{-1}$, and $M^T = 1 - \mathcal{A}$, this product is given by
\begin{align}\label{first_identity_product}
    \left( M^{- 1} \right)^T \cdot M^{-1} &= \frac{1}{1 - \mathcal{A}} \cdot \frac{1}{1 + \mathcal{A}} \nonumber \\
    &= \frac{1}{1 - \mathcal{A}^2} \, ,
\end{align}
which is the same as the expression for $\left( M^{-1} \right)^{(S)}$ in equation (\ref{minv_idxfree_defns}). Therefore
\begin{align}\label{first_identity_index_free}
    \left( M^{- 1} \right)^T \cdot M^{-1} = \left( M^{-1} \right)^{(S)} \, ,
\end{align}
or in index notation, that
\begin{align}
    \left( M^{-1} \right)^{CA} \tensor{\left( M^{-1} \right)}{_C^B} = \left( M^{-1} \right)^{(AB)} \, .
\end{align}
We also see that the order of the two operators in the product (\ref{first_identity_product}) is immaterial, since
\begin{align}
    \left( M^{- 1} \right)^T \cdot M^{-1} = \frac{1}{1 - \mathcal{A}} \cdot \frac{1}{1 + \mathcal{A}} =  \frac{1}{1 + \mathcal{A}} \cdot \frac{1}{1 - \mathcal{A}} = M^{-1} \cdot \left( M^{- 1} \right)^T \, ,
\end{align}
so in fact we have the stronger result
\begin{align}\label{second_identity_index_free}
    \left( M^{- 1} \right)^T \cdot M^{-1} = \left( M^{-1} \right)^{(S)} = M^{-1} \cdot \left( M^{- 1} \right)^T  \, ,
\end{align}
or with explicit indices,
\begin{align}\label{stronger_result_one}
    \left( M^{-1} \right)^{CA} \tensor{\left( M^{-1} \right)}{_C^B} = \left( M^{-1} \right)^{(AB)} = \tensor{\left( M^{-1} \right)}{^A_D} \left( M^{-1} \right)^{BD} \, .
\end{align}
Now let us consider the combination
\begin{align}\label{partial_fractions}
    \left( M^{-1} \right)^T \cdot M^{-1} \cdot \left( M^{-1} \right)^T &= \frac{1}{1 - \mathcal{A}} \cdot \frac{1}{1 + \mathcal{A}} \cdot \frac{1}{1 - \mathcal{A}} \nonumber \\
    &= \frac{1}{1 - \mathcal{A}^2} \cdot \frac{1}{1 - \mathcal{A}} \nonumber \\
    &= \frac{1}{2} \left( \frac{1}{1 + \mathcal{A}} + \frac{1}{1 - \mathcal{A}} \right) \cdot \frac{1}{1 - \mathcal{A}} \nonumber \\
    &= \frac{1}{2} \frac{1}{1 - \mathcal{A}^2} + \frac{1}{2} \left( \frac{1}{1 - \mathcal{A}} \right)^2 \nonumber \\
    &= \frac{1}{2} \left( M^{-1} \right)^{(S)} + \frac{1}{2} \left( M^T \right)^{-2} \, .
\end{align}
In the third line of (\ref{partial_fractions}) we have used a partial fractions decomposition, and in the final step we have recognized the definitions of $\left( M^{-1} \right)^{(S)}$ in  (\ref{minv_idxfree_defns})  and of $M^T$ in (\ref{MT_defn}).

Restoring indices, the manipulation (\ref{partial_fractions}) gives the identity
\begin{align}\label{triple_product}
    \left( M^{-1} \right)^{C A}  \left( M^{-1} \right)_{CD} \left( M^{-1} \right)^{BD} = \frac{1}{2} \left( M^{-1} \right)^{(AB)} + \frac{1}{2} \tensor{\left( M^{-1} \right)}{_C^A} \left( M^{-1} \right)^{B C} \, .
\end{align}
We will need a few more such relations. Consider
\begin{align}
    \left( M^{-1} \right)^{(S)} \cdot M^T \cdot M \cdot \left( M^{-1} \right)^{(S)} &= \frac{1}{1 - \mathcal{A}^2} \cdot \left( 1 - \mathcal{A} \right) \cdot \left( 1 + \mathcal{A} \right) \cdot \frac{1}{1 - \mathcal{A}^2} \nonumber \\
    &= \frac{1}{1 - \mathcal{A}^2} \nonumber \\
    &= \left( M^{-1} \right)^{(S)} \, ,
\end{align}
which proves the identity
\begin{align}\label{id_4M_sym_sym}
    \left( M^{-1} \right)_{(CB)}  M^{AC} \tensor{M}{_A^D} \left( M^{-1} \right)_{(DE)} &= \left( M^{-1} \right)_{(BE)} \, .
\end{align}
The analogue of this result for the antisymmetric part of $M^{-1}$ is
\begin{align}
    \left( M^{-1} \right)^{(A)} \cdot M^T \cdot M \cdot \left( M^{-1} \right)^{(A)} &= \frac{\mathcal{A}}{1 - \mathcal{A}^2} \cdot \left( 1 - \mathcal{A} \right) \cdot \left( 1 + \mathcal{A} \right) \cdot \frac{\mathcal{A}}{1 - \mathcal{A}^2} \nonumber \\
    &= \frac{\mathcal{A}^2}{1 - \mathcal{A}^2} \nonumber \\
    &= 1 - \frac{1}{1 - \mathcal{A}^2} \, \nonumber \\
    &= 1 - \left( M^{-1} \right)^{(S)} \, ,
\end{align}
which translates to the index-notation identity
\begin{align}\label{id_4m_asm_asm}
    \left( M^{-1} \right)_{[CB]} M^{AC} \tensor{M}{_A^D} \left( M^{-1} \right)_{[DE]} &= \gamma_{BE} - \left( M^{-1} \right)_{(BE)} \, .
\end{align}
We can also consider a product with one symmetric and one antisymmetric part of $M^{-1}$,
\begin{align}
    \left( M^{-1} \right)^{(A)} \cdot M^T \cdot M \cdot \left( M^{-1} \right)^{(S)} &= - \frac{\mathcal{A}}{1 - \mathcal{A}^2} \cdot \left( 1 - \mathcal{A} \right) \cdot \left( 1 + \mathcal{A} \right) \cdot \frac{1}{1 - \mathcal{A}^2} \nonumber \\
    &= - \frac{\mathcal{A}}{1 - \mathcal{A}^2} \nonumber \\
    &= \left( M^{-1} \right)^{(A)} \, ,
\end{align}
or with indices,
\begin{align}\label{id_4m_asym_sym}
    \left( M^{-1} \right)_{[BA]} M^{C A} \tensor{M}{_C^D} \left( M^{-1} \right)_{(DE)} = \left( M^{-1} \right)_{[BE]} \, .
\end{align}
We will conclude this appendix by deriving two mixed identities for products involving both $M^{-1}$ and $\mathcal{A}$. Note that
\begin{align}
    \left( M^{-1} \right)^{(S)} \cdot \mathcal{A} &= \frac{\mathcal{A}}{1 - \mathcal{A}^2} = - \left( M^{-1} \right)^{(A)} \, ,
\end{align}
or
\begin{align}\label{MA_id_one}
    \left( M^{-1} \right)_{(AE)} \tensor{\mathcal{A}}{^A_B} = - \left( M^{-1} \right)_{[EB]} \, ,
\end{align}
and similarly
\begin{align}
    \left( M^{-1} \right)^{(A)} \cdot \mathcal{A} &= - \frac{\mathcal{A}}{1 - \mathcal{A}^2} \cdot \mathcal{A} \nonumber \\
    &= - \frac{\mathcal{A}^2}{1 - \mathcal{A}^2} \nonumber \\
    &= 1 - \left( M^{-1} \right)^{(S)} \, , 
\end{align}
which proves the index notation identity
\begin{align}\label{MA_id_two}
    \left( M^{-1} \right)_{[EA]} \tensor{\mathcal{A}}{^A_B}  = \gamma_{EB}  - \left( M^{-1} \right)_{(EB)} \, .
\end{align}

\section{Details of T-Dual AFSM Calculations}\label{app:tdual_details}

In this appendix, we collect some intermediate steps in the calculations whose results are presented in section \ref{sec:t_duality}.

\subsubsection*{\ul{\it Removal of gauge fields from the master action}}

First we will show some intermediate steps in the derivation of the T-dual auxiliary field sigma model action (\ref{TD_AFSM}) from integrating out the gauge fields in the gauge-fixed master action (\ref{AFSM_master_GF_two}). Let us divide the calculation into steps, computing each of the terms in (\ref{AFSM_master_GF_two}) sequentially after substituting the solutions (\ref{omega_solns}) for the gauge fields. The first term is
\begin{align}\label{combine_one}
    \frac{1}{2} \omega_+^A \omega_{-, A} &= - \frac{1}{2} \left( M^{-1} \right)^{AB} \left( \partial_+ \Lambda_B - 2 v_{+, B} \right) \left( M^{-1} \right)_{CA} \left( \partial_- \Lambda^C + 2 v_{-}^C \right) \nonumber \\
    &= - \frac{1}{2} \tensor{\left( M^{-2} \right)}{_C^B} \left( \partial_+ \Lambda_B - 2 v_{+, B} \right) \left( \partial_- \Lambda^C + 2 v_-^C \right) \, .
\end{align}
Likewise, the mixed terms involving auxiliary field couplings are
\begin{align}\label{mixed_one}
     \omega_+^A v_{-, A} + \omega_-^A v_{+, A}  &= \left( M^{-1} \right)^{AB} \left( \partial_+ \Lambda_B - 2 v_{+, B} \right) v_{-, A} - \left( M^{-1} \right)^{BA} \left( \partial_- \Lambda_B + 2 v_{-, B} \right) v_{+, A} \, .
\end{align}
and
\begin{align}\label{mixed_two}
    \omega_+^A \partial_- \Lambda_A - \omega_-^A \partial_+ \Lambda_A &= \left( M^{-1} \right)^{AB} \left( \partial_+ \Lambda_B - 2 v_{+, B} \right) \partial_- \Lambda_A + \left( M^{-1} \right)^{BA} \left( \partial_- \Lambda_B + 2 v_{-, B} \right) \partial_+ \Lambda_A \, ,
\end{align}
after using the $\omega$ equation of motion. The structure constant term is
\begin{align}\label{combine_two}
    \tensor{f}{_A_B^C} \omega_+^A \omega_-^B \Lambda_C &= - \tensor{f}{_A_B^C} \left( M^{-1} \right)^{AD} \left( \partial_+ \Lambda_D - 2 v_{+, D} \right) \left( M^{-1} \right)^{EB} \left( \partial_- \Lambda_E + 2 v_{-, E} \right) \Lambda_C \nonumber \\
    &= - \left( M^{-1} \right)^{AD} f_{ABC} \left( M^{-1} \right)^{EB} \left( \partial_+ \Lambda_D - 2 v_{+, D} \right) \left( \partial_- \Lambda_E + 2 v_{-, E} \right) \Lambda^C \, .
\end{align}
We see that the results can be succinctly expressed in terms of the fields 
\begin{align}
    \mathfrak{K}_+^B = \left( M^{-1} \right)^{BA} \left( \partial_+ \Lambda_A - 2 v_{+, A} \right) \, , \qquad \mathfrak{K}_-^B = - \left( M^{-1} \right)^{AB} \left( \partial_- \Lambda_A + 2 v_{-, A} \right) \, ,
\end{align}
defined in (\ref{frak_K_defn}). Using this notation, we note that the term (\ref{combine_one}) combines with (\ref{combine_two}) as
\begin{align}\hspace{-10pt}\label{L_TDAFSM_with_K}
    &\frac{1}{2} \omega_+^A \omega_{-, A} + \frac{1}{2} \tensor{f}{_A_B^C} \omega_+^A \omega_-^B \Lambda_C \nonumber \\
    &\quad = - \frac{1}{2} \tensor{\left( M^{-2} \right)}{^C^B} \left( \partial_+ \Lambda_B - 2 v_{+, B} \right) \left( \partial_- \Lambda_C + 2 v_{-, C} \right) \nonumber \\
    &\qquad \qquad - \frac{1}{2} \left( M^{-1} \right)^{AD} f_{ABC} \left( M^{-1} \right)^{EB} \left( \partial_+ \Lambda_D - 2 v_{+, D} \right) \left( \partial_- \Lambda_E + 2 v_{-, E} \right) \Lambda^C  \nonumber \\
    &\quad = \frac{1}{2} \mathfrak{K}_+^C \mathfrak{K}_C + \frac{1}{2} f_{ABC} \Lambda^C \mathfrak{K}_+^A \mathfrak{K}_-^B \nonumber \\
    &\quad = \frac{1}{2} \left( \gamma_{AB} + f_{ABC} \Lambda^C \right) \mathfrak{K}_+^A \mathfrak{K}_-^B \nonumber \\
    &\quad = \frac{1}{2} \mathfrak{K}_+^A M_{BA} \mathfrak{K}_-^B \, .
\end{align}
Similarly, the mixed contributions (\ref{mixed_one}) and (\ref{mixed_two}) can also be expressed in terms of $\mathfrak{K}_{\pm}$. Assembling the various contributions which arise in this way, one finds
\begin{align}
    \mathcal{L}_{\text{TD-AFSM}} &= \frac{1}{2} \mathfrak{K}_+^A M_{BA} \mathfrak{K}_-^B + \mathfrak{K}_+^A \left( v_{-, A} + \frac{1}{2} \partial_- \Lambda_A \right) + \mathfrak{K}_-^A \left( v_{+, A}  - \frac{1}{2} \partial_+ \Lambda_A  \right) \nonumber \\
    &\qquad \qquad + v_+^A v_-^A + E ( \nu_2 , \ldots , \nu_N ) \, \nonumber \\
    &= \frac{1}{2} \mathfrak{K}_+^A M_{BA} \mathfrak{K}_-^B - \frac{1}{2} \mathfrak{K}_+^A M_{BA} \mathfrak{K}_{-}^B - \frac{1}{2} \mathfrak{K}_-^A M_{AB} \mathfrak{K}_{+}^{B} + v_+^A v_-^A + E ( \nu_2 , \ldots , \nu_N ) \nonumber \\
    &= - \frac{1}{2} \mathfrak{K}_+^A M_{BA} \mathfrak{K}_-^B + v_+^A v_-^A + E ( \nu_2 , \ldots , \nu_N ) \, .
\end{align}
Alternatively, if we expand out the $\mathfrak{K}_{\pm}$ using their definition, we get
\begin{align}\label{remove_gauge_fields_result}
    \mathcal{L}_{\text{TD-AFSM}} &= \frac{1}{2} \left( \partial_+ \Lambda_A - 2 v_{+, A} \right) \left( M^{-1} \right)^{BA} \left( \partial_- \Lambda_A + 2 v_{-, A} \right)  + v_+^A v_{-, A} + E ( \nu_2 , \ldots, \nu_N ) \nonumber \\
    &= \frac{1}{2} \partial_+ \Lambda^A \left( M^{-1} \right)^{BA} - 2 v_{+, A} \left( M^{-1} \right)^{BA} v_{-, A} - v_{+, A} \left( M^{-1} \right)^{BA} \partial_- \Lambda_A  \nonumber \\
    &\qquad + \partial_+ \Lambda_A \left( M^{-1} \right)^{BA} v_{-, A} + v_+^A v_{-,A} + E ( \nu_2 , \ldots, \nu_N ) \, .
\end{align}
This matches (\ref{TD_AFSM}), as desired.

\subsubsection*{\ul{\it Conservation of $k_{\pm}$ on-shell}}

We now demonstrate that, when the equations of motion for both $\Lambda_A$ and $v_\alpha$ in the T-dual AFSM are satisfied, the field $k_{\pm}$ defined in equation (\ref{kpmA_defn}) is conserved. 

To show this, it will be convenient to first invert the definition (\ref{frak_K_defn}) of $\mathfrak{K}_{\pm}^A$, which gives
\begin{align}\label{inverted_frak_K}
    M_{CB} \CC{K}_+^B &= \partial_+ \Lambda_C - 2 v_{+, C} \, , \nonumber \\
    M_{BC} \CC{K}_-^B &= - \partial_- \Lambda_C - 2 v_{-, C} \, .
\end{align}
\hspace{-3pt}\mbox{We differentiate the first line of (\ref{inverted_frak_K}) with respect to $\sigma^-$ and the second with respect to $\sigma^+$:}
\begin{align}\label{k_cons_intermediate}
    \left( \partial_- M_{CB} \right) \CC{K}_+^B + M_{CB} \partial_- \CC{K}_+^B &= \partial_+ \partial_- \Lambda_C - 2 \partial_- v_{+, C} \, , \nonumber \\
    \left( \partial_+ M_{BC} \right) \CC{K}_-^B + M_{BC} \partial_+ \CC{K}_-^B &= - \partial_+ \partial_- \Lambda_C - 2 \partial_+ v_{-, C} \, .
\end{align}
On the other hand,
\begin{align}
    \partial_- M_{AB} = - \tensor{f}{_A_B^C} \partial_- \Lambda_C \, , \qquad \partial_+ M_{BA} = - \tensor{f}{_B_A^C} \partial_+ \Lambda_C \, .
\end{align}
Using these equations in (\ref{k_cons_intermediate}) gives
\begin{align}\label{k_cons_intermediate_two}
    - \tensor{f}{_C_B^D} \partial_- \Lambda_D \CC{K}_+^B + M_{CB} \partial_- \CC{K}_+^B &= \partial_+ \partial_- \Lambda_C - 2 \partial_- v_{+, C} \, , \nonumber \\
    - \tensor{f}{_B_C^D} \partial_+ \Lambda_D \CC{K}_-^B + M_{BC} \partial_+ \CC{K}_-^B &= - \partial_+ \partial_- \Lambda_C - 2 \partial_+ v_{-, C} \, .
\end{align}
On the left side of each equation in (\ref{k_cons_intermediate_two}) we use $M_{AB} = \gamma_{AB} - \tensor{f}{_A_B^C} \Lambda_C$ to find
\begin{align}\label{k_cons_intermediate_three}
    \partial_- \CC{K}_{+, C} - \tensor{f}{_C_B^D} \Lambda_D \partial_- \CC{K}_+ &= \partial_+ \partial_- \Lambda_C - 2 \partial_- v_{+, C} + \tensor{f}{_C_B^D} \partial_- \Lambda_D \CC{K}_+^B  \, , \nonumber \\
    \partial_+ \CC{K}_{-, C} - \tensor{f}{_B_C^D} \Lambda_D \partial_+ \CC{K}_-^B &= - \partial_+ \partial_- \Lambda_C - 2 \partial_+ v_{-, C} + \tensor{f}{_B_C^D} \partial_+ \Lambda_D \CC{K}_-^B  \, .
\end{align}
Summing the two previous equations gives
\begin{align}
    \partial_- \left( \CC{K}_{+, C} + 2 v_{+, C} \right) + \partial_+ \left( \CC{K}_{-, C} + 2 v_{-, C} \right) &= \tensor{f}{_C_B^D} \Lambda_D \partial_- \CC{K}_+ + \tensor{f}{_B_C^D} \Lambda_D \partial_+ \CC{K}_-^B  + \tensor{f}{_C_B^D} \partial_- \Lambda_D \CC{K}_+^B \nonumber \\
    &\qquad + \tensor{f}{_B_C^D} \partial_+ \Lambda_D \CC{K}_-^B  \, . 
\end{align}
On the left side, we recognize the definition of $k_{\pm}$ in equation (\ref{kpmA_defn}). On the right, we replace the derivatives $\partial_{\pm} \Lambda_D$ using the inverted definitions (\ref{inverted_frak_K}). Doing this, one finds
\begin{align}\label{k_cons_intermediate_later}
    - \left( \partial_- k_{+, C} + \partial_+ k_{-, C} \right) &= \tensor{f}{_C_B^D} \Lambda_D \partial_- \CC{K}_+ + \tensor{f}{_B_C^D} \Lambda_D \partial_+ \CC{K}_-^B  - \tensor{f}{_C_B^D} \left( M_{ED} \CC{K}_-^E + 2 v_{-, D} \right) \CC{K}_+^B \nonumber \\
    &\qquad + \tensor{f}{_B_C^D} \left( M_{D E} \CC{K}_+^E + 2 v_{+, D} \right)  \CC{K}_-^B \,.
\end{align}
We now assume that the equation of motion (\ref{TD_frakK_flat}) for $\Lambda_A$ is satisfied, which means that we may simplify the first two terms using the flatness equation:
\begin{align}
    \tensor{f}{_C_B^D} \Lambda_D \partial_- \CC{K}_+ + \tensor{f}{_B_C^D} \Lambda_D \partial_+ \CC{K}_-^B  &= \tensor{f}{_B_C^D} \Lambda_D  \left( \partial_+ \CC{K}_-^B - \partial_- \CC{K}_+^B \right) \nonumber \\
    &= - \tensor{f}{_B_C^D} \Lambda_D  \tensor{f}{_E_F^B} \CC{K}_+^E \CC{K}_-^F \, .
\end{align}
Substituting this into (\ref{k_cons_intermediate_later}) gives
\begin{align}\label{k_cons_intermediate_later_two}
    \partial_- k_{+, C} + \partial_+ k_{-, C} &= \tensor{f}{_B_C^D} \Lambda_D  \tensor{f}{_E_F^B} \CC{K}_+^E \CC{K}_-^F + \tensor{f}{_C_B^D} \left( M_{ED} \CC{K}_-^E + 2 v_{-, D} \right) \CC{K}_+^B \nonumber \\
    &\qquad - \tensor{f}{_B_C^D} \left( M_{D E} \CC{K}_+^E + 2 v_{+, D} \right)  \CC{K}_-^B \, . 
\end{align}
Let us focus on the two terms on the right side which involve auxiliary fields. We note that
\begin{align}
    \tensor{f}{_C_B^D}  v_{-, D} \CC{K}_+^B - \tensor{f}{_B_C^D} 2 v_{+, D} \CC{K}_-^B &= \left( [ \CC{K}_+ , v_- ] + [ \CC{K}_- , v_+ ] \right)^C \nonumber \\
    &\deq 0 \, ,
\end{align}
where in the last step we have used the implication (\ref{tdual_fundamental_commutator}) of the auxiliary field equation of motion. Therefore, the terms containing auxiliary fields in (\ref{k_cons_intermediate_later_two}) cancel with one another. 

Next let us consider the terms in (\ref{k_cons_intermediate_later_two}) which involve $M_{ED}$ and $M_{DE}$. The expansion $M_{ED} = \gamma_{ED} + \tensor{f}{_E_D^F} \Lambda_F$ will generate two terms, one involving the Killing-Cartan form $\gamma$ and another involving structure constants. However, the terms involving $\gamma$ will drop out by symmetry. Indeed, these terms are
\begin{align}
    \tensor{f}{_C_B^D} \gamma_{ED} \CC{K}_-^E \CC{K}_+^B - \tensor{f}{_B_C^D} \gamma_{D E} \CC{K}_+^E\CC{K}_-^B = \tensor{f}{_C_B_D} \left( \CC{K}_{-}^{D} \CC{K}_+^B + \CC{K}_{+}^{D} \CC{K}_-^B \right) = 0 \, ,
\end{align}
since the combination in parentheses is symmetric in $B$ and $D$, while the structure constants are antisymmetric. Therefore, after removing both the $v_{\pm}$ terms and the $\gamma_{ED}$ terms, the only surviving contributions in (\ref{k_cons_intermediate_later_two}) are
\begin{align}\label{k_cons_intermediate_later_three}
    \partial_- k_{+, C} + \partial_+ k_{-, C} &= \tensor{f}{_B_C^D} \Lambda_D  \tensor{f}{_E_F^B} \CC{K}_+^E \CC{K}_-^F + \tensor{f}{_C_B^D} \tensor{f}{_D_E^F} \Lambda_F \CC{K}_-^E \CC{K}_+^B + \tensor{f}{_B_C^D} \tensor{f}{_D_E^F} \Lambda_F \CC{K}_+^E  \CC{K}_-^B \, \nonumber \\
    &= \tensor{f}{^D_B_C} \tensor{f}{_E_F^B} \left( \Lambda_D \CC{K}_+^E \CC{K}_-^F + \mathfrak{K}_{+, D} \mathfrak{K}_-^E \Lambda^F + \mathfrak{K}_{-, D} \Lambda^E \mathfrak{K}_+^F \right) \, ,
\end{align}
where we have used the symmetry of the structure constants and relabeled indices.

We now recognize this combination as the one associated with the Jacobi identity,
\begin{align}
    [ X , [ Y, Z ] ] + [ Y, [ Z, X ] ] + [ Z, [ X, Y ] ] = 0 \, ,
\end{align}
or in components,
\begin{align}
    \tensor{f}{_A_B^C} \tensor{f}{^B_D_E} \left( X^A  Y^D Z^E + Y^A  Z^D X^E + Z^A X^D Y^E \right) = 0 \, .
\end{align}
This allows us to conclude that
\begin{align}\label{k_cons_final}
    \partial_- k_{+, C} + \partial_+ k_{-, C} = 0 \, , 
\end{align}
when both the $\Lambda_A$ equation of motion and the auxiliary field equation of motion are satisfied.

\subsubsection*{\ul{\it Relations for momenta from type 1 generating functional}}

We now verify that the type 1 generating functional (\ref{TD_PCM_gen_F}) indeed leads to the relations (\ref{tdual_canonical}) for the canonical momenta in terms of the fields of the two models.

We recall that, for a type 1 generating functional, the momenta are given in terms of the functional derivatives listed in equation (\ref{type_1_field_theory_momenta}). The momentum $\Pi^A$ which is canonically conjugate to $\Lambda_A$, for instance, is
\begin{align}
    \Pi^A &= - \frac{\delta}{\delta \Lambda_A ( \sigma' )} \left( - \int d \sigma \, \Lambda_B ( \sigma ) j_\sigma^B ( \sigma ) \right) \nonumber \\
    &= \int d \sigma \, j_\sigma^A ( \sigma ) \delta ( \sigma - \sigma' ) \nonumber \\
    &= j_\sigma^A ( \sigma' ) \, ,
\end{align}
which can also be expressed using the pull-back map $\partial_\sigma \phi^\mu$ as
\begin{align}
    \Pi_A = j_\mu^A \partial_\sigma \phi^\mu \, .
\end{align}
This confirms the expression on the second line of (\ref{tdual_canonical}).

To find the corresponding result for $\pi_\mu$, we consider a small variation in the target-space coordinates $\phi^\mu ( \sigma, \tau ) $. Using the inverse vielbein $j^\mu_A$, we can convert between an infinitesimal fluctuation $\delta \phi^\mu$ of the target-space coordinates and a variation $\epsilon = \epsilon^A T_A \in \mathfrak{g}$ as
\begin{align}
    \delta \phi^\mu = j^\mu_A \epsilon^A \, .
\end{align}
Furthermore, we have already seen around equation (\ref{afsm_variation}) that, under such a fluctuation by $\epsilon$, the associated variation of $j_\sigma$ is
\begin{align}
    \delta j_\sigma = [ j_\sigma,  \epsilon ] + \partial_\sigma \epsilon \, .
\end{align}
Thus the variation of the generating functional $F$ is
\begin{align}
    \delta F &= - \int d \sigma \, \tr \left( \left( [ j_\sigma, \epsilon ] + \partial_\sigma \epsilon \right) \Lambda \right) \nonumber \\
    &= \int d \sigma \, \tr \left( \epsilon \left( \partial_\sigma \Lambda - [ \Lambda, j_\sigma ] \right) \right) \, ,
\end{align}
where we have integrated by parts and used the relation (\ref{cycle_trace_commutators}).

Writing the commutator using explicit Lie algebra indices and then converting back to target-space indices using the vielbein as $\epsilon^A = j^A_\mu \delta \phi^\mu$, this is
\begin{align}
    \delta F &= \int d \sigma \, \left( \epsilon^A \left( \partial_\sigma \Lambda^A - \tensor{f}{_B_C^A} \Lambda^B j_\sigma^C \right) \right) \nonumber \\
    &= \int d \sigma \, \left( \delta \phi^\mu j_\mu^A \left( \partial_\sigma \Lambda_A - \tensor{f}{_B_C_A} \Lambda^B j_\sigma^C \right) \right) \, .
\end{align}
We conclude that
\begin{align}
    \pi_\mu = \frac{\delta F}{\delta \phi^\mu} = j_\mu^A \partial_\sigma \Lambda_A - j_\mu^A \tensor{f}{_B_C_A} \Lambda^B j_\sigma^C \, ,
\end{align}
which can be written as
\begin{align}\label{final_F_to_pi_check}
    \pi_\mu = j_\mu^A \partial_\sigma \Lambda_A - \tensor{f}{_A_B^C} j_\nu^A j_\mu^B \Lambda_C \partial_\sigma \phi^\nu \, ,
\end{align}
in agreement with the first line of (\ref{tdual_canonical}).

\subsubsection*{\ul{\it Derivation of Hamiltonian for T-dual auxiliary field sigma model}}

Let us perform the Legendre transform of the Lagrangian for the T-dual auxiliary field sigma model, whose action was derived in equation (\ref{TD_AFSM}). We repeat this Lagrangian for convenience:
\begin{align}\label{td_afsm_lag_appendix}
    \mathcal{L}_{\text{TD-AFSM}} &= \frac{1}{2} ( \partial_- \Lambda_C )  \left( M^{-1} \right)^{CD} ( \partial_+ \Lambda_D ) - 2 v_{-, C} \left( M^{-1} \right)^{CD} v_{+, D}  \nonumber \\
    &\qquad + v_{-, C} \left( M^{-1} \right)^{CD} \partial_+ \Lambda_D  - \partial_- \Lambda_C \left( M^{-1} \right)^{CD}  v_{+, D} + v_+^A \gamma_{AB} v_-^B 
     \nonumber \\
     &\qquad
     + E ( \nu_2 , \ldots , \nu_N ) \, .
\end{align}
In our conventions for light-cone coordinates, given any vectors $u_{\alpha, C}$ and $w_{\alpha, C}$, one has
\begin{align}
    u_{-, C} \left( M^{-1} \right)^{CD} w_{+, D} &= \left( u_{\tau, C} - u_{\sigma, C} \right) \left( M^{-1} \right)^{CD} \left( w_{\tau, D} + w_{\sigma, D} \right) \nonumber \\
    &= u_{\tau, C} \left( M^{-1} \right)^{CD} w_{\tau, D} - u_{\sigma, C} \left( M^{-1} \right)^{CD} w_{\sigma, D} + u_{\tau, C} \left( M^{-1} \right)^{CD} w_{\sigma, D} \nonumber \\
    &\qquad - u_{\sigma, C} \left( M^{-1} \right)^{CD} w_{\tau, D} \, .
\end{align}
Using this to expand the Lagrangian (\ref{td_afsm_lag_appendix}) explicitly gives
\begin{align}\hspace{-10pt}
    \mathcal{L}_{\text{TD-AFSM}} &= \frac{1}{2} ( \partial_\tau \Lambda_C ) \left( M^{-1} \right)^{CD} ( \partial_\tau \Lambda_D ) - \frac{1}{2} ( \partial_\sigma \Lambda_C ) \left( M^{-1} \right)^{CD} ( \partial_\sigma \Lambda_D ) + ( \partial_\tau \Lambda_C ) \left( M^{-1} \right)^{[CD]} ( \partial_\sigma \Lambda_D ) \nonumber \\
    &\qquad - 2 v_{\tau, C} \left( M^{-1} \right)^{CD} v_{\tau, D} + 2 v_{\sigma, C} \left( M^{-1} \right)^{CD} v_{\sigma, D} - 4 v_{\tau, C} \left( M^{-1} \right)^{[CD]} v_{\sigma, D} \nonumber \\
    &\qquad + v_{\tau, C} \left( M^{-1} \right)^{CD} ( \partial_\tau \Lambda_D ) - v_{\sigma, C} \left( M^{-1} \right)^{CD} ( \partial_\sigma \Lambda_D ) + v_{\tau, C} \left( M^{-1} \right)^{CD} ( \partial_\sigma \Lambda_D ) \nonumber \\
    &\qquad \qquad \qquad - v_{\sigma, C} \left( M^{-1} \right)^{CD} ( \partial_\tau \Lambda_D )\nonumber \\
    &\qquad - ( \partial_\tau \Lambda_C ) \left( M^{-1} \right)^{CD} v_{\tau, D} + ( \partial_\sigma \Lambda_C ) \left( M^{-1} \right)^{CD} v_{\sigma, D} - ( \partial_\tau \Lambda_C ) \left( M^{-1} \right)^{CD} v_{\sigma, D} \nonumber \\
    &\qquad \qquad \qquad + ( \partial_\sigma \Lambda_C ) \left( M^{-1} \right)^{CD} v_{\tau, D} \nonumber \\
    &\qquad + v_\tau^A v_{\tau, A} - v_\sigma^A v_{\sigma, A} + E ( \nu_2 , \ldots , \nu_N ) \, .
\end{align}
The momentum $\Pi^C$ which is canonically conjugate to $\Lambda^C$ is
\begin{align}
    \Pi^C &= \frac{\partial \mathcal{L}}{\partial ( \partial_\tau \Lambda_C ) } \nonumber \\
    &= \left( M^{-1} \right)^{(CD)} ( \partial_\tau \Lambda_D ) + \left( M^{-1} \right)^{[CD]} ( \partial_\sigma \Lambda_D ) + ( v_{\tau, D} - v_{\sigma, D} ) \left( M^{-1} \right)^{DC}  \nonumber \\
    &\qquad \qquad - \left( M^{-1} \right)^{CD} ( v_{\tau, D} + v_{\sigma, D} )\, .
\end{align}
The Hamiltonian is then defined by the Legendre transform
\begin{align}
    &\mathcal{H}_{\text{TD-AFSM}} = \Pi_C ( \partial_\tau \Lambda^C ) - \mathcal{L} \nonumber \\
    &= ( \partial_\tau \Lambda_C ) \left( M^{-1} \right)^{(CD)} ( \partial_\tau \Lambda_D ) + ( \partial_\tau \Lambda_C ) \left( M^{-1} \right)^{[CD]} ( \partial_\sigma \Lambda_D )  + (  v_{\tau, C}  - v_{\sigma, C} ) \left( M^{-1} \right)^{CD} ( \partial_\tau \Lambda_D ) \nonumber \\
    &\quad - ( \partial_\tau \Lambda_C ) \left( M^{-1} \right)^{CD} ( v_{\tau, D}  +  v_{\sigma, D} )  - \Bigg[ \frac{1}{2} ( \partial_\tau \Lambda_C ) \left( M^{-1} \right)^{CD} ( \partial_\tau \Lambda_D ) - \frac{1}{2} ( \partial_\sigma \Lambda_C ) \left( M^{-1} \right)^{CD} ( \partial_\sigma \Lambda_D ) \nonumber \\
    &\quad + ( \partial_\tau \Lambda_C ) \left( M^{-1} \right)^{[CD]} ( \partial_\sigma \Lambda_D )  - 2 v_{\tau, C} \left( M^{-1} \right)^{CD} v_{\tau, D} + 2 v_{\sigma, C} \left( M^{-1} \right)^{CD} v_{\sigma, D}  \nonumber \\
    &\quad - 4 v_{\tau, C} \left( M^{-1} \right)^{[CD]} v_{\sigma, D} + v_{\tau, C} \left( M^{-1} \right)^{CD} ( \partial_\tau \Lambda_D )  - v_{\sigma, C} \left( M^{-1} \right)^{CD} ( \partial_\sigma \Lambda_D ) \nonumber \\
    &\quad \quad \quad + v_{\tau, C} \left( M^{-1} \right)^{CD} ( \partial_\sigma \Lambda_D )  -  v_{\sigma, C} \left( M^{-1} \right)^{CD} ( \partial_\tau \Lambda_D )  \nonumber \\
    &\quad - ( \partial_\tau \Lambda_C ) \left( M^{-1} \right)^{CD} v_{\tau, D}   + ( \partial_\sigma \Lambda_C ) \left( M^{-1} \right)^{CD} v_{\sigma, D} -  ( \partial_\tau \Lambda_C ) \left( M^{-1} \right)^{CD} v_{\sigma, D}  \nonumber \\
    &\quad \quad \quad + ( \partial_\sigma \Lambda_C ) \left( M^{-1} \right)^{CD} v_{\tau, D} + v_\tau^A v_{\tau, A} - v_\sigma^A v_{\sigma, A} + E ( \nu_2 , \ldots , \nu_N ) \Bigg] \, ,
\end{align}
and after simplifying, this becomes
\begin{align}\label{cancel_hamiltonian}
    \mathcal{H}_{\text{TD-AFSM}} &= \frac{1}{2} ( \partial_\tau \Lambda_C ) \left( M^{-1} \right)^{CD} ( \partial_\tau \Lambda_D ) + \frac{1}{2} ( \partial_\sigma \Lambda_D ) \left( M^{-1} \right)^{CD} ( \partial_\sigma \Lambda_D ) \nonumber \\
    &\qquad + 2 v_{\tau, C} \left( M^{-1} \right)^{CD} v_{\tau, D} - 2 v_{\sigma, C} \left( M^{-1} \right)^{CD} v_{\sigma, D} + 4 v_{\tau, C} \left( M^{-1} \right)^{[CD]} v_{\sigma, D} \nonumber \\
    &\qquad + ( v_{\sigma, C} - v_{\tau, C} ) \left( M^{-1} \right)^{CD} ( \partial_\sigma \Lambda_D ) - ( \partial_\sigma \Lambda_C ) \left( M^{-1} \right)^{CD} ( v_{\sigma, D} + v_{\tau, D} ) \nonumber \\
    &\qquad - v_{\tau}^A v_{\tau, A} + v_{\sigma}^A v_{\sigma, A} - E ( \nu_2 , \ldots , \nu_N ) \, ,
\end{align}
which can be rewritten as
\begin{align}\label{cancel_hamiltonian_rewritten}
    \mathcal{H}_{\text{TD-AFSM}} &= \frac{1}{2} ( \partial_\tau \Lambda_C ) \left( M^{-1} \right)^{CD} ( \partial_\tau \Lambda_D ) + \frac{1}{2} ( \partial_\sigma \Lambda_D ) \left( M^{-1} \right)^{CD} ( \partial_\sigma \Lambda_D ) \nonumber \\
    &\qquad + 2 v_{\tau, C} \left( M^{-1} \right)^{CD} v_{\tau, D} - 2 v_{\sigma, C} \left( M^{-1} \right)^{CD} v_{\sigma, D} + 4 v_{\tau, C} \left( M^{-1} \right)^{[CD]} v_{\sigma, D} \nonumber \\
    &\qquad + 2 ( v_{\sigma, C} ) \left( M^{-1} \right)^{[CD]} ( \partial_\sigma \Lambda_D ) - 2 ( v_{\tau, C} ) \left( M^{-1} \right)^{(CD)} ( \partial_\sigma \Lambda_D ) \nonumber \\
    &\qquad - v_{\tau}^A v_{\tau, A} + v_{\sigma}^A v_{\sigma, A} - E ( \nu_2 , \ldots , \nu_N ) \, .
\end{align}
Next, our task is to express everything in terms of the conjugate momentum $\Pi_A$. To do this, it is convenient to first rewrite
\begin{align}\hspace{-10pt}\label{tdpcm_pi}
    \Pi^C &= \left( M^{-1} \right)^{(CD)} ( \partial_\tau \Lambda_D ) + \left( M^{-1} \right)^{[CD]} ( \partial_\sigma \Lambda_D ) + ( v_{\tau, D} - v_{\sigma, D} ) \left( M^{-1} \right)^{DC}  - \left( M^{-1} \right)^{CD} ( v_{\tau, D} + v_{\sigma, D} ) \nonumber \\
    &= \left( M^{-1} \right)^{(CD)} ( \partial_\tau \Lambda_D ) + \left( M^{-1} \right)^{[CD]} ( \partial_\sigma \Lambda_D ) + 2 v_{\tau, D} \left( M^{-1} \right)^{[DC]} - 2 v_{\sigma, D} \left( M^{-1} \right)^{(CD)} \nonumber \\
    &= \left( M^{-1} \right)^{(CD)} \left( \partial_\tau \Lambda_D - 2 v_{\sigma, D} \right) + \left( M^{-1} \right)^{[CD]} \left( \partial_\sigma \Lambda_D - 2 v_{\tau, D} \right) \, .
\end{align}
We now claim that (\ref{cancel_hamiltonian_rewritten}) is equivalent to
\begin{align}\label{final_tdpcm_ham_app}
    \mathcal{H}_{\text{claim}} &= \frac{1}{2} \Pi_C M^{AC} \tensor{M}{_A^D} \Pi_D - \Pi_A \tensor{f}{^A_B_C} \Lambda^C \left( \partial_\sigma \Lambda^B - 2 v_\tau^B \right) + \frac{1}{2} ( \partial_\sigma \Lambda^A ) ( \partial_\sigma \Lambda_A )  \nonumber \\
    &\qquad 
    + v_\tau^A v_{\tau, A} + v_\sigma^A v_{\sigma, A}
    + 2 v_\sigma^A \Pi_A  - 2 v_\tau^A \partial_\sigma \Lambda_A - E ( \nu_2 , \ldots , \nu_N ) \, .
\end{align}
To show this, we substitute the expression (\ref{tdpcm_pi}) for $\Pi_C$ into (\ref{final_tdpcm_ham_app}) and confirm that the result reduces to (\ref{cancel_hamiltonian_rewritten}). Let us break the calculation into steps and compute each combination involving $\Pi$ in the Hamiltonian (\ref{final_tdpcm_ham_app}). The term quadratic in momenta is
\begin{align}\hspace{-10pt}
    \frac{1}{2} \Pi_C M^{AC} \tensor{M}{_A^D} \Pi_D &= \frac{1}{2} \left( \left( M^{-1} \right)_{(CB)} \left( \partial_\tau \Lambda^B - 2 v_{\sigma}^{B} \right) + \left( M^{-1} \right)_{[CB]} \left( \partial_\sigma \Lambda^B - 2 v_{\tau}^{B} \right) \right) \nonumber \\
    &\qquad \qquad M^{AC} \tensor{M}{_A^D} \left( \left( M^{-1} \right)_{(DE)} \left( \partial_\tau \Lambda^E - 2 v_{\sigma}^{E} \right) + \left( M^{-1} \right)_{[DE]} \left( \partial_\sigma \Lambda^E - 2 v_{\tau}^{E} \right) \right) \, \nonumber \\
    &= \frac{1}{2} \left( \partial_\tau \Lambda^B - 2 v_{\sigma}^{B} \right) \left( M^{-1} \right)_{(CB)}  M^{AC} \tensor{M}{_A^D} \left( M^{-1} \right)_{(DE)} \left( \partial_\tau \Lambda^E - 2 v_{\sigma}^{E} \right) \nonumber \\
    &\quad + \frac{1}{2} \left( \partial_\sigma \Lambda^B - 2 v_{\tau}^{B} \right) \left( M^{-1} \right)_{[CB]} M^{AC} \tensor{M}{_A^D} \left( M^{-1} \right)_{[DE]} \left( \partial_\sigma \Lambda^E - 2 v_{\tau}^{E} \right) \nonumber \\
    &\quad + \left( \partial_\tau \Lambda^B - 2 v_{\sigma}^{B} \right)  \left( M^{-1} \right)_{[BA]} M^{C A} \tensor{M}{_C^D} \left( M^{-1} \right)_{(DE)} \left( \partial_\sigma \Lambda^E - 2 v_{\tau}^{E} \right)  \, .
\end{align}
We use the identities (\ref{id_4M_sym_sym}), (\ref{id_4m_asm_asm}), and (\ref{id_4m_asym_sym}) to simplify the products of matrices, finding
\begin{align}\label{Pi_Pi_coupling_simp}
    \frac{1}{2} \Pi_C M^{AC} \tensor{M}{_A^D} \Pi_D &= \frac{1}{2} \left( \partial_\tau \Lambda^B - 2 v_{\sigma}^{B} \right) \left( M^{-1} \right)_{(BE)} \left( \partial_\tau \Lambda^E - 2 v_{\sigma}^{E} \right) \nonumber \\
    &\quad + \frac{1}{2} \left( \partial_\sigma \Lambda^B - 2 v_{\tau}^{B} \right) \left( \gamma_{BE} - \left( M^{-1} \right)_{(BE)} \right)  \left( \partial_\sigma \Lambda^E - 2 v_{\tau}^{E} \right) \nonumber \\
    &\quad + \left( \partial_\tau \Lambda^B - 2 v_{\sigma}^{B} \right)  \left( M^{-1} \right)_{[BE]} \left( \partial_\sigma \Lambda^E - 2 v_{\tau}^{E} \right)  \, .
\end{align}
Next, we must also compute
\begin{align}
    \Pi_A \tensor{f}{^A_B_C} \Lambda^C \left( \partial_\sigma \Lambda^B - 2 v_\tau^B \right) &= \left( \left( M^{-1} \right)_{(AE)} \left( \partial_\tau \Lambda^E - 2 v_{\sigma}^{E} \right) + \left( M^{-1} \right)_{[AE]} \left( \partial_\sigma \Lambda^E - 2 v_{\tau}^{E} \right) \right) \nonumber \\
    &\qquad \qquad \qquad \cdot \tensor{f}{^A_B_C} \Lambda^C \left( \partial_\sigma \Lambda^B - 2 v_\tau^B \right) \, ,
\end{align}
where we use the identities
\begin{align}
    \left( M^{-1} \right)_{(AE)} \tensor{f}{^A_B_C} \Lambda^C = \left( M^{-1} \right)_{[EB]} \, , \qquad \left( M^{-1} \right)_{[AE]} \tensor{f}{^A_B_C} \Lambda^C = \gamma_{EB}  - \left( M^{-1} \right)_{(EB)} \, ,
\end{align}
which follow from (\ref{MA_id_one}) and (\ref{MA_id_two}) with $\mathcal{A}_{AB} = - \tensor{f}{_A_B^C} \Lambda_C$. This gives
\begin{align}\label{Pi_f_coupling_simp}
    \Pi_A \tensor{f}{^A_B_C} \Lambda^C \left( \partial_\sigma \Lambda^B - 2 v_\tau^B \right) &= \left( \partial_\tau \Lambda^E - 2 v_{\sigma}^{E} \right) \left( M^{-1} \right)_{[EB]} 
 \left( \partial_\sigma \Lambda^B - 2 v_\tau^B \right) \nonumber \\
    &\qquad + \left( \partial_\sigma \Lambda^E - 2 v_{\tau}^{E} \right) \left(  \gamma_{EB} - \left( M^{-1} \right)_{(EB)}  \right) \left( \partial_\sigma \Lambda^B - 2 v_\tau^B \right) \, .
\end{align}
Finally, we need the expression
\begin{align}\label{Pi_v_coupling_simp}
    v_{\sigma,C} \Pi^C = v_{\sigma, C} \left( M^{-1} \right)^{(CD)} \left( \partial_\tau \Lambda_D - 2 v_{\sigma, D} \right) + v_{\sigma, C} \left( M^{-1} \right)^{[CD]} \left( \partial_\sigma \Lambda_D - 2 v_{\tau, D} \right) \, .
\end{align}
Substituting (\ref{Pi_Pi_coupling_simp}), (\ref{Pi_f_coupling_simp}), and (\ref{Pi_v_coupling_simp}) into (\ref{final_tdpcm_ham_app}) gives
\begin{align}\label{tdpcm_ham_computing}
    \mathcal{H}_{\text{claim}} &= \frac{1}{2} \left( \partial_\tau \Lambda^B - 2 v_{\sigma}^{B} \right) \left( M^{-1} \right)_{(BE)} \left( \partial_\tau \Lambda^E - 2 v_{\sigma}^{E} \right) \nonumber \\
    &\quad + \frac{1}{2} \left( \partial_\sigma \Lambda^B - 2 v_{\tau}^{B} \right) \left( \gamma_{BE} - \left( M^{-1} \right)_{(BE)} \right)  \left( \partial_\sigma \Lambda^E - 2 v_{\tau}^{E} \right)  \nonumber \\
    &\quad + \left( \partial_\tau \Lambda^B - 2 v_{\sigma}^{B} \right)  \left( M^{-1} \right)_{[BE]} \left( \partial_\sigma \Lambda^E - 2 v_{\tau}^{E} \right) \nonumber \\
    &\quad - \left( \partial_\tau \Lambda^E - 2 v_{\sigma}^{E} \right) \left( M^{-1} \right)_{[EB]} \left( \partial_\sigma \Lambda^B - 2 v_\tau^B \right) \nonumber \\
    &\quad -  \left( \partial_\sigma \Lambda^E - 2 v_{\tau}^{E} \right) \left( \gamma_{EB} - \left( M^{-1} \right)_{(EB)} \right) \left( \partial_\sigma \Lambda^B - 2 v_\tau^B \right) \nonumber \\
    &\quad + \frac{1}{2} ( \partial_\sigma \Lambda^A ) ( \partial_\sigma \Lambda_A ) + v_\tau^A v_{\tau, A} + v_\sigma^A v_{\sigma, A} + 2 v_{\sigma, C} \left( M^{-1} \right)^{(CD)} \left( \partial_\tau \Lambda_D - 2 v_{\sigma, D} \right) \nonumber \\
    &\qquad  + 2 v_{\sigma, C} \left( M^{-1} \right)^{[CD]} \left( \partial_\sigma \Lambda_D - 2 v_{\tau, D} \right)  - 2 v_\tau^A \partial_\sigma \Lambda_A - E ( \nu_2 , \ldots , \nu_N )\, ,
\end{align}
or after collecting terms and simplifying,
\begin{align}\label{tdpcm_ham_computing_two}
    \mathcal{H}_{\text{claim}} &= \frac{1}{2} \left( \partial_\tau \Lambda^B - 2 v_{\sigma}^{B} \right) \left( M^{-1} \right)_{(BE)} \left( \partial_\tau \Lambda^E - 2 v_{\sigma}^{E} \right) \nonumber \\
    &\quad - \frac{1}{2} \left( \partial_\sigma \Lambda^B - 2 v_{\tau}^{B} \right) \left( \gamma_{BE} - \left( M^{-1} \right)_{(BE)} \right)  \left( \partial_\sigma \Lambda^E - 2 v_{\tau}^{E} \right) \nonumber \\
    &\quad + \frac{1}{2} ( \partial_\sigma \Lambda^A ) ( \partial_\sigma \Lambda_A ) + v_\tau^A v_{\tau, A} + v_\sigma^A v_{\sigma, A} + 2 v_{\sigma, C} \left( M^{-1} \right)^{(CD)} \left( \partial_\tau \Lambda_D - 2 v_{\sigma, D} \right) \nonumber \\
    &\qquad  + 2 v_{\sigma, C} \left( M^{-1} \right)^{[CD]} \left( \partial_\sigma \Lambda_D - 2 v_{\tau, D} \right)  - 2 v_\tau^A \partial_\sigma \Lambda_A  - E ( \nu_2 , \ldots , \nu_N ) \, .
\end{align}
Expanding out various terms and performing more algebra, we obtain
\begin{align}\label{tdpcm_ham_computing_three}
    \mathcal{H}_{\text{claim}} &= \frac{1}{2} ( \partial_\tau \Lambda^B ) \left( M^{-1} \right)_{(BE)} ( \partial_\tau \Lambda^E )  - 2 v_\sigma^B \left( M^{-1} \right)_{(BE)} ( \partial_\tau \Lambda^E )  + 2 v_\sigma^B \left( M^{-1} \right)_{(BE)} v_\sigma^E  \nonumber \\
    &\quad  - \frac{1}{2}\partial_\sigma \Lambda^B \partial_\sigma \Lambda_B   + 2 \partial_\sigma \Lambda^B v_{\tau, B} - 2 v_\tau^B v_{\tau, B } + \frac{1}{2} ( \partial_\sigma \Lambda^B ) \left( M^{-1} \right)_{(BE)} ( \partial_\sigma \Lambda^E )  \nonumber \\
    &\quad - 2 v_\tau^B \left( M^{-1} \right)_{(BE)} ( \partial_\sigma \Lambda^E ) + 2 v_\tau^B \left( M^{-1} \right)_{(BE)} v_\tau^E + \frac{1}{2} ( \partial_\sigma \Lambda^A ) ( \partial_\sigma \Lambda_A ) \nonumber \\
    &\quad  + v_\tau^A v_{\tau, A}  + v_\sigma^A v_{\sigma, A} + 2 v_{\sigma, C} \left( M^{-1} \right)^{(CD)} \left( \partial_\tau \Lambda_D \right) - 4 v_{\sigma, C} \left( M^{-1} \right)^{(CD)} \left( v_{\sigma, D} \right) \nonumber \\
    &\quad  + 2 v_{\sigma, C} \left( M^{-1} \right)^{[CD]} \left( \partial_\sigma \Lambda_D  \right) - 4 v_{\sigma, C} \left( M^{-1} \right)^{[CD]} \left(  v_{\tau, D} \right)  - 2 v_\tau^A \partial_\sigma \Lambda_A - E ( \nu_2 , \ldots , \nu_N ) \, \nonumber \\
    &= \frac{1}{2} ( \partial_\tau \Lambda^B ) \left( M^{-1} \right)_{(BE)} ( \partial_\tau \Lambda^E ) - 2 v_\sigma^B \left( M^{-1} \right)_{(BE)} v_\sigma^E + \frac{1}{2} ( \partial_\sigma \Lambda^B ) \left( M^{-1} \right)_{(BE)} ( \partial_\sigma \Lambda^E )  \nonumber \\
    &\quad - 2 v_\tau^B \left( M^{-1} \right)_{(BE)} ( \partial_\sigma \Lambda^E ) + 2 v_\tau^B \left( M^{-1} \right)_{(BE)} v_\tau^E - v_\tau^A v_{\tau, A} + v_\sigma^A v_{\sigma, A}  \nonumber \\
    &\quad  + 2 v_{\sigma, C} \left( M^{-1} \right)^{[CD]} \left( \partial_\sigma \Lambda_D  \right) + 4 v_{\tau, C} \left( M^{-1} \right)^{[CD]} \left(  v_{\sigma, D} \right)  - E ( \nu_2 , \ldots , \nu_N ) \, ,
\end{align}
which is equivalent to
\begin{align}\label{cancel_hamiltonian_rewritten_again}
    \mathcal{H}_{\text{claim}} &= \frac{1}{2} ( \partial_\tau \Lambda_C ) \left( M^{-1} \right)^{CD} ( \partial_\tau \Lambda_D ) + \frac{1}{2} ( \partial_\sigma \Lambda_D ) \left( M^{-1} \right)^{CD} ( \partial_\sigma \Lambda_D ) \nonumber \\
    &\qquad + 2 v_{\tau, C} \left( M^{-1} \right)^{CD} v_{\tau, D} - 2 v_{\sigma, C} \left( M^{-1} \right)^{CD} v_{\sigma, D} + 4 v_{\tau, C} \left( M^{-1} \right)^{[CD]} v_{\sigma, D} \nonumber \\
    &\qquad + 2 ( v_{\sigma, C} ) \left( M^{-1} \right)^{[CD]} ( \partial_\sigma \Lambda_D ) - 2 ( v_{\tau, C} ) \left( M^{-1} \right)^{(CD)} ( \partial_\sigma \Lambda_D ) \nonumber \\
    &\qquad - v_{\tau}^A v_{\tau, A} + v_{\sigma}^A v_{\sigma, A} - E ( \nu_2 , \ldots , \nu_N ) \, .
\end{align}
This is the expression of equation (\ref{cancel_hamiltonian_rewritten}). We have therefore proven that $\mathcal{H}_{\text{claim}} = \mathcal{H}_{\text{TD-AFSM}}$, establishing that (\ref{final_tdpcm_ham_app}) is the correct Hamiltonian, written in terms of momenta.

\bibliographystyle{utphys}
\bibliography{master}

\end{document}